\documentclass[12pt,reqno]{amsart}

\usepackage[margin=1.0in]{geometry}
\usepackage{graphicx}
\usepackage{amsmath}
\usepackage{amsthm}
\usepackage{amsfonts}
\usepackage{amssymb}
\usepackage{euscript}
\usepackage{mathrsfs}
\usepackage{caption}
\usepackage{subcaption}
\usepackage{textcomp}
\usepackage{hyperref}
\usepackage{slashed}
\usepackage{mathtools}
\usepackage[normalem]{ulem}

\newtheorem*{proposition*}{Proposition}
\newtheorem*{theorem*}{Theorem}
\newtheorem*{conjecture*}{Conjecture}
\newtheorem*{claim*}{Claim}
\newtheorem*{lemma*}{Lemma}
\newtheorem*{corollary*}{Corollary}
\newtheorem{theorem}{Theorem}[section]
\newtheorem{proposition}[theorem]{Proposition}

\newtheorem{lemma}[theorem]{Lemma}
\newtheorem{corollary}[theorem]{Corollary}

\newtheorem*{definition*}{Definition}

\newtheorem*{assumption*}{Assumption}

\newtheorem*{remark*}{Remark}
\newtheorem{remark}{Remark}[section]

\newcommand{\R}{\mathbb{R}}
\newcommand{\s}{\mathbb{S}}

\newcommand{\N}{\mathbb{N}}

\newcommand{\scrM}{\mathcal{M}}
\newcommand{\scrQ}{\mathcal{Q}}

\DeclareMathOperator{\tr}{\textnormal{tr}}

\DeclareMathOperator{\arctanh}{\textnormal{arctanh}}

\newcommand{\snabla}{\slashed{\nabla}}

\numberwithin{equation}{section}
\allowdisplaybreaks

\begin{document}

\title[Linear waves on constant radius limits]{Linear waves on constant radius limits of cosmological black hole spacetimes}
\author{Dejan Gajic}
\address{University of Cambridge, Department of Applied Mathematics and Theoretical Physics, Wilberforce Road, Cambridge CB3 0WA, United Kingdom}
\email{D.Gajic@damtp.cam.ac.uk}

\begin{abstract}
In this paper we consider the Klein-Gordon equation on spherically symmetric background spacetimes with a constant area radius. The spacetimes under consideration are Nariai and Pleba\'nski-Hacyan, and can be considered constant radius limits of Reissner-Nordstr\"om-de Sitter spacetimes. We prove boundedness in the case of a non-negative Klein Gordon mass and decay unless the mass is zero. In the latter case we prove decay of solutions that are supported on all harmonic modes with angular momentum $l\geq 1$. We show that the $l=0$ modes of solutions to the massless Klein-Gordon equation do not decay. They are subject to conservation laws along degenerate Killing horizons. We apply the estimates in Nariai to give decay of solutions to the massive Klein-Gordon equation on an $n$-dimensional de Sitter background, using only the vector field method and with no restrictions on the positive Klein-Gordon mass.
\end{abstract}

\maketitle

\tableofcontents

\section{Introduction}

Understanding the asymptotics and global causal structures present in solutions $(\mathcal{M},g)$ to the vacuum Einstein equations 
\begin{equation}
\label{eq:einsteineqs}
\textnormal{Ric}(g)-\frac{1}{2}R g+\Lambda g=0,
\end{equation}
where $\Lambda$ is the \emph{cosmological constant}, is intimately connected to the question of stability, in the context of the Cauchy problem. In the $\Lambda=0$ setting, an expectation, dubbed the \emph{Final State conjecture}, is that every vacuum black hole solution eventually settles down to a member of the Kerr family.

A necessary precondition for this expectation to hold up is the stability of subextremal Kerr. This remains an open problem. See the lecture notes \cite{dafrod5} for more details.

The picture in the $\Lambda>0$ case is more intricate. The spherically symmetric vacuum setting already provides a wider variety of solutions. Unlike members of the Schwarzschild family, the analogous spherically symmetric Schwarzschild-de Sitter spacetimes (see Figure \ref{fig:schdes}) have an \emph{extremal} limit, ``extremal Schwarzschild-de Sitter", composed of an (infinite) sequence of collapsing and expanding spacetime regions. The metric in the collapsing regions can be constructed as the limit of a sequence of subextremal Schwarzschild-de Sitter interior regions, in which the difference between the event horizon and the cosmological horizon radii goes to zero. Analogously, an expanding region of extremal Schwarzschild-de Sitter can be constructed as a limit of the expanding regions of a sequence of subextremal Schwarzschild-de Sitter spacetimes, in which the difference between the radii of the horizons goes to zero. 
\begin{figure}
\begin{center}
\includegraphics[width=3.5in]{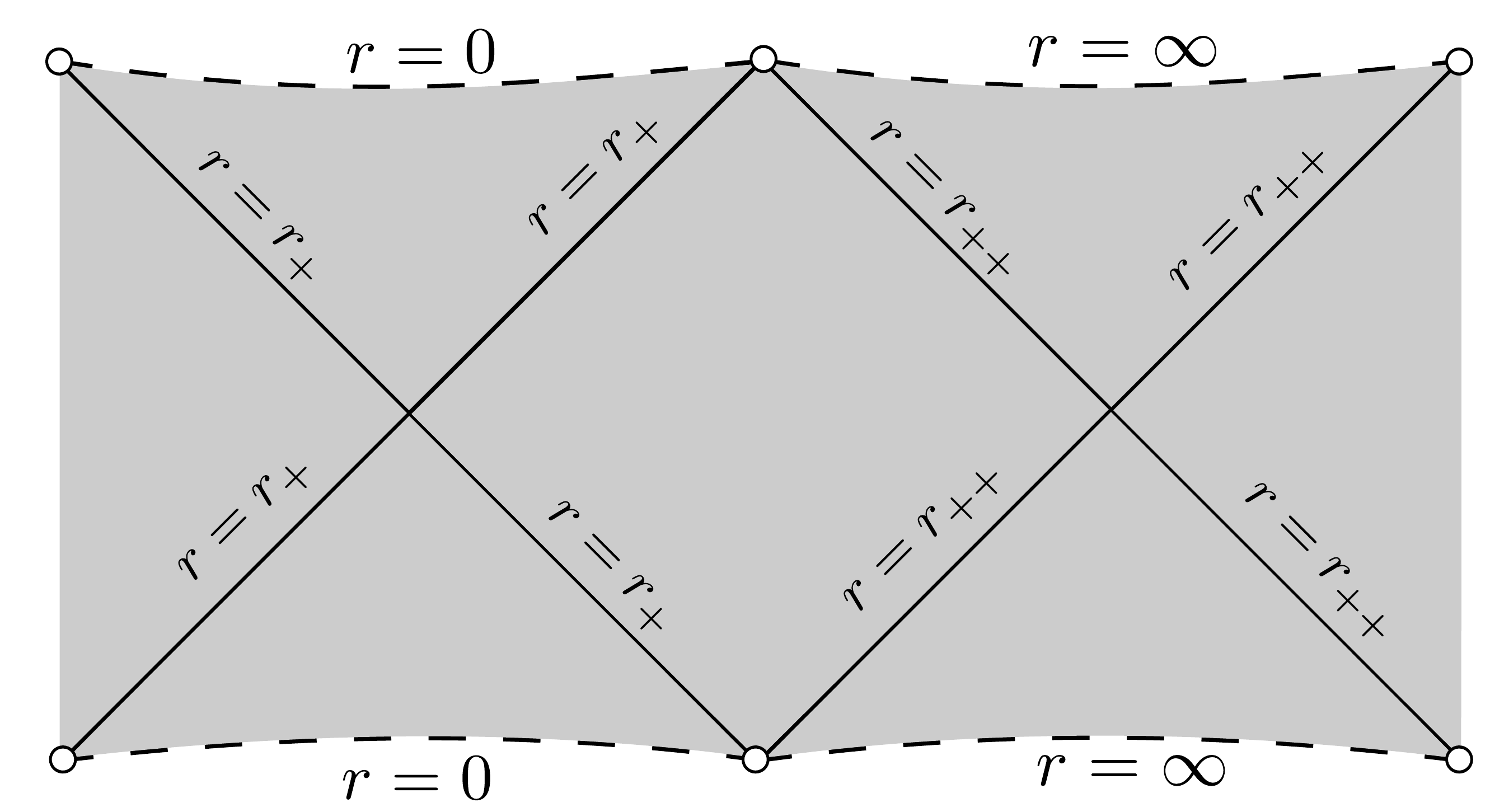}
\caption{The Penrose diagram of subextremal Schwarzschild-de Sitter, where each point in the diagram corresponds to a sphere of radius $r$. The radius of the event horizon is $r_+$ and the radius of the cosmological horizon is $r_{++}$.}
\label{fig:schdes}
\end{center}
\end{figure}
There exists another, lesser-known limit of subextremal Schwarzschild-de Sitter spacetimes, which we will call the \emph{constant radius limit}.

A Nariai spacetime is the maximal extension of a region that can be viewed as the limit of the regions \emph{between} the event and cosmological horizons of a sequence of subextremal Schwarzschild-de Sitter spacetimes, in which the difference between the horizon radii goes to zero. The Nariai radius is constant everywhere and equal to the radius of the event horizons of the corresponding extremal limit. Nariai is a homogeneous, geodesically complete spacetime that was first constructed in \cite{nar1,nar2}. Ginsparg-Perry showed in \cite{ginper1} that Nariai arises as a constant radius limit. 

If we add charge and consider the electrovacuum Einstein equations with $\Lambda >0$, we find additional constant radius limits of Reissner-Nordstr\"om-de Sitter: \emph{charged} Nariai, \emph{cosmological} Bertotti-Robinson and Pleba\'nski-Hacyan \cite{plebhac1}. Similarly, in the electrovacuum $\Lambda=0$ setting, the constant radius limit of Reissner-Nordstr\"om is the (regular) Bertotti-Robinson spacetime, first described in \cite{ber1,rob1}. We refer to all these spacetimes as \emph{constant radius spacetimes}. See Section \ref{sec:birkh} for a definition and an overview of their main properties.

Motivated by the strong cosmic censorship conjecture in the cosmological setting, Dafermos-Rendall considered the Einstein-Vlasov matter model with $\Lambda>0$, describing collisionless dust, in spherical, hyperbolic and planar symmetry in \cite{dafrend1}. They showed in particular that asymptotically extremal Schwarzschild-de Sitter black hole regions and asymptotically Nariai regions can form in the evolution of spherically symmetric Einstein-Vlasov data, even for expanding data that is arbitrarily close to homogeneity. Moreover, they proved that the formation of such regions constitutes the only possible obstruction to the validity of the strong cosmic censorship conjecture for this matter model. See Theorem 1.3 of \cite{dafrend1}.

The role of constant radius spacetimes in understanding the asymptotics and global structures of solutions to (\ref{eq:einsteineqs}) should therefore not be underestimated. A natural starting point is to address their stability properties.

The numerical study \cite{bey2} investigated spacetimes arising from perturbations of Nariai initial data, with topology $\s^1\times \s^2$, within the Gowdy symmetry class. Note however that the restriction of Gowdy symmetry sidesteps an important feature present in the dynamics under spherical symmetry, namely the formation of asymptotically Schwarzschild-de Sitter black holes.

Results pertaining to the nonlinear stability of solutions to (\ref{eq:einsteineqs}) were first obtained by Friedrich in \cite{fried1} for de Sitter, with the use of conformal methods. The first full nonlinear stability result in the $\Lambda=0$ case followed soon after in the monumental proof of global stability of the Minkowski spacetime by Christodoulou-Klainerman in \cite{chrklai1}. Their proof relies on quantitative estimates for connection and curvature quantities and motivates the study of quantitative decay of solutions to the linear wave equation
\begin{equation}
\label{eq:waveequation}
\square_g\psi=0,
\end{equation}
in a sufficiently robust setting. Indeed, (\ref{eq:waveequation}) on a fixed background should be viewed as a ``poor man's linearisation'' of the Einstein equations. It is in this spirit that we initiate a study of (\ref{eq:waveequation}), and more generally the Klein-Gordon equation,
\begin{equation}
\label{eq:kgequation}
(\square_g-\mu^2)\psi=0,
\end{equation}
with $\mu\in \R$, on fixed constant radius spacetime backgrounds.

\subsection{Previous results on the wave equation in black hole spacetimes}
\label{sec:prevres}

The study of (\ref{eq:waveequation}) on black hole backgrounds has been a field of active research for many years.

The first mode stability results on Schwarzschild and Kerr backgrounds were obtained in \cite{regwhe1, vish1, whi1}. Boundedness of the full solution to (\ref{eq:waveequation}) was obtained by Kay-Wald \cite{kaywald1} and in a more robust and complete setting by Dafermos-Rodnianski in \cite{dafrod2}, who introduced the celebrated red-shift vector field. More recently, decay properties have been proved in Schwarzschild and Kerr, see \cite{bluester1,dafrod3, dafrod4,luk1,marmettattoh1, toh1,dafrod5, dafrodsch1,andblue1}. The results in Schwarzschild can be extended to subextremal Reissner-Norstr\"om, see \cite{dafrod5} and \cite{bluesof1}.

Of particular interest is the case of an extremal black hole background. A mathematical analysis of (\ref{eq:waveequation}) on extremal black hole backgrounds was first performed by Aretakis in a series of papers \cite{are1,are2,are4,are5}. Aretakis established quantitative decay rates in time for $\psi$ and all its derivatives outside the event horizon, and moreover for $\psi$ and its tangential derivatives along the event horizon. Perhaps most surprisingly, non-decay was established for transversal derivatives $\partial_r\psi$ along the event horizon, together with blow-up of $\partial_r^k\psi$, for $k\geq2$. These instability properties have been extended to the Maxwell equations and linearised gravity by Lucietti-Reall in \cite{lure1} and to higher-dimensional black hole spacetimes by Murata in \cite{mu1}. Moreover, the instability was shown to also hold for initial data supported away from the event horizon in a numerical setting, by Lucietti-Murata-Reall-Tanahashi in \cite{lumure1} (this result was soon after proved in \cite{are5}).

A mathematical analysis of (\ref{eq:waveequation}) in the region between the event horizon and cosmological horizon of Schwarzschild-de Sitter was carried out by Bony-H\"afner \cite{bohaf1} in a scattering theory framework. See also results by Melrose-S\'a Barreto-Vasy \cite{mebava1}. They proved exponential decay in time of solutions $\psi-\psi_0$ to (\ref{eq:waveequation}) away from the event horizon, where $\psi_0$ is the zeroth harmonic mode. In \cite{dafrod6} Dafermos-Rodnianski included the event horizon in their analysis, using the vector field method. They established decay in time, faster than any given polynomial rate. Exponential decay in time was proved in slowly rotating Kerr-de Sitter by Dyatlov in \cite{dya1,dya3}, including the event horizon. An analysis of (\ref{eq:waveequation}) and (\ref{eq:kgequation}) in the expanding region of Schwarzschild-de Sitter and Kerr-de Sitter was performed by Schlue in \cite{schlue1}.

The Klein-Gordon equation with a non-positive mass,
\begin{equation}
\label{eq:kgnonposmass}
(\square_g+\mu^2)\psi=0,
\end{equation}
has been studied in the $\Lambda<0$ setting on asymptotically AdS black hole spacetimes in \cite{hol1,holsmu3,holwar1,holsmu4}. In particular, Holzegel \cite{hol1} and Holzegel-Warnick \cite{holwar1} obtained boundedness of solutions to (\ref{eq:kgnonposmass}) in Kerr-AdS and Holzegel-Smulevici \cite{holsmu3,holsmu4} established logarithmic decay in time, showing moreover that generic solutions cannot decay faster than logarithmically.

\subsection{Properties of constant radius spacetimes}
\label{sec:birkh}
An $n$-dimensional de Sitter spacetime is denoted by $dS_n$, and an $n$-dimensional anti-de Sitter spacetime is denoted by $AdS_n$. We define the constant radius spacetimes $(\mathcal{M},g)$ to be solutions of the electrovacuum Einstein equations with $\Lambda>0$ that are isometric to warped products $\mathcal{\mathcal{Q}}\times_r \s^2$, where the spherical area radius $r=r_0$ is constant and $\mathcal{Q}$ is a 2-dimensional Lorentzian manifold that is isometric to either $dS_2$, $AdS_2$ or $\R^{1+1}$ with a constant Gaussian curvature $K$.

Table \ref{tbl:cr} gives an overview of the geometry of $\mathcal{Q}$ corresponding to each combination of $\Lambda$ and charge parameter $e\in \R$.\footnote{The charge parameter $e\in \R$ determines uniquely the electromagnetic tensor in spherically symmetric electrovacuum spacetimes.} The values appearing in Table \ref{tbl:cr} are derived in Section \ref{sec:derivation}. Moreover, the constant curvature $K$ is given by
\begin{equation}
\label{eq:K}
K=r_0^{-2}(1-2e^2 r_0^{-2}).
\end{equation}

\begin {table}[h!]
\caption{The relation between $\mathcal{Q}$, $r_0$ and the parameters $\Lambda>0$ and $e\in \R$.}
\centering
\begin{tabular}{ l | l | l| c | c }
\label{tbl:cr}
$\Lambda$ & $e^2$ & $r_0^2$ & $\mathcal{Q}$ & Spacetime \\
\hline                     
$=0$ & $\neq 0$ & $=e^2$ & $AdS_2$ & Bertotti-Robinson\\
$>0$ & $< \frac{1}{2 \Lambda}$ & $=\frac{1}{2\Lambda}\left(1-\sqrt{1-4\Lambda e^2}\right)<2e^2$ & $AdS_2$ & cosmological\\
\vspace{1pt} & \vspace{1pt} & \vspace{1pt} & \vspace{1pt} & Bertotti-Robinson\\
$>0$ & $< \frac{1}{2 \Lambda}$ & $=\frac{1}{2\Lambda}\left(1-\sqrt{1-4\Lambda e^2}\right)>2e^2$ & $dS_2$ & charged Nariai\\
$>0$ & $< \frac{1}{2 \Lambda}$ & $=\frac{1}{2\Lambda}\left(1+\sqrt{1-4\Lambda e^2}\right)$ & $dS_2$ & charged Nariai\\
$>0$ & $=0$ &=$\frac{1}{\Lambda}$ & $dS_2$ & Nariai\\
$>0$ & $=\frac{1}{2 \Lambda}$ & $=\frac{1}{2 \Lambda}$ & $\R^{1+1}$ & Pleba\'nski-Hacyan
\end{tabular}
\end{table}

All constant radius spacetimes are geodesically complete and homogeneous. Their causal structures are represented in Figure \ref{fig:cr}. The Penrose diagrams of (cosmological) Bertotti-Robinson, (charged) Nariai, and Pleba\'nski-Hacyan are identical to the Penrose diagrams of $AdS_n$, $dS_n$ and $\R^{n+1}$, respectively, where $n=2$. These follow from constructing conformal compactifications of $AdS_2$, $dS_2$ and $\R^{1+1}$, and embedding them in a compact region of $\R^2$. See \cite{hael} for a construction in the $n=4$ case.\footnote{When $n>2$ the conformal compactifications of $AdS_n$, $dS_n$ and $\R^{n+1}$ are instead embedded in $\R\times \s^{n-1}$.} 

 \begin{figure}[h!]
\centering
\begin{subfigure}[b]{0.3\textwidth}
\centering
\includegraphics[width=1in]{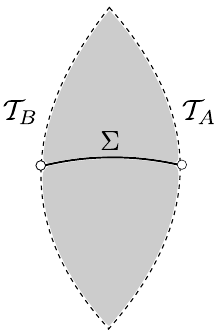}
\caption{(Cosmological) Bertotti-Robinson}
\end{subfigure}
\begin{subfigure}[b]{0.3\textwidth}
\includegraphics[width=\textwidth]{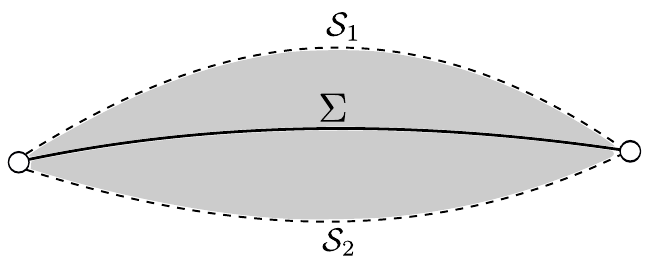}
\caption{(Charged) Nariai}
\end{subfigure}
\begin{subfigure}[b]{0.3\textwidth}
\centering
\includegraphics[width=\textwidth]{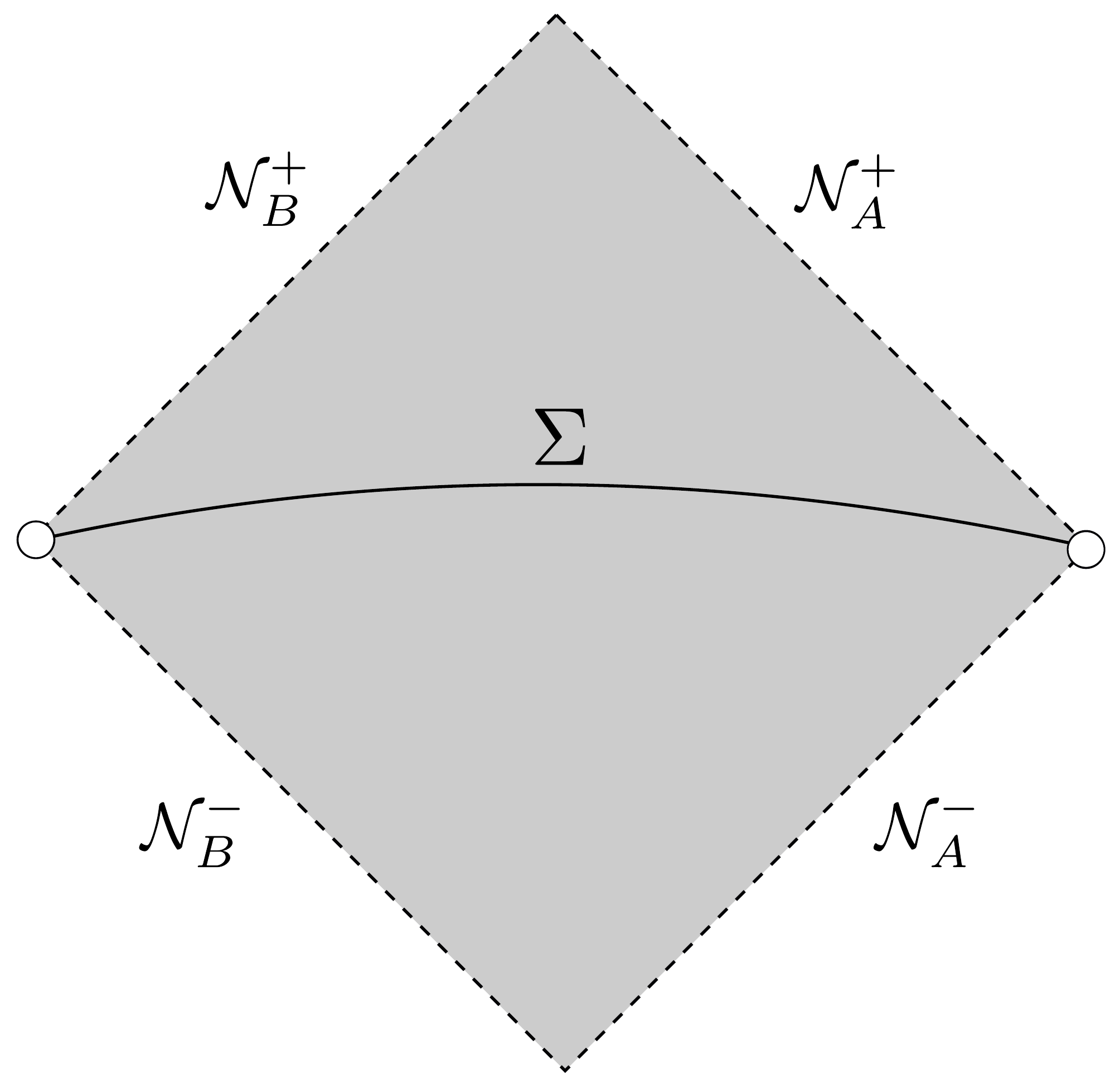}
\caption{Pleba\'nski-Hacyan}
\end{subfigure}
\caption{The Penrose diagrams of the electrovacuum spherically symmetric constant radius spacetimes. Each point on the diagram represents a sphere of radius $r$, where $r$ has the \emph{same} value everywhere. The constant radius spacetimes are geodesically complete and the dashed lines represent points in the spacetimes that are reached by geodesics in infinite affine time.}
\label{fig:cr}
\end{figure}

In global coordinates $(\tilde{t},\tilde{x})$ on $dS_2$, with $\tilde{t}\in \R$, $\tilde{x}\in \R$, and $(\theta,\phi)$ the standard spherical polar coordinates on $\s^2$, the metric on (charged) Nariai takes the form:
\begin{equation*}
g_{\textnormal{Nariai}}=K^{-1}(-d\tilde{t}^2+\cosh^2\tilde{t}d\tilde{x}^2)+r_0^2(d\theta^2+\sin^2\theta d\phi^2).
\end{equation*}
Here, $r_0^2(e,\Lambda)$ takes on the values given in Table \ref{tbl:cr} and $K(e,\Lambda)$ is given by (\ref{eq:K}).

Similarly, there exists global coordinates $(\tilde{t},\tilde{x})$ on $AdS_2$, with $\tilde{t}\in \R$, $\tilde{x}\in \R$, such that the metric on (cosmological) Bertotti-Robinson is given by,
\begin{equation*}
g_{\textnormal{Bertotti-Robinson}}=|K|^{-1}(-\cosh^2\tilde{x} d\tilde{t}^2+d\tilde{x}^2)+r_0^2(d\theta^2+\sin^2\theta d\phi^2).
\end{equation*}

Finally, the metric on Pleba\'nski-Hacyan is simply given by
\begin{equation*}
g_{\textnormal{Pleba\'nski-Hacyan}}=-dt^2+dx^2+r_0^2(d\theta^2+\sin^2\theta d\phi^2),
\end{equation*}
where $(t,x)$ are global coordinates on $\R^{1+1}$, with $t\in\R$ and $x\in \R$.

Constant radius spacetimes are often forgotten in the literature, when appealing to Birkhoff's Theorem for spherically symmetric electrovacuum spacetimes with $\Lambda\geq 0$. Birkhoff's Theorem is stated as follows:
\begin{theorem}[Birkhoff's Theorem]
A spherically symmetric solution $(\mathcal{M},g)$ of the electrovacuum Einstein equations with a non-negative cosmological constant $\Lambda$ is locally isometric to a Reissner-Nordstr\"om-de Sitter solution with parameters $\Lambda\geq 0$ and $e,M\in \R$, \textbf{or a constant radius spacetime with parameters $\Lambda\geq 0$ and $e\in \R$}.
\end{theorem}

In Section \ref{sec:kvf} we will show that the constant radius spacetimes share the property that each sphere of constant radius forms the intersection of an ingoing and an outgoing degenerate Killing horizon, with respect to a suitable Killing vector field. The degenerate Killing horizons are isometric to the horizons of either extremal Reissner-Nordstr\"om, extremal Schwarzschild-de Sitter or extremal Reissner-Nordstr\"om-de Sitter. Constant radius spacetimes can be viewed in particular as examples of \emph{near-horizon spacetimes of extremal black holes}, which all contain two intersecting degenerate Killing horizons. For an extensive review on near-horizon spacetimes, see \cite{kulu1} and the references therein.
 
 \subsection{Overview of results and techniques}
 \label{sec:overviewresults}
Since (cosmological) Bertotti-Robinson is not globally hyperbolic, the Cauchy problem for (\ref{eq:kgequation}) is not well-posed without additional boundary conditions. (Cosmological) Bertotti-Robinson spacetimes will therefore not be studied in this paper. 

We will investigate the asymptotic behaviour of solutions to (\ref{eq:kgequation}) on Nariai (from now on we include charged Nariai when we refer to Nariai) and Pleba\'nski-Hacyan (PH). This section provides a brief overview of the results proved in this paper. The main theorems are formulated in Section \ref{sec:maintheorems}.

\subsubsection{The Klein-Gordon equation on $dS_2$ and $\R^{1+1}$}
We can decompose solutions $\psi$ of (\ref{eq:kgequation}) into harmonic modes $\psi=\sum_{l=0}^{\infty}\psi_l$ satisfying the 1+1-dimensional Klein-Gordon equation
\begin{equation}
\label{eq:kqequationfixedl}
(\square_{g_{\mathcal{Q}}}-\mu^2-m_l^2)\psi_l=0,
\end{equation}
where $\square_{g_{\mathcal{Q}}}$ is the wave operator corresponding to the induced metric $g_{\mathcal{Q}}$ on $\mathcal{Q}$, where $\mathcal{Q}=dS_2$ in the Nariai case and $\mathcal{Q}=\R^{1+1}$ in the PH case. The mass parameter $\mu^2+m_l^2$ increases with $l$ and $m_0=0$. See Section \ref{sec:sphericalharmonics} for more details. In order to obtain uniform boundedness and decay of the full solution $\psi$, we will in particular determine the precise dependence on $l$ in the estimates for $\psi_l$.

\subsubsection{Boundedness and decay in Nariai: a local method}
In Nariai we can either make use of a \emph{local} or a \emph{global} red-shift effect to establish boundedness and decay. 

In Sections \ref{sec:boundn} and Section \ref{sec:decn} we restrict to a causally independent \emph{static region} in Nariai, which is bounded by two cosmological horizons. Since Nariai is homogeneous, there is no loss of generality in doing so. In Section \ref{sec:kvf} we show that there exists a timelike Killing vector field, which we call $T$, in the interior of the static region. The cosmological horizons are non-degenerate Killing horizons with respect to $T$. We foliate the static region with compact spacelike hypersurfaces by the isometric flow along $T$ of an initial spacelike hypersurface. The energy flux of the current $J^T$ through the spacelike hypersurfaces is non-negative definite and conserved, but it degenerates at the horizons. In order to obtain uniform boundedness of a non-degenerate energy, we appeal to a theorem of Dafermos-Rodnianski \cite{dafrod5} that ensures the existence of a \emph{local} red-shift vector field $N$ corresponding to any non-degenerate Killing horizon in a stationary spacetime. The energy flux of the current $J^N$ is positive definite, non-degenerate and uniformly bounded near the horizon. Pointwise boundedness estimates then follow by commuting with $T$ and applying standard elliptic and Sobolev estimates.

In order to prove uniform energy decay in the static region in Nariai, we first prove integrated local energy decay. We construct a suitable Morawetz vector field, which we call $V$, such that spacetime integrals of the divergence of the current $J^V$ control spacetime integrals of derivatives of $\psi$, after a slight modification. Integrated local energy decay is a powerful and robust tool for proving energy decay for solutions to wave equations.

The main obstruction to integrated local energy decay in Nariai is the \emph{trapping} of null geodesics. Each sphere foliating the Nariai spacetime contains a trapped null geodesic. Whenever trapped null geodesics are present in a spacetime, derivatives are lost in energy estimates, see \cite{rals} and \cite{sbier1}. Fortunately, though trapping is present at each sphere in Nariai, it can be considered unstable, as shown in Section \ref{sec:trapping}. Trapping in Nariai results in the loss of an angular derivative in the integrated local energy decay statement and in the final energy decay statement. 

Trapping is a high-frequency effect, so if we restrict $\psi$ to a single harmonic mode $\psi_l$ with $l\geq 1$, it does not form an obstruction and we are able to prove a stronger energy decay statement.

We make further use of the local red-shift effect to control integrated energies near the cosmological horizons.

The advantage of restricting to a static region in Nariai and proving integrated local energy decay is that the method does not involve the geometry outside the static region.

\subsubsection{Boundedness and decay in Nariai: a global method}

In Sections \ref{sec:globalboundnariai} and \ref{sec:globaldecaynariai} we work in global coordinates on $dS_2$, introduced in Section \ref{sec:narstco}, to construct a global vector field $\tilde{N}$. The energy flux of the corresponding energy current $J^{\tilde{N}}$ through a constant $\tilde{t}$ slice, where $\tilde{t}$ is a standard global time coordinate in $dS_2$, is controlled by an initial energy. We make use of the exponential expansion in $\tilde{t}$ of the volume form corresponding to constant $\tilde{t}$ slices to bound the $L^2$ norm of $\psi$ at a constant $\tilde{t}$ by the initial $\tilde{N}$-energy. By commuting with the spacelike Killing vector field present in $dS_2$, and with the Killing vector fields generating spherical symmetry we are able to obtain pointwise boundedness of $\psi$ from standard Sobolev estimates. 

Moreover, after a modification of the energy current $J^{\tilde{N}}$, we prove exponential decay in $\tilde{t}$ of the $\tilde{N}$-energy, resulting in pointwise exponential decay in $\tilde{t}$ of $\psi-\psi_0$ if $\mu=0$ and of the full solution $\psi$ if $\mu\neq 0$.

\subsubsection{Decay in $n$-dimensional de Sitter space}
\label{sec:subsubdsn}
The local method applied to fixed modes $\psi=\psi_0$ provides in particular a proof, using only vector field methods, of pointwise decay of solutions to the $\mu\neq 0$ Klein-Gordon equation (\ref{eq:kqequationfixedl}) in the static region of $dS_2$. Similarly, the global method applied to $\psi=\psi_0$ also gives pointwise decay.

We generalise the global method for decay of $\psi=\psi_0$ to higher dimensions to obtain pointwise time decay of solutions to (\ref{eq:kgequation}) with $\mu\neq 0$ on $n$-dimensional de Sitter space, with $n\geq 2$. De Sitter space $dS_n$ is an $n$-dimensional submanifold of $\R^{n+1}$ of constant positive curvature. It is diffeomorphic to $\R\times \s^{n-1}$ and can be globally foliated by Cauchy hypersurfaces that are diffeomorphic to $\s^{n-1}$.

Boundedness and decay results in $dS_4$ for solutions to (\ref{eq:kgequation}), with the restriction $\mu=0$ or $\mu^2\geq\frac{2\Lambda}{3}$, were proved by Schlue in his thesis \cite{schluethesis} using energy methods, by considering the static region and expanding region in $dS_4$ separately. Schlue also made use of the presence of a local red-shift effect near cosmological horizons and a global red-shift effect due to expansion.

Exponential decay for smooth solutions to (\ref{eq:kgequation}) with $\mu\neq 0$, was established in a general class of asymptotically de Sitter Lorentzian manifolds by Vasy in \cite{vas1}, using scattering theory.

In the literature one often encounters a restriction to the open subset $dS_{n,flat}$ of $dS_n$, that can be covered by a single chart $(t',x_1',\ldots,x'_{n-1})\in \R\times \R^{n-1}$, with $n\geq 2$. This is referred to as the ``flat slicing'' of de Sitter space. The Klein-Gordon equation (\ref{eq:kgequation}) in $dS_{n,flat}$ is equivalent to the following equation for the rescaled function $u=e^{\frac{n-1}{2}\sqrt{\frac{\Lambda}{3}}t}\psi$:
\begin{equation*}
\left(-\partial_{t'}^2+e^{-2\sqrt{\frac{\Lambda}{3}}t}\sum_{i=1}^{n-1}\partial_{x'_i}^2-\left(\mu^2-\frac{(n-1)^2\Lambda}{12}\right)\right)u=0.
\end{equation*}

In \cite{yag1} Yagdjian constructed fundamental solutions for (\ref{eq:kgequation}) in $dS_{n,flat}$ to show in particular exponential decay of $||\psi||_{H^s(\R^{n-1})}(t')$ in $t'$, with $s>\frac{n-1}{2}$, when the Klein-Gordon mass $\mu^2$ satisfies the bounds $\mu^2\in (0,\frac{((n-1)^2-1)\Lambda}{12})\cup[\frac{(n-1)^2\Lambda}{12},\infty)$.\footnote{The quantity $\mu^2-\frac{(n-1)^2\Lambda}{12}$ is treated as an ``effective mass''.} Similar results for (\ref{eq:kgequation}) with $\mu^2>\frac{(n-1)^2\Lambda}{12}$ were obtained in a general class of asymptotically de Sitter spacetimes by Baskin \cite{bas1}.

Note that the generalisation of the global method for decay in Nariai gives exponential decay of $||\psi||_{H^s(\s^{n-1})}(\tilde{t})$ in $\tilde{t}$ in the \emph{full} spacetime $dS_n$, or $||\psi||_{H^s(\R^{n-1})}(t')$ in $t'$ in the subset $dS_{n,flat}$, with $s>\frac{n-1}{2}$ and \emph{only} the restriction $\mu\neq 0$.

\subsubsection{Boundedness and decay in Pleba\'nski-Hacyan}

In PH there exists a globally timelike Killing vector field that is simply the generator of time translations in the 2-dimensional quotient spacetime, $\mathcal{Q}=\R^{1+1}$. We denote this vector field by $T$. We immediately obtain uniform boundedness of the non-degenerate energy flux of the current $J^T$ through a foliation of isometric spacelike hypersurfaces.

Although there is a globally timelike Killing vector field present in PH, we need to adopt a different strategy compared to Nariai to obtain pointwise decay of $\psi$. We can prove integrated local energy decay, with the loss of an angular derivative due to the presence of unstable trapping at each sphere foliating PH, similar to the Nariai case. However, the robust method introduced in \cite{dafrod7} for obtaining energy decay from integrated local energy decay and energy boundedness, even in the case of the massless equation (\ref{eq:waveequation}), fails because PH is not asymptotically flat. In fact, we show that the energy flux of the current $J^T$ through the future null boundaries $\mathcal{N}_A^+$ and $\mathcal{N}_B^+$ of PH vanishes, where we consider solutions $\psi-\psi_0$ in the $\mu=0$ case. This implies that one cannot obtain energy decay by foliating PH with hypersurfaces asymptoting at $\mathcal{N}^+_A$ and $\mathcal{N}_B^+$.

In \cite{klai1} Klainerman established pointwise decay of solutions to (\ref{eq:kgequation}) with $\mu\neq 0$ in $\R^{n+1}$ by using only vector field methods. In particular, if we take $n=1$ these results extend to the fixed modes $\psi_l$ in PH, with $l\geq 1$, by (\ref{eq:kqequationfixedl}). The arguments in \cite{klai1} rely fundamentally on the presence of a timelike Killing vector field and a Killing vector field that generates boosts. Since both Killing vector fields are present in PH, we are able to apply the methods in \cite{klai1} to obtain pointwise decay of the full solution $\psi$ in PH, where we consider $\psi-\psi_0$ in the $\mu=0$ case.

\subsubsection{Non-decay and conservation laws}

As mentioned in Section \ref{sec:birkh}, we can view each null hypersurface in a constant radius spacetime as a degenerate Killing horizon with respect to a suitable Killing field. One might expect conservation laws to form potential problems for decay statements, in light of the behaviour of solutions to (\ref{eq:waveequation}) in extremal black holes and the general picture presented by Aretakis in \cite{are3}. In the case of equation (\ref{eq:waveequation}) in both Nariai and PH these conservation laws are manifested in the conservation of (higher-order) transversal derivatives of $\psi_0$ along each null hypersurface. This follows also from d'Alembert's formula, since $\psi_0$ satisfies a 1+1-dimensional wave equation in double-null coordinates. Moreover, by d'Alembert's formula there is no pointwise decay of $\psi_0$ for generic initial data. 

The results in \cite{are3} for spherically symmetric manifolds also imply that the higher harmonic modes $\psi_l$, $l\geq 1$ are not subject to any (higher-order) conservation laws along degenerate Killing horizons.

\subsection{Outline}
We construct and discuss the geometry of constant radius spacetimes in Section \ref{sec:geometry}. In Section \ref{sec:basictools} we give a short overview of some properties of the Klein-Gordon equation (\ref{eq:kgequation}) on Lorentzian manifolds. The main theorems of the paper, corresponding to the results described in Section \ref{sec:overviewresults}, are then stated in Section \ref{sec:maintheorems}.

We first consider Nariai. We restrict to the static region in Nariai, proving uniform boundedness in Section \ref{sec:boundn} and uniform decay in Section \ref{sec:decn}, using a local method. 

We apply a global method for an alternative proof of uniform boundedness in Section \ref{sec:globalboundnariai} and uniform decay in Section \ref{sec:globaldecaynariai}. The same method is generalised in Section \ref{sec:decaydsn} to obtain uniform decay in $n$-dimensional de Sitter space.

Subsequently, we consider PH. We prove uniform boundedness in Section \ref{sec:boundph} and uniform decay in Section \ref{sec:decph}. 

Finally, in Section \ref{sec:nondec} we show that $\psi_0$ does not decay in either Nariai and PH, in the case where $\mu=0$.

\subsection{Open problems}
A remaining open problem is to venture outside of spherical symmetry by considering (\ref{eq:kgequation}) on a rotating Nariai background. One can construct rotating Nariai by considering a sequence of Kerr-de Sitter spacetimes in which the radii of the event horizon and cosmological horizon approach an extremal value in the limit, with a strictly smaller Cauchy horizon radius. The limit of the corresponding sequence of regions between the event horizons and cosmological horizons is the static region of a spacetime which is called \emph{rotating Nariai}, see \cite{boman}. The metric in the stationary region of rotating Nariai is given by
\begin{equation*}
\begin{split}
g_{rN;a}&=(r_0^2+a^2\cos^2\theta)\left[-f(r_0)(r_0-r_-)(1-x^2)dt^2+\frac{1}{f(r_0)(r_0-r_-)(1-x^2)}dx^2\right]\\
&+\frac{r_0^2+a^2\cos^2\theta}{1+\frac{\Lambda a^2}{3}\cos^2\theta}d\theta^2+\frac{(1+\frac{\Lambda a^2}{3}\cos^2\theta)\sin^2\theta}{r_0^2+a^2\cos^2\theta}\left(2ar_0 x dt+\frac{r_0^2+a^2}{1+\frac{\Lambda a^2}{3}}d\phi\right)^2,
\end{split}
\end{equation*}
where $r_-(a)$ is the radius of the Cauchy horizon, $f(r)$ is a positive linear function of $r>0$ and $r_0(a)>r_-(a)$ is the radius of the event horizons of extremal Kerr-de Sitter with coinciding cosmological and event horizons. For a vanishing angular momentum parameter $a$ of Kerr-de Sitter, rotating Nariai reduces to regular Nariai. Moreover, 2-surfaces of constant $t$ and $x$ are ellipsoids of constant area. Not only is the trapping of null geodesics more complicated in rotating Nariai, there is also \emph{superradiance} present, which is absent in the non-rotating case.

Similarly, one can consider the limit of a sequence of regions between the Cauchy horizon and event horizon of Kerr-de Sitter to obtain the rotating version of Bertotti-Robinson. Moreover, by letting all three horizon radii approach the same value in the limit, one can obtain the rotating version of Pleba\'nski-Hacyan. The metric of rotating Pleba\'nski-Hacyan is given by
\begin{equation*}
\begin{split}
g_{rPH}&=\frac{(r_0^2+a^2\cos^2\theta)}{2f(r_0)}\left(-x^2dt^2+dx^2\right)+\frac{r_0^2+a^2\cos^2\theta}{1+\frac{\Lambda a^2}{3}\cos^2\theta}d\theta^2\\
&+\frac{(1+\frac{\Lambda a^2}{3}\cos^2\theta)\sin^2\theta}{r_0^2+a^2\cos^2\theta}\left(\frac{2ar_0x^2}{f(r_0)} dt+\frac{r_0^2+a^2}{1+\frac{\Lambda a^2}{3}}d\phi\right)^2.
\end{split}
\end{equation*}
In the above expression $a$ is not a free parameter, but rather the value of the angular momentum parameter such that $r_0(a)$ is the radius of the event horizon of extremal Kerr-de Sitter in the case that all horizons coincide. The term $-x^2dt^2+dx^2$ can be viewed as the $\R^{1+1}$ metric in Rindler coordinates\footnote{Rindler coordinates are coordinates that cover a subset of Minkowski spacetime. An observer at rest with respect to Rindler time has a constant positive proper acceleration.}, with a Rindler horizon at $x=0$.

\subsection{Acknowledgements}
I would like to thank Mihalis Dafermos for introducing the problem to me and for his comments and invaluable advice, and Volker Schlue for his comments on the manuscript.

\section{The geometry of constant radius spacetimes}
We will discuss the geometric properties of constant radius spacetimes and introduce suitable coordinate charts.
\label{sec:geometry}
\subsection{A classification of the constant radius spacetimes}
\label{sec:derivation}
In this section we will show that any spacetime with a warped product structure $\mathcal{Q}\times_r\s^2$, where $r=r_0$ is constant,  satisfying the electrovacuum Einstein equations, is locally isometric to a constant radius spacetime, as defined in Section \ref{sec:birkh}. Moreover, we will derive the relations between the geometry of $\mathcal{Q}$ and the value of $r_0$, as they appear in Table \ref{tbl:cr}. 

We can choose local double-null coordinates on $\mathcal{Q}$ such that the metric on $\mathcal{Q}\times_r\s^2$ becomes
\begin{equation}
\label{def:crmetricdoublenull}
g=-\Omega^2(u,v)dudv+r^2\slashed{g}_{\s^2}.
\end{equation}
In this form, the Einstein equations reduce to a system of partial differential equations on the 1+1-dimensional Lorentzian manifold $\mathcal{Q}$, where $r: \mathcal{Q}\to (0,\infty)$. See for example \cite{dafrod1}. We rederive the system of partial differential equations on $\mathcal{Q}$ in Appendix \ref{app:reform} for $\Lambda>0$ to set the notation.

One of these pde, see also (\ref{eenull1b}), is
\begin{equation}
\label{omgnai}
\partial_u \partial_v \log (\Omega^2)=\frac{K}{2}\Omega^2=\frac{2}{r^2}\partial_u r \partial_v r+ \left[\frac{1}{2r^2}-\frac{e^2}{r^4} \right]\Omega^2,
\end{equation}
where $e\in \R$ is the charge corresponding to the spherically symmetric electromagnetic tensor and $K$ is the Gaussian curvature of $\mathcal{Q}$. 

Under the assumption of constant $r=r_0$, it follows from (\ref{omgnai}) that $K$ is constant and is given by
\begin{equation*}
K=r_0^{-4}(r_0^2-2e^2).
\end{equation*}

All maximally symmetric $n+1$ dimensional Lorentzian manifolds with the same constant curvature are locally isometric, see \cite{eis} for a proof. Moreover, we do not need the assumption of maximal symmetry if $n=1$, so we can conclude that $\mathcal{Q}$ must be locally isometric to $dS_2$ if $K>0$, $AdS_2$ if $K<0$ and $\R^{1+1}$ if $K=0$. Given $\Lambda$ and $e$, we will determine the possible signs of $K$ and the values of $r_0$.

In the absence of matter fields other than the electromagnetic stress-energy tensor, we have that
\begin{equation*}
m(r)+\frac{e^2}{2r}-\frac{\Lambda}{6}r^3=M,
\end{equation*}
where $M\in \R$ is a constant and $m(r)$ is the Hawking mass, defined by
\begin{equation*}
m(r)=r\left(\frac{1}{2}+2\Omega^{-2}\partial_u r\partial_v r\right).
\end{equation*}
See also (\ref{eq:varpi}). Hence, $m(r_0)=\frac{r_0}{2}$ and we deduce that the following polynomial equation must hold
\begin{equation}
\label{eq:polyeqcr1}
r_0^2-2Mr_0+e^2-\frac{\Lambda}{3}r_0^4=0.
\end{equation}
This equation also appears in Reissner-Nordstr\"om(-de Sitter) spacetimes for the radii of the event horizons, Cauchy horizons and cosmological horizons.

We also have an equation for $\partial_u\partial_v r$, see (\ref{eefin}),
\begin{equation}
\label{eq:polyeqcr2}
\partial_u \partial_v r=-\frac{1}{r}\partial_ur \partial_v r+(e^2-r^2+\Lambda r^4)\frac{\Omega^2}{4r^3}.
\end{equation}
Since $r=r_0$, the equation above is equivalent to the polynomial equation
\begin{equation*}
e^2-r_0^2+\Lambda r_0^4=0.
\end{equation*}

\subsubsection{$\Lambda=0$}

Consider first the special case $\Lambda=0$.  By (\ref{eq:polyeqcr1}) and (\ref{eq:polyeqcr2}) we find that $r_0^2=e^2=M^2$ and $K=-\frac{1}{e^2}<0$. The corresponding spacetimes are therefore locally isometric to a member of the 1-parameter Bertotti-Robinson family, described in Section \ref{sec:birkh}.

\subsubsection{$\Lambda>0$, $e=0$ or $e^2=\frac{1}{4\Lambda}$}
Suppose now that $\Lambda>0$.  Then (\ref{eq:polyeqcr1}) implies that
\begin{equation*}
r_0^2=\frac{1}{2\Lambda}\left(1\pm \sqrt{1-4\Lambda e^2}\right).
\end{equation*}

Consider the special case $e=0$. Then we must have that $r_0^2=\frac{1}{\Lambda}$ and $K=\Lambda$. The corresponding spacetimes are locally isometric to a member of the 1-parameter Nariai family, described in Section \ref{sec:birkh}.

Consider the special case $e^2=\frac{1}{4\Lambda}$. Then $r_0^2=\frac{1}{2\Lambda}=2e^2$, so $K=0$. The corresponding spacetimes are locally isometric to a member of the 1-parameter Pleba\'nski-Hacyan family, described in Section \ref{sec:birkh}.

\subsubsection{$\Lambda>0$, $0<e^2<\frac{1}{4\Lambda}$}
We are left with the cases where $0<e^2<\frac{1}{4\Lambda}$. 

If $r_0^2>\frac{1}{2\Lambda}$, we must have
\begin{equation*}
r_0^2=\frac{1}{2\Lambda}\left(1+\sqrt{1-4\Lambda e^2}\right).
\end{equation*}
Moreover, $r_0^2>2e^2$, so $K>0$. We refer to the corresponding 2-parameter family of spacetimes $dS_2\times_{r_0} \s^2$ as \emph{charged} Nariai solutions.

If $r_0^2<\frac{1}{2\Lambda}$, we must have
\begin{equation*}
r_0^2=\frac{1}{2\Lambda}\left(1-\sqrt{1-4\Lambda e^2}\right).
\end{equation*}
If $r_0^2<2e^2$, then $K<0$ and we refer to the corresponding 2-parameter of spacetimes $AdS_2\times_{r_0}\s^2$ as \emph{cosmological} Bertotti-Robinson solutions. If $2e^2<r_0^2<\frac{1}{2\Lambda}$, $K>0$, and we again refer to the corresponding 2-parameter family of spacetimes $dS_2\times_{r_0}\s^2$ as charged Nariai solutions.

\begin{remark}
We can also consider the $\Lambda<0$ case. If $e\neq 0$, spherically symmetric electrovacuum spacetimes with a constant radius function $r_0$ must satisfy
\begin{equation*}
r_0^2=\frac{1}{2|\Lambda|}\left(-1+\sqrt{1+4|\Lambda| e^2}\right).
\end{equation*}
Consequently, $K<0$ for all $e\neq 0$ and $\mathcal{Q}$ must be locally isometric to $AdS_2$.
\end{remark}

\subsection{Killing vector fields in constant radius spacetimes}
\label{sec:kvf}
We can easily find the Killing vector fields corresponding to the induced metric $g_{\mathcal{Q}}$ on $\mathcal{Q}$. Choose $(u,v)$ coordinates such that $\Omega^2(0,v)=1$ and $\Omega^2(u,0)=1$, where $u,v \in \R$. Note that these coordinates do not cover the entire $dS_2$ or $AdS_2$ manifolds. 

For all constant radius spacetimes we can then construct the following basis of Killing vector fields on $\mathcal{Q}$,
\begin{align*}
T&=\left(1-\frac{K}{4}{u^2}\right)\frac{\partial}{\partial u}+\left(1-\frac{K}{4}{v^2}\right)\frac{\partial}{\partial v},\\
X&=\left( 1+ \frac{K}{4}v^2 \right) \frac{\partial}{\partial v}-\left( 1+ \frac{K}{4}u^2 \right) \frac{\partial}{\partial u},\\
Y&=v\frac{\partial}{\partial v}- u \frac{\partial}{\partial u}.
\end{align*}
One can check that $T$, $X$ and $Y$ indeed span a 3-dimensional Lie algebra. In particular, in PH the vector field $T$ generates time translations, $X$ generates space translations and $Y$ generates boosts in $\R^{1+1}$. More generally, in the context of near-horizon spacetimes, $Y$ is the generator of the \emph{dilation symmetry}.

Furthermore, $T+X$ and $T-X$ generate outgoing and ingoing degenerate Killing horizons $\{u=0\}$ and $\{v=0\}$, respectively. By homogeneity of the spacetimes, each null hypersurface in a constant radius spacetime is a degenerate horizon with respect to a suitable Killing vector field.

\subsection{Coordinate charts on Nariai}
Consider the Nariai spacetime of Section \ref{sec:birkh}.
\subsubsection{Global coordinates}
\label{sec:narstco}
There exist coordinates $(\tilde{t},\tilde{x},\theta,\phi)$, where $\tilde{t}\in \R$ and $\tilde{x}\in \R$ are coordinates on $dS_2$, that cover the entire Nariai spacetime, such that
\begin{equation*}
g=K^{-1}(-d\tilde{t}^2+\cosh^2\tilde{t}d\tilde{x}^2)+r_0^2\slashed{g}_{\s^2}.
\end{equation*}
Via the global coordinates above, we will derive an expression for the term $\Omega^2(u,v)$ appearing in (\ref{def:crmetricdoublenull}), with the gauge choice $\Omega^2(u,0)=\Omega^2(0,v)=1$. We will also see that the isometry corresponding to the spacelike Killing vector field $X$ is naturally expressed in these global coordinates. In order to write the above metric in a double-null foliation, we define $\hat{t}(\tilde{t})$ such that
\begin{equation*}
\frac{d\hat{t}}{d\tilde{t}}= \frac{1}{\cosh \tilde{t}}. 
\end{equation*}
One can check that the following suffices,
\begin{equation*}
\hat{t}(\tilde{t})=2 \arctan\left(\tanh\left(\frac{\tilde{t}}{2}\right)\right),
\end{equation*}
where $\hat{t}\in (-\frac{\pi}{2},\frac{\pi}{2})$. The inverse is given by
\begin{equation*}
\tilde{t}(\hat{t})=2\arctanh\left(\tan\left(\frac{\hat{t}}{2}\right)\right).
\end{equation*}
Now,
\begin{equation*}
\begin{split}
g&=K^{-1} \cosh^2(\tilde{t}(\hat{t}))\left(-d\hat{t}^2+d\tilde{x}^2\right)+r_0^2\slashed{g}_{\s^2}\\
&=-K^{-1}\cosh^2(\tilde{t}(\hat{t}))d\hat{u}d\hat{v}+r_0^2\slashed{g}_{\s^2},
\end{split}
\end{equation*}
where $\hat{u}=\hat{t}-\tilde{x}$ and $\hat{v}=\hat{t}+\tilde{x}$. We write $\tilde{t}(\hat{t})=\tilde{t}(\hat{u},\hat{v})$. To obtain $\Omega^2(0,v)=1$, we first define $v(\hat{v})$ as the solution to the ODE
\begin{align*}
\frac{dv}{d\hat{v}}&=K^{-\frac{1}{2}}\cosh^2\left(\tilde{t}(0,\hat{v})\right)=K^{-\frac{1}{2}} \cos^{-2}\left(\frac{\hat{v}}{2}\right),\\
v(0)&=0.
\end{align*}
We integrate to find
\begin{equation*}
v(\hat{v})=\frac{2}{\sqrt{K}}\tan \frac{\hat{v}}{2}.
\end{equation*}
Hence, we can only cover the region $-\pi<\hat{v}< \pi$ if we replace the outgoing null coordinate $\hat{v}$ with $v$. The metric becomes
\begin{equation*}
g=-\frac{\cosh^2(\tilde{t}(\hat{u},\hat{v}))}{\sqrt{K}\cosh^2(\tilde{t}(0,\hat{v}))}d\hat{u}dv.
\end{equation*}
Define $u(\hat{u})$ by
\begin{align*}
\frac{du}{d\hat{u}}&=\frac{\cosh^2\tilde{t}(\hat{u},0))}{\sqrt{K}\cosh^2(\tilde{t}(0,0))}=K^{-\frac{1}{2}}\cosh^2(\tilde{t}(\hat{u},0))=K^{-\frac{1}{2}}\cos^{-2}\left(\frac{\hat{u}}{2}\right),\\
u(0)&=0.
\end{align*}
Then
\begin{equation*}
u(\hat{u})=\frac{2}{\sqrt{K}} \tan \frac{\hat{u}}{2}.
\end{equation*}
We need to restrict to $-\pi<\hat{u}<\pi$. Consequently, 
\begin{equation}
\label{def:nariainicegaugedoublenull}
\begin{split}
g&=-\frac{\cosh^2(\tilde{t}(\hat{u},\hat{v}))}{\cosh^2(\tilde{t}(0,\hat{v}))\cosh^2(\tilde{t}(\hat{u},0))}dudv+r_0^2\gamma\\
&=-\frac{\cos^2\left(\frac{\hat{u}}{2}\right)\cos^2\left(\frac{\hat{v}}{2}\right)}{\cos^2\left(\frac{\hat{u}+\hat{v}}{2}\right)}dudv+r_0^2\gamma,
\end{split}
\end{equation}
where $u,v\in \R$. We see immediately that $\Omega^2(0,v)=\Omega^2(u,0)=1$. The $(u,v)$ coordinates only cover the region $-\pi<\hat{u},\hat{v}<\pi$ of the spacetime with moreover $-\pi <\hat{u}+\hat{v}<\pi$.

The coordinate vector field $\partial_{\tilde{x}}$ is a Killing vector field. In $(u,v)$ coordinates we find that
\begin{equation*}
\begin{split}
X&=\sqrt{K}\partial_{\tilde{x}},\\
T&=\sqrt{K}\left[\left(\sin^2 \left(\frac{\hat{u}}{2}\right)-\cos^2 \left(\frac{\hat{u}}{2}\right)\right)\partial_{\hat{u}}+\left(\sin^2 \left(\frac{\hat{v}}{2}\right)-\cos^2 \left(\frac{\hat{v}}{2}\right)\right)\partial_{\hat{v}}\right].
\end{split}
\end{equation*}
$T$ is timelike in the open rectangle $\mathcal{R}:=\{(\hat{u},\hat{v},\theta,\phi)\in \mathcal{M}\,|\,(\hat{u},\hat{v})\in (-\frac{\pi}{2},\frac{\pi}{2})^2\}$.

\subsubsection{Static coordinates}
\label{sec:staticcoordinates}
Although the chart $(\tilde{t},\tilde{x},\theta,\phi)$ covers the entire Nariai spacetime, the Killing vector field $T$ is not naturally expressed. Indeed, we can equip the rectangle $\mathcal{R}$ introduced above with $\emph{static coordinates}$, which we will denote by $(t,x,\theta,\phi)$, with $t\in \R$, $x\in \left(-\frac{1}{\sqrt{K}},\frac{1}{\sqrt{K}}\right)$ and $\theta,\phi$ coordinates on $\s^2$, such that the metric is given by
\begin{equation}
g=-(1-K x^2)dt^2+(1-Kx^2)^{-1}dx^2+r_0^2\slashed{g}_{\s^2},
\end{equation}
where $K>0$ is the Gaussian curvature introduced above. In order to construct null coordinates, we introduce the tortoise function $x_*$ such that
\begin{align*}
\frac{dx_*}{dx}&=(1-Kx^2)^{-1},\\
x_*(0)&=0.
\end{align*}
We integrate the above ode to find,
\begin{align*}
x_*(x)&=\frac{1}{\sqrt{K}}\arctanh \left( \sqrt{K}x\right),\\
x(x_*)&=\frac{1}{\sqrt{K}}\tanh \left( \sqrt{K}x_*\right).
\end{align*}
Hence, $x_*\in \R$, with $x_*\to \infty$ when $x\to \frac{1}{\sqrt{K}}$ and $x_*\to -\infty$ when $x\to -\frac{1}{\sqrt{K}}$. Now define the null coordinates $u$ and $v$ by
\begin{align*}
u&:=t-x_*,\\
v&:=t+x_*.
\end{align*}
Then
\begin{equation}
\label{def:nariaistaticcoords}
g=-(1-Kx^2)dudv+r_0^2\slashed{g}_{\s^2}.
\end{equation}

\begin{figure}
\begin{center}
\includegraphics[width=2.5in]{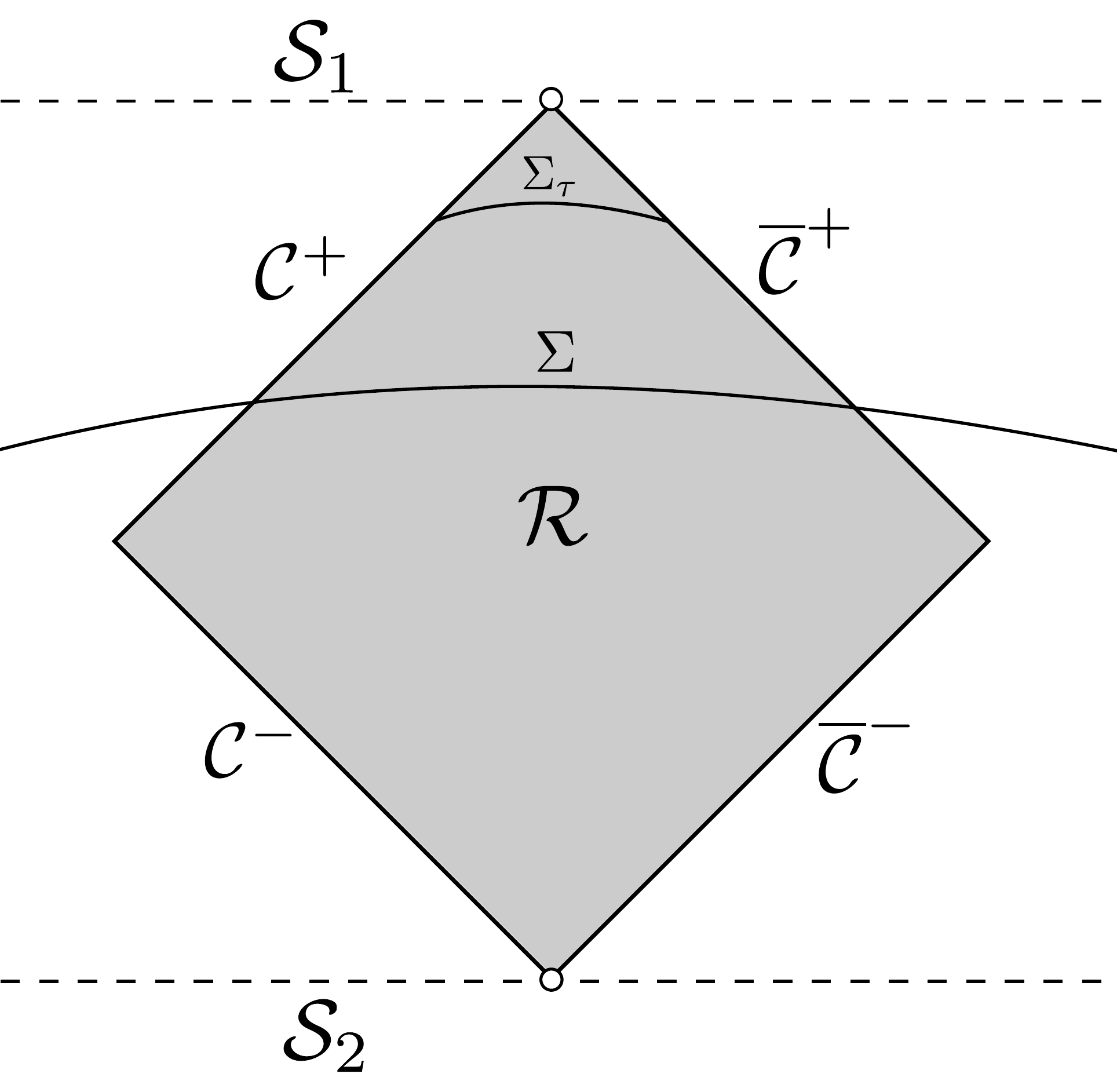}
\caption{Part of the Penrose diagram of Nariai with a shaded static region $\mathcal{R}$.}
\label{fig:staticnariai}
\end{center}
\end{figure}

We can relate the null coordinates appearing in (\ref{def:nariaistaticcoords}) to the null coordinates in which the metric is expressed by (\ref{def:nariainicegaugedoublenull}), that moreover define the vector field $T$, to establish that in static coordinates $T$ is given by
\begin{equation*}
T=\partial_t.
\end{equation*}
Static coordinates are therefore the coordinates in which the vector field $T$ is most naturally expressed. As we remarked in the aforementioned section, $T$ is timelike in $\mathcal{R}$ and becomes null on $x=\pm \frac{1}{\sqrt{K}}$. We characterise the boundary of $\mathcal{R}$ as follows in terms of ingoing and outgoing \emph{cosmological horizons},
\begin{equation*}
\partial \mathcal{R}=\mathcal{C}^+\cup\mathcal{C}^-\cup\overline{\mathcal{C}}^+\cup\overline{\mathcal{C}}^-,
\end{equation*}
where we can express in global double null coordinates $(\hat{u},\hat{v})$
\begin{align*}
\mathcal{C}^+&=\left\{\hat{u}=\frac{\pi}{2},\, \hat{v}\geq -\frac{\pi}{2}\right\},\\
\overline{\mathcal{C}}^+&=\left\{\hat{v}=\frac{\pi}{2},\, \hat{u}\geq -\frac{\pi}{2}\right\},\\
\mathcal{C}^-&=\left\{\hat{v}=-\frac{\pi}{2},\, \hat{u}\leq \frac{\pi}{2}\right\},\\
\overline{\mathcal{C}}^-&=\left\{\hat{u}=-\frac{\pi}{2},\, \hat{v}\leq \frac{\pi}{2}\right\}.
\end{align*}
Moreover,
\begin{align*}
\{x=\frac{1}{\sqrt{K}}\}&=\mathcal{C}^+\cup \mathcal{C}^-,\\
\{x=-\frac{1}{\sqrt{K}}\}&=\overline{\mathcal{C}}^+ \cup\overline{\mathcal{C}}^-.
\end{align*}

By passing to ingoing or outgoing Eddington-Finkelstein-type coordinates, we can easily see that $\mathcal{C}^+$ and $\overline{\mathcal{C}}^+$ have a positive surface gravity $\kappa=\sqrt{K}$ with respect to the timelike Killing vector field $T$. Positivity of the surface gravity will play an important role in proving uniform boundedness and decay of solutions to (\ref{eq:kgequation}).

We can express the global time $\tilde{t}$ as a function of the static coordinates $t$ and $x$,
\begin{equation*}
\tilde{t}(t,x)=2 \arctanh\left(\tan\left(\frac{\hat{t}(t,x)}{2}\right)\right),
\end{equation*}
 where
 \begin{equation*}
 \hat{t}(t,x)=\arctan\left(\tanh\left(\frac{\sqrt{K}}{2}(t+x_*(x))\right)\right)+\arctan\left(\tanh\left(\frac{\sqrt{K}}{2}(t-x_*(x))\right)\right)
 \end{equation*}
 In particular, $\tilde{t}(t,0)=\sqrt{K}t$ and
 \begin{equation}
 \label{eq:tildetvst}
 \tilde{t}(t,x)\geq \sqrt{K}(t-|x_*(x)|).
 \end{equation}
\subsection{The constant radius limit}
For completeness, we will give a precise construction of Nariai and Pleba\'nski-Hacyan as the constant radius limits of the region between the cosmological and event horizons of a sequence of subextremal Reissner-Nordstr\"om-de Sitter spacetimes. A similar construction can be performed for the region between the event and Cauchy horizons, where the constant radius limit is Bertotti-Robinson.

The Reissner-Nordstr\"om-de Sitter metric can be expressed as
\begin{equation*}
g=-\frac{f(r)}{r^2}(r-r_+)(r_{++}-r)d\tilde{t}^2+\frac{r^2}{f(r)(r-r_+)(r_{++}-r)}dr^2+r^2 \slashed{g}_{\s^2},
\end{equation*}
where $\tilde{t}\in \R$, $r\in (r_+,r_{++})$ and $r_+$ and $r_{++}$ are the radii of the event horizon and cosmological horizon, respectively. See for example \cite{gripod1}. The function $f(r)$ is a quadratic polynomial with $f(r)>0$ for $r\in [r_+,\infty)$. Moreover, $r_{+}< r_0 <r_{++}$ for all subextremal solutions, where $r_0$ is the radius of the corresponding extremal solution.

We define $\epsilon=\epsilon(M,e)>0$ by $\epsilon=r_{++}-r_0=r_0-r_+$ and we introduce the shifted coordinate $\rho:=r-r_0$, $\rho\in(-\epsilon,\epsilon)$, such that
\begin{equation*}
g=-\frac{f(r)}{r^2}(\rho+\epsilon)(\epsilon-\rho)d\tilde{t}^2+\frac{r^2}{f(r)(\rho+\epsilon)(\epsilon-\rho)}d\rho^2+r^2 \slashed{g}_{\s^2}.
\end{equation*}
Now rescale, $\chi:=\epsilon^{-1}\rho$ and $\tau:=\epsilon \tilde{t}$, to obtain
\begin{equation*}
g=-\frac{f(r)}{r^2} (1-\chi^2)d\tau^2+\frac{r^2}{f(r)(1-\chi^2)}d\chi^2+r^2\slashed{g}_{\s^2},
\end{equation*}
where $\chi\in(-1,1)$ and $\tau\in \R$. Fix $\Lambda>0$ and consider a sequence of subextremal Reissner-Nordstr\"om-de Sitter metrics $g_n$, with parameters $M_n$ and $e_n$, such that $\lim_{n\to \infty} \epsilon_n=0$. Then the metric components $(g_n)_{\alpha \beta}$ in $(\tau,\chi,\theta,\phi)$ coordinates, with $\tau\in \R$, $\chi\in (-1,1)$, $0<\theta<\pi$ and $0<\phi<2\pi$ converge to a limit $g_N$, with
\begin{equation}
\label{eq:nariaimetricversion1}
g_N=-\frac{f(r_0)}{r_0^2} (1-\chi^2)d\tau^2+\frac{r_0^2}{f(r_0)(1-\chi^2)}d\chi^2+r_0^2\slashed{g}_{\s^2}.
\end{equation}
We introduce a final rescaling, $x=\frac{r_0}{\sqrt{f(r_0)}}\chi$ and $t=\frac{\sqrt{f(r_0)}}{r_0}\tau$, to rewrite (\ref{eq:nariaimetricversion1}) as
\begin{equation}
\label{eq:nariaimetricversion2}
g_N=-\left(1-\frac{f(r_0)}{r_0^2}x^2\right)dt^2+\left(1-\frac{f(r_0)}{r_0^2}x^2\right)^{-1}dx^2+r_0^2 \slashed{g}_{\s^2}.
\end{equation}
The metric (\ref{eq:nariaimetricversion2}) is precisely a Nariai metric in static coordinates, where $K=\frac{f(r_0)}{r_0^2}$ and $x\in (-\frac{1}{\sqrt{K}},\frac{1}{\sqrt{K}})$.

Pleba\'nski-Hacyan can be constructed as the constant radius limit of a sequence of extremal Reissner-Nordstr\"om-de Sitter spacetimes with the radii of two horizons coinciding. We consider extremal Reissner-Nordstr\"om-de Sitter in ingoing Eddington-Finkelstein-type coordinates $(\tilde{v},r)$,
\begin{equation*}
g=\frac{h(r)}{r^2}(r-r_+)^2(r-r_-)d\tilde{v}^2+2d\tilde{v}dr+r^2\slashed{g}_{\s^2},
\end{equation*}
where $h(r)$ is a first-order polynomial with $h(r)>0$ and $r_-$ is the radius of the Cauchy horizon and $r_-<r_0<r_+=r_{++}$. The case where $r_-=r_+$ proceeds analogously. Here, $r\in(r_-,r_+)$ and $\tilde{v}\in \R$. 

Let $\epsilon=\epsilon(e)>0$ be defined by $\epsilon=r_{++}-r_0=r_0-r_-$ and consider the shifted coordinate $\rho=r-r_0$, $\rho\in (-\epsilon,\epsilon)$. Moreover, rescale $\chi=\epsilon^{-1}\rho$, $\chi\in (-1,1)$, and $v=\epsilon \tilde{v}$, $v\in \R$, such that
\begin{equation*}
g=-\epsilon \frac{h(r)}{r^2}(\chi-1)(1-\chi^2)dv^2+2dvd\chi+r^2\slashed{g}_{\s^2}.
\end{equation*}
As above, we consider a sequence of charge parameters $e_n$ such that \newline $\lim_{n\to \infty} \epsilon(e_n)=0$ and we obtain the limiting spacetime,
\begin{equation*}
g_{PH}=2dvd\chi+r_0^2\slashed{g}_{\s^2},
\end{equation*}
where we can extend $\chi\in \R$. If we rescale $u:=-2\chi$, we see that the induced metric $g_{\mathcal{Q}}$ is the standard metric on $\R^{1+1}$ in double-null coordinates, $u\in \R$, $v\in\R$.
\pagebreak
\subsection{Foliations by spacelike hypersurfaces}
\subsubsection{Nariai}
\label{sec:foliationsnariai}
Recall from Section \ref{sec:narstco} that we can cover Nariai by the global coordinates $(\tilde{t},\tilde{x},\theta,\phi)$ with the metric given by
\begin{equation*}
g=K^{-1}\left(-d\tilde{t}^2+\cosh^2\tilde{t}d\tilde{x}^2\right)+r_0^2 \slashed{g}_{\s^2}.
\end{equation*} 
We can therefore foliate Nariai by spacelike hypersurfaces $\tilde{\Sigma}_{\tau}:=\{\tilde{t}=\tau\}$, where we take $\tau\in \R$. See Figure \ref{fig:globalnariai}. The induced volume form on $\Sigma_{\tau}$ is given by
\begin{equation*}
d\mu_{\tilde{\Sigma}_{\tau}}=K^{-\frac{1}{2}}r_0^2\cosh \tau\,dxd\mu_{\s^2}.
\end{equation*}
Note that the volume form is expanding exponentially in the time parameter $\tau$. The global foliation of Nariai will be used in Sections \ref{sec:globalboundnariai} and \ref{sec:globaldecaynariai}.

We will also consider a different foliation of a static region $\mathcal{R}$ in Nariai. See Figure \ref{fig:staticnariai}. Let $\Sigma$ be an SO(3)-invariant spacelike hypersurface in Nariai, that is given by a level set of a smooth function $h_N: \mathcal{M} \to \R$, i.e\, $\Sigma= h_N^{-1}(\{0\})$. Let $n_{\Sigma}$ be the future-directed unit normal to $\Sigma$. We make the assumption that the following uniform bounds hold
\begin{align*}
C_1&\leq -g^{-1}(dh_N,dh_N)\leq C_2,\\
C_1&\leq -g(n_{\Sigma},T)\leq C_2,
\end{align*}
where $C_1$ and $C_2$ are positive constants. We consider a static region $\mathcal{R}$ such that $\Sigma$ intersects the cosmological horizons $\mathcal{C}^+$ and $\overline{\mathcal{C}}^+$, see Figure \ref{fig:staticnariai}. By the assumptions on the normal, we can express the volume form corresponding to the metric induced on $\Sigma$ by
\begin{equation*}
d\mu_{\Sigma}= b_N(\rho) d\rho d\mu_{\s^2},
\end{equation*}
where $(\rho,\theta,\phi)$ are coordinates along $\Sigma$, with $\rho\in \R$, and $b_N: \R \to \R_+$ is a function that is uniformly bounded above and away from zero everywhere.
\begin{figure}
\begin{center}
\includegraphics[width=3.0in]{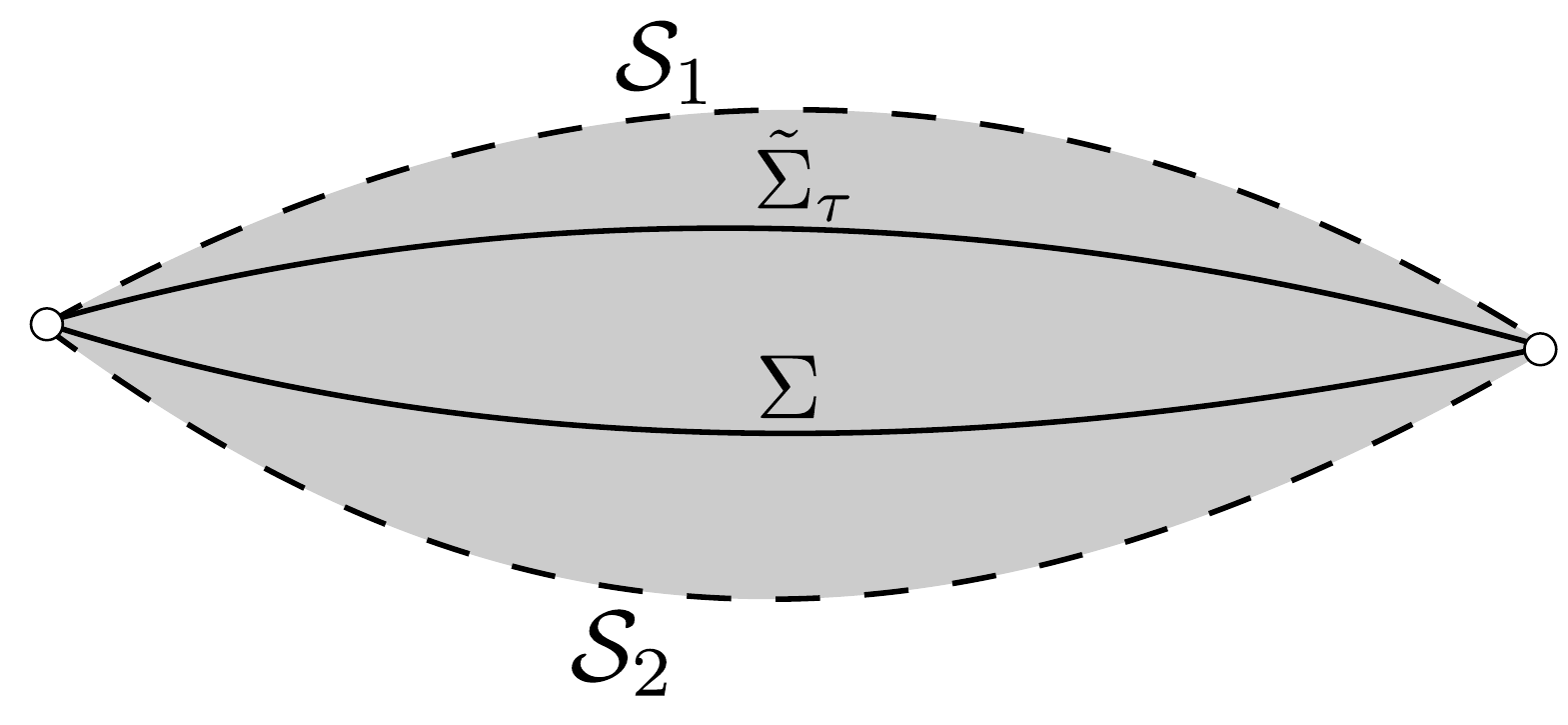}
\caption{The global foliation of Nariai.}
\label{fig:globalnariai}
\end{center}
\end{figure}

Let $\Sigma_{\tau}:=\phi_{\tau}(\Sigma\cap \mathcal{R})$, where $\phi_{\tau}$ is the isometric flow corresponding to the vector field $T$. Consequently,
\begin{equation*}
\mathcal{R}\cap J^+(\Sigma)=\bigcup_{\tau\in[0,\infty)} \Sigma_{\tau}.
\end{equation*}
By construction, the volume form restricted to $\Sigma_{\tau}$ can be expressed in the Lie propagated coordinates $(\rho,\theta,\phi)$ on $\Sigma_{\tau}$
\begin{equation*}
d\mu_{\Sigma_{\tau}}= b_N(\rho) d\rho d\mu_{\s^2}
\end{equation*}
We will use the shorthand notation $n_{\tau}:=n_{\Sigma_{\tau}}$.

The uniform estimates in Section \ref{sec:boundn} and Section \ref{sec:decn} will be carried out in the region $\mathcal{R}(0,\tau)$, defined by
\begin{equation*}
\mathcal{R}(0,\tau):=\bigcup_{\bar{\tau}\in[0,\tau)} \Sigma_{\bar{\tau}}.
\end{equation*}
More generally, we will consider the regions
\begin{equation*}
\mathcal{R}(\tau_1,\tau_2):=\bigcup_{\bar{\tau}\in[\tau_1,\tau_2)} \Sigma_{\bar{\tau}}.
\end{equation*}

\subsubsection{Pleba\'nski-Hacyan}
In PH, we can express the metric $g$ on $\mathcal{M}=\R^2\times \s^2$ in  $(t,x,\theta,\phi)$ coordinates, where $(t,x)$ are rectilinear coordinates on $\R^{1+1}$, to obtain
\begin{equation*}
g=-dt^2+dx^2+r_0^2\slashed{g}_{\s^2}.
\end{equation*}

As in the Nariai case, we take $\Sigma$ to be the $SO(3)$-invariant level set of a smooth function $h_{PH}: \mathcal{M} \to \R$ and require
\begin{align*}
C_1&\leq -g^{-1}(dh_{PH},dh_{PH})\leq C_2,\\
C_1&\leq -g(n_{\Sigma},T)\leq C_2,
\end{align*}
where $C_1$ and $C_2$ are positive constants.
\label{sec:foliationsph}
\begin{figure}
\begin{center}
\includegraphics[width=2.5in]{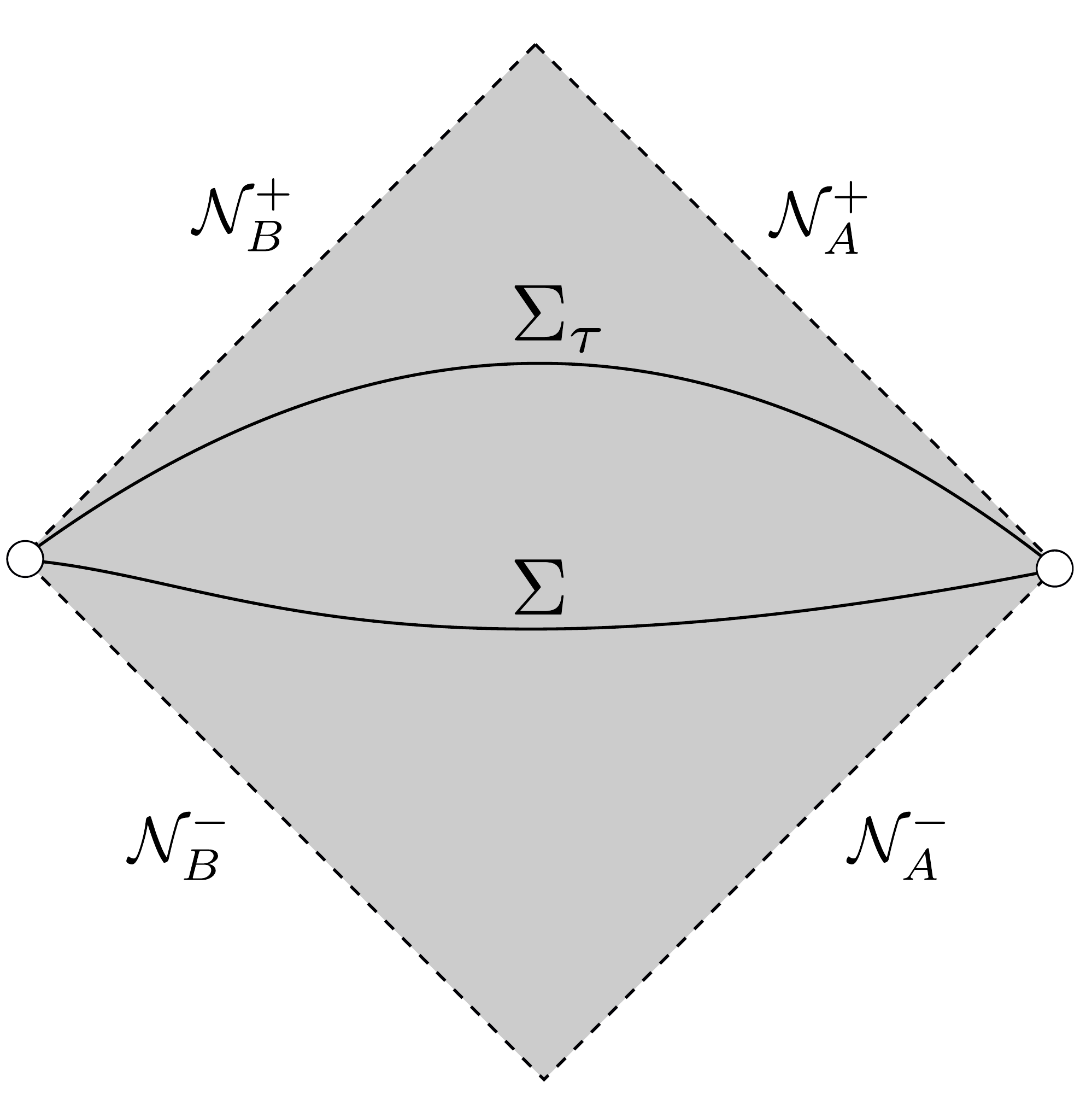}
\caption{The global foliation of PH.}
\label{fig:globalph}
\end{center}
\end{figure}
Since $T$ is a globally timelike Killing vector field, we can define $\Sigma_{\tau}:=\phi_{\tau}(\Sigma)$, such that
\begin{equation*}
\mathcal{M}\cap J^+(\Sigma)=\bigcup_{\tau\in[0,\infty)} \Sigma_{\tau}.
\end{equation*}
Moreover, we define the region $\mathcal{R}(0,\tau)$ by
\begin{equation*}
\mathcal{R}(0,\tau)=\bigcup_{\bar{\tau}\in[0,\tau)} \Sigma_{\bar{\tau}}.
\end{equation*}
See Figure \ref{fig:globalph}. More generally, we can consider the regions
\begin{equation*}
\mathcal{R}(\tau_1,\tau_2):=\bigcup_{\bar{\tau}\in[\tau_1,\tau_2)} \Sigma_{\bar{\tau}}.
\end{equation*}

We can construct coordinates $(\rho,\theta,\phi)$ on $\Sigma$, with $\rho\in \R$, such that the the volume form restricted to $\Sigma_{\tau}$ is given by
\begin{equation*}
d\mu_{\Sigma_{\tau}}= b_{PH}(\rho) d\rho d\mu_{\s^2},
\end{equation*}
where $b_{PH}: \R \to \R_+$ is a function that is uniformly bounded above and away from zero everywhere.

For our purposes it is also convenient to consider a foliation of the interior of a lightcone in PH. See Figure \ref{fig:hypph}. We fix the origin $(t,x)=(0,0)$ to lie in in $\Sigma$. Let $(\tau,\rho,\theta,\phi)$ be coordinates on $\Sigma_{\tau}$, where $\rho=x|_{\Sigma}$. Then the point $(\tau,0)$ is equal to $(t,0)$ in $(t,x)$ coordinates. Moreover, by the assumptions on $\Sigma$, we have that $\tau \sim t$. We consider the lightcone $\mathcal{C}$ defined by
\begin{equation*}
\mathcal{C}:=\{t\leq |x|\}.
\end{equation*}

$\mathcal{C}$ can be foliated by hyperboloids $\mathcal{H}_{s}=\{(t,x)\in \mathcal{Q}\: :\: t^2-x^2=s^2,\:t\geq 0\}$, where $s>0$. The future-directed normal $n_{\mathcal{H}_{s}}$ to $\mathcal{H}_s$ is given by,
\begin{equation*}
n_{\mathcal{H}_s}=\frac{1}{s}\left(t \frac{\partial}{\partial t}+x\frac{\partial}{\partial x}\right).
\end{equation*}
Indeed, since $s$ is constant on $\mathcal{H}_{s}$, $ds^{\sharp}$ points in the direction of the normal and
\begin{equation*}
2 s ds=d(s^2)=d(t^2-x^2)=2tdt-2xdx.
\end{equation*}

\begin{figure}
\begin{center}
\includegraphics[width=3.5in]{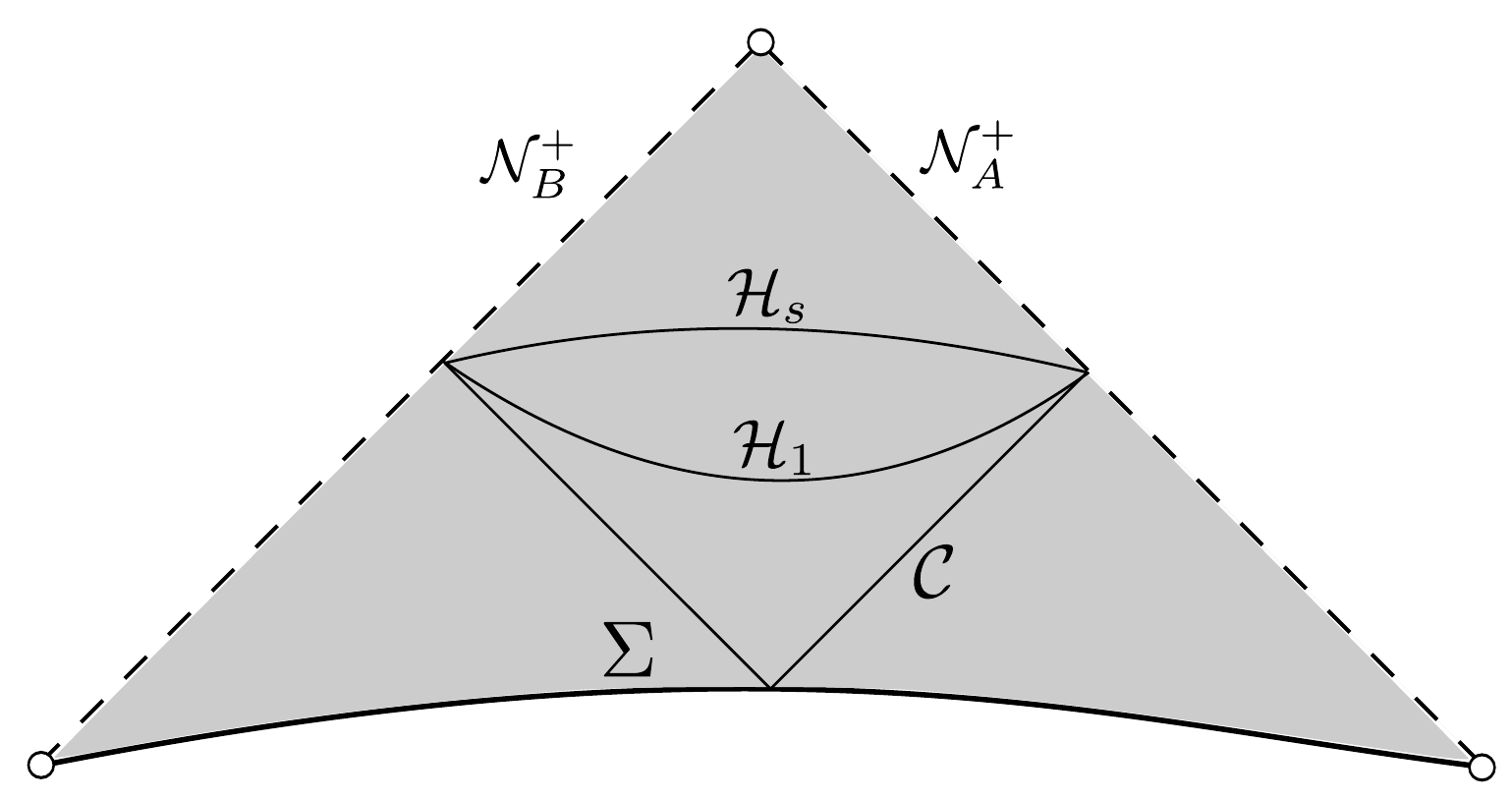}
\caption{The hyperboloidal foliation of the lightcone $\mathcal{C}$ in PH.}
\label{fig:hypph}
\end{center}
\end{figure}

\subsection{Trapping of null geodesics}
\label{sec:trapping}
In \cite{sbier1}, Sbierski showed that on a globally hyperbolic Lorentzian manifold, there exist solutions to (\ref{eq:waveequation})  with an energy that is localised around any given null geodesic. As a consequence, he showed in various black hole spacetimes that the presence of trapped null geodesics (a precise definition of ``trapped'' in Nariai and PH will be given below) must necessarily lead to a loss of derivatives in local energy decay statements. 

Since Nariai and PH are both spherically symmetric and have a timelike Killing vector field if we restrict to a static region in Nariai, the dynamics of geodesics in both spacetimes can easily be derived and a sensible definition of trapped null geodesics can be given. 

\subsubsection{Nariai}
Let $\gamma$ be a null geodesic in $\mathcal{R}$ of Nariai. We say $\gamma$ is \emph{trapped with respect to the foliation $\Sigma_{\tau}$} if there exists a compact subset $K\subset \Sigma_0$ such that $\textnormal{supp} \,\gamma \subset \bigcup_{\tau\geq 0}\phi_{\tau}(K)$.

Without loss of generality, we may assume $\gamma_{\theta}=\frac{\pi}{2}$. Let us denote
\begin{align*}
\epsilon&=-g(T,\dot{\gamma})=(1-Kx^2)\dot{\gamma}_t,\\
\ell&=g(\partial_{\phi},\dot{\gamma})=r_0^2 \dot{\gamma}_{\phi}.
\end{align*}
By the properties of Killing vector fields along geodesics, the quantities $\epsilon>0$ and $\ell>0$ are conserved along $\gamma$.

The equation $g(\dot{\gamma},\dot{\gamma})=0$ can now be rewritten as
\begin{equation*}
\epsilon^2=\dot{\gamma}_x^2+ V_N(x),
\end{equation*}
where $V_N$ is the effective potential corresponding to the geodesic, given by
\begin{equation*}
V_N(x)=\frac{\ell^2}{r_0^2}(1-Kx^2).
\end{equation*}
The potential has a maximum at $x=0$. Consequently, there exist null geodesics of the form $\gamma(s)=(t(s),0,\phi(s),\frac{\pi}{2})$ such that $\textnormal{supp}\,\gamma \subset \{x=0\}$. The submanifold $\{x=0\}$ is called the \emph{photon sphere}. 

Each geodesic is characterised by a triple $(x(0),\epsilon,\ell)$. Trapped null geodesics correspond to the triple $(0,\epsilon,\ell)$. A non-trivial perturbation of $x(0)=0$, causes the resulting null geodesic to no longer be trapped. That is to say, no longer contained in a bounded region $-x_0\leq x \leq x_0$. The trapping can therefore be considered unstable.

Note that by homogeneity of the spacetime, each sphere foliating Nariai contains a trapped null geodesic, i.e.\ each sphere is a photon sphere.

\subsubsection{Pleba\'nski-Hacyan}
As in Narai, we say a null geodesic $\gamma$ is trapped in PH with respect to the foliation $\Sigma_{\tau}$ if there exists a compact subset $K\subset \Sigma$ such that $\textnormal{supp} \,\gamma \subset \bigcup_{\tau\geq 0}\phi_{\tau}(K)$. Moreover, we can define the following conserved quantities along a geodesic $\gamma$ with fixed $\gamma_{\theta}=\frac{\pi}{2}$:
\begin{align*}
\epsilon&=-g(T,\dot{\gamma})=\dot{\gamma}_t,\\
\ell&=g(\partial_{\phi},\dot{\gamma})=r_0^2 \dot{\gamma}_{\phi}.
\end{align*}
We have that
\begin{equation*}
\epsilon^2=\dot{\gamma}_x^2+ V_{PH}(x),
\end{equation*}
where the effective potential $V_{PH}$ is the constant function
\begin{equation*}
V_N(x)=\frac{\ell^2}{r_0^2}.
\end{equation*}

There exist trapped null geodesics of the form $\gamma(s)=(t(s),x_0,\phi(s),\frac{\pi}{2})$ for each $x_0\in \R$, such that $\epsilon^2=\frac{\ell^2}{r_0^2}$. Hence, each sphere foliating PH is a photon sphere.

In this case, the perturbations of a triple $(x(0),\epsilon,\ell)=(x_0,\epsilon=\frac{\sqrt{\ell}}{r_0},\ell)$ corresponding to a trapped null geodesic, result in the geodesic remaining trapped, if and only if the constraint $\epsilon^2=\frac{\ell^2}{r_0^2}$ remains satisfied. In this sense, perturbations resulting in trapped null geodesics form a codimension 1 subset of all perturbations of the triple $(x(0),\epsilon,\ell)$. The trapping is unstable.

\section{The Klein-Gordon equation on Lorentzian manifolds}
\label{sec:basictools}
\subsection{The Cauchy problem}
\label{sec:cauchykgequation}
Let $\Sigma$ be a Cauchy hypersurface in a globally hyperbolic Lorentzian manifold $(\scrM,g)$. We consider the equation (\ref{eq:kgequation}) with initial data imposed on $\Sigma$.
\begin{theorem}
\label{waveeq}
For a fixed $\mu\in \R$ and $\Psi\in H^2_{\textnormal{loc}}(\Sigma)$, $\Psi'\in H^1_{\textnormal{loc}}(\Sigma)$, there exists a unique $\psi: \scrM \to \R$, with $\psi|_{\mathcal{S}} \in H^2_{\textnormal{loc}}(\mathcal{S})$, $n_{\mathcal{S}}\psi|_{\mathcal{S}} \in H^1_{\textnormal{loc}}(\mathcal{S})$, for all spacelike submanifolds $\mathcal{S}\subset \scrM$ with unit future normal $n_{\mathcal{S}}$ in $\scrM$, satisfying
\begin{align*}
 (\square_g -\mu^2)\psi&=0,\\
\psi|_{\Sigma}&=\Psi,\\
n_{\Sigma}\psi|_{\Sigma}&=\Psi'.
\end{align*}
For $m \geq 1$, if $\Psi \in H^{m+1}_{\textnormal{loc}}(\Sigma)$, $\Psi' \in H^{m}_{\textnormal{loc}}(\Sigma)$, then $\psi|_{\mathcal{S}} \in H^{m+1}_{\textnormal{loc}}(\mathcal{S})$, $n_{\mathcal{S}}\psi|_{\mathcal{S}} \in H^{m}_{\textnormal{loc}}(\mathcal{S})$. Moreover, if $\Psi_1$, $\Psi_1'$ and $\Psi_2$, $\Psi_2'$ as above and $\Psi_1=\Psi_2$, $\Psi_1'=\Psi_2'$ in an open set $\mathcal{U}\subset \Sigma$, then $\psi_1=\psi_2$ in $\mathcal{M}\setminus J^{\pm}(\Sigma \setminus \bar{\mathcal{U}})$.
\end{theorem}

The last statement in the theorem above is called the domain of dependence property of the wave equation. In particular, it implies that a solution $\psi$ of (\ref{eq:kgequation}) in the static region $\mathcal{R}\cap J^+(\Sigma)$ of Nariai is independent of the data $\Psi,\Psi'$ on $\Sigma\setminus \Sigma_0$. This follows from the fact that there exist no causal geodesics connecting points in $\Sigma\setminus \Sigma_0$ to points in $\mathcal{R}\cap J^+(\Sigma)$. For this reason, we are able to restrict our analysis to the region $\mathcal{R}$ of Nariai.

Note that in both Nariai and PH we do not impose any decay of $\Psi,\Psi'$ along the initial hypersurface $\Sigma$.

\subsection{Spherical harmonic mode decomposition}
\label{sec:sphericalharmonics}
Since Nariai and PH are spherically symmetric, we can decompose $\psi$ as follows:
\begin{equation*}
\psi(t,x,\theta,\phi)=\sum_{l=0}^{\infty}\sum_{m=-l}^{m=l}\psi_{m,l}(t,x) Y^{m,l}(\theta,\phi),
\end{equation*}
where $Y^{m,l}\in L^2(\s^2)$ are the spherical harmonics, which form an orthonormal basis under the standard inner product on $L^2(\s^2)$. For convenience, we denote by $(t,x,\theta,\phi)$ the static coordinates in the region $\mathcal{R}$ of Nariai and also the standard global coordinates in PH.

We define the harmonic modes by
\begin{equation*}
\psi_l(t,x,\theta,\phi):=\sum_{m=-l}^{m=l}\psi_{m,l}(t,x) Y^{m,l}(\theta,\phi).
\end{equation*}
In particular, since $Y^{0,0}\equiv 1$, $\psi_0$ is a spherically symmetric harmonic mode. Moreover, we denote
\begin{equation*}
\psi_{l\geq L}:=\sum_{l=L}^{\infty} \psi_l.
\end{equation*}

By the properties of spherical harmonics, we have that
\begin{equation*}
\square_g(\psi_{m,l}Y^{m,l})=\left(S \psi_{m,l}-\frac{l(l+1)}{r_0^2}\psi_{m,l}\right)Y^{m,l},
\end{equation*}
where $S$ is an operator on $\mathcal{Q}$. By linear independence of $Y^{m,l}$, we therefore have that $(\square_g-\mu^2) \psi=0$ if and only if $(\square_g-\mu^2)(\psi_{m,l}Y^{m,l})=0$, so in particular, each $\psi_l$ satisfies (\ref{eq:kgequation}) separately. Moreover, for $\psi=\psi_l$, (\ref{eq:kgequation}) reduces to (\ref{eq:kqequationfixedl}) with
\begin{equation*}
m_l^2=r_0^{-2}l(l+1).
\end{equation*}

In the massless $\mu=0$ case, we can find an explicit expression for $\psi_0$ in terms of the spherical means of the initial data,
\begin{align*}
\Psi_0(t,x)&=\frac{1}{\textnormal{Area}\,S_r}\int_{\s^2} \Psi(t,x,\theta,\phi) r^2d\mu_{g_{\s^2}},\\
\Psi_0'(t,x)&=\frac{1}{\textnormal{Area}\,S_r}\int_{\s^2} \Psi'(t,x,\theta,\phi) r^2d\mu_{g_{\s^2}}.
\end{align*}
by using the property that $r=r_0$ is constant on the spacetimes. 

Indeed, in a double-null foliation $(u,v,\theta,\phi)$, we have that
\begin{equation}
\label{spsymkg}
\begin{split}
0&=(\square_g-\mu^2)\psi_0=\frac{1}{\sqrt{-\det g}}\partial_{\alpha}\left(\sqrt{-\det g} g^{\alpha \beta}\partial_{\beta}\psi_0\right)-\mu^2\psi_0\\
&=-4\Omega^{-2}\partial_u\partial_v\psi_0-\mu^2\psi_0.
\end{split}
\end{equation}

When $\mu=0$, (\ref{spsymkg}) is equivalent to the 1+1-dimensional wave equation in $(u,v)$ coordinates. We therefore have that
\begin{equation*}
\partial_u\psi_0(u,v)=\partial_u\psi_0(u,v_{\Sigma}(u)),
\end{equation*}
where $(u,v_{\Sigma}(u))$ is a point on $\Sigma$.

Consequently, in both Nariai and PH, we can write down an explicit expression for $\psi_0$ that is reminiscent of d'Alembert's formula,
\begin{equation}
\label{dalm1}
\psi_0(u,v)=\psi_0(u_{\Sigma}(v),v)+\int_{u_{\Sigma}(v)}^u \partial_u\psi_0(\bar{u},v_{\Sigma}(\bar{u}))\,d\bar u.
\end{equation}

\subsection{The vector field method}
\label{sec:vecbasictools}
Let $V$ be a vector field in a Lorentzian manifold $(\mathcal{M},g)$. We consider the stress-energy tensor $\mathbf{T}[\psi]$ corresponding to (\ref{eq:kgequation}), with components
\begin{equation*}
\mathbf{T}_{\alpha \beta}[\psi]=\partial_{\alpha}\psi \partial_{\beta} \psi-\frac{1}{2}g_{\alpha \beta} \left(\partial^{\gamma} \psi \partial_{\gamma} \psi+\mu^2\psi^2\right).
\end{equation*}
Let $J^V[\psi]$ denote the energy current corresponding to $V$, which is obtained by applying $V$ as a vector field multiplier, i.e.\ in components
\begin{equation*}
J^V_{\alpha}(\psi)=\mathbf{T}_{\alpha \beta}[\psi] V^{\beta}.
\end{equation*}
An energy flux is an integral of $J^V[\psi]$ contracted with the normal to a hypersurface with the natural volume form corresponding to the metric induced on the hypersurface. If the hypersurface is null, the volume form is chosen such that Stokes' theorem holds as below.

Consider a bounded spacetime region $\mathcal{D}$, with a boundary $\partial \mathcal{D}$ that is a union of spacelike hypersurfaces $\mathcal{S}_1$, $\mathcal{S}_2$, timelike hypersurfaces $\mathcal{T}_1$, $\mathcal{T}_2$ and null hypersurfaces $\mathcal{N}_1$ and $\mathcal{N}_2$. We have by Stokes' Theorem that
\begin{figure}[h! t]
\begin{center}
\includegraphics[width=3.0in]{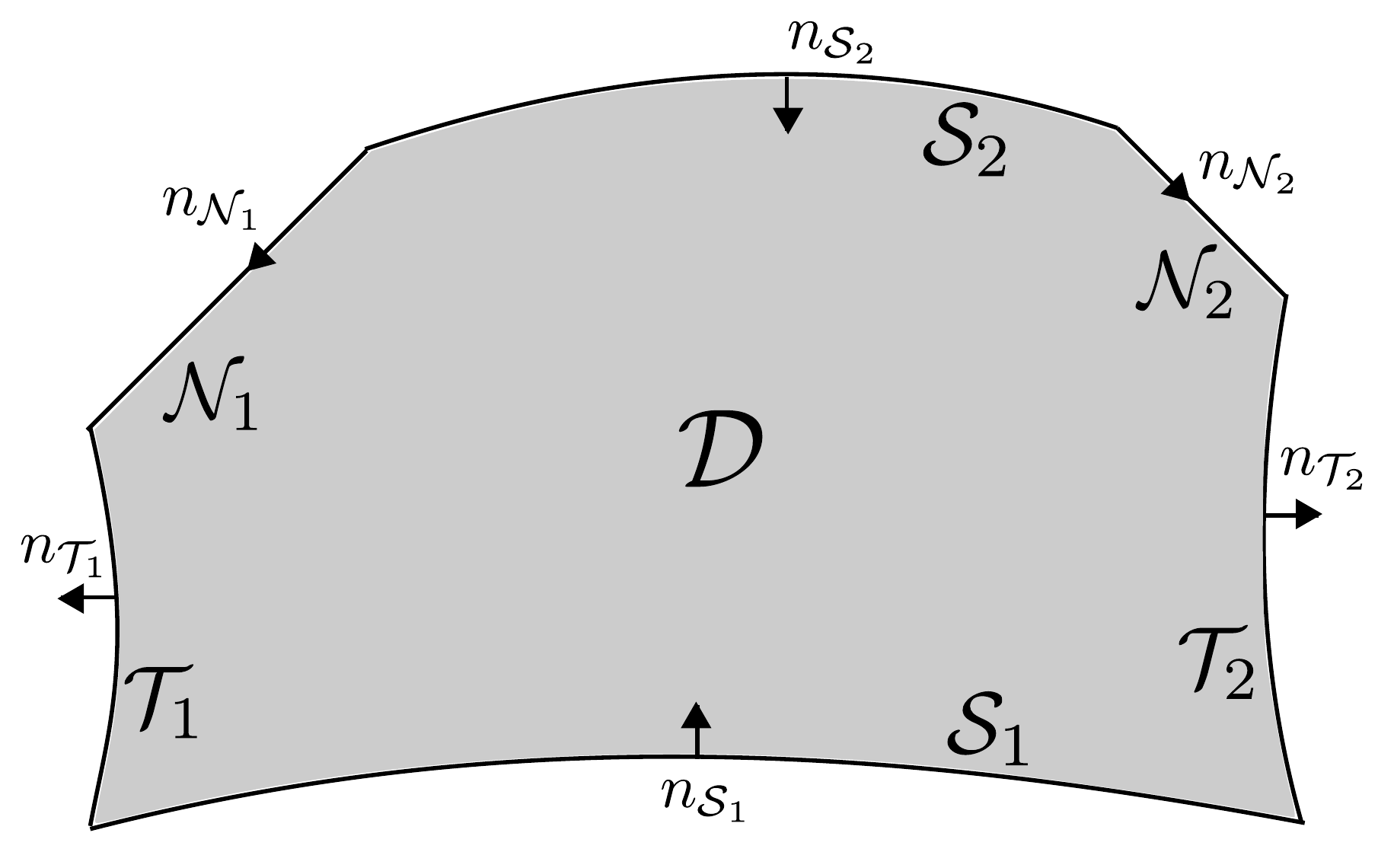}
\caption{A diagrammatic representation of a spacetime region $\mathcal{D}$ with boundary $\partial \mathcal{D}$ and a choice of normal directions $n_{\partial \mathcal{D}}$.}
\end{center}
\end{figure}
\begin{equation*}
\begin{split}
\int_{\mathcal{D}} \textnormal{div}\, J^V[\psi]&=\int_{\mathcal{S}_1}J^V[\psi]\cdot n_{\mathcal{S}_1}+\int_{\mathcal{S}_2}J^V[\psi]\cdot n_{\mathcal{S}_2}+\int_{\mathcal{T}_1}J^V[\psi]\cdot n_{\mathcal{T}_1}+\int_{\mathcal{T}_2}J^V[\psi]\cdot n_{\mathcal{T}_2}\\
&+\int_{\mathcal{N}_1}J^V[\psi]\cdot n_{\mathcal{N}_1}+\int_{\mathcal{N}_2}J^V[\psi]\cdot n_{\mathcal{N}_2},
\end{split}
\end{equation*}
where the null normals $n_{\mathcal{N}_i}$ and the volume form on $\mathcal{N}_i$ are chosen such that the theorem holds as above.

In the language of \cite{chr2} we refer to the divergence term 
\begin{equation*}
K^V[\psi]:=\mathbf{T}^{\alpha \beta}[\psi]\nabla_{\alpha}V_{\beta}=\textnormal{div} J^V[\psi]
\end{equation*}
as the \emph{compatible current} to $J^V[\psi]$.

The vector field method consists of applying Stokes' Theorem with carefully chosen vector fields to suitable spacetime regions. In particular, if $(\mathcal{M},g)$ is a spacetime satisfying the Dominant Energy Condition, $J^V\cdot W$ satisfies the following positivity property:
\begin{lemma}
\label{dec}
Let $V,W$ be future-directed causal vector fields. Then $J^W[\psi]\cdot V\geq 0$ if $(\mathcal{M},g)$ satisfies the dominant energy condition.
\end{lemma}

Moreover, if $(\mathcal{M},g)$ is stationary, with a timelike Killing vector field $T$, $K^T$ vanishes and we obtain conservation of the energy fluxes of $J^T$ with respect to suitable spacelike hypersurfaces.

\section{Main theorems}
\label{sec:maintheorems}

\subsection{Results in Nariai and $dS_n$}
In this section we present the results for (\ref{eq:kgequation}) on a Nariai background and on $n$-dimensional de Sitter space.
\subsubsection{Results in Nariai}
We first consider the Cauchy problem for (\ref{eq:kgequation}) (see Section \ref{sec:cauchykgequation}) in a static region $\mathcal{R}$ in Nariai, introduced in Section \ref{sec:staticcoordinates}, that is covered by the static coordinates $(t,x,\theta,\phi)$.  We foliate $\mathcal{R}$ by the compact spacelike hypersurfaces $\Sigma_{\tau}$, see Section \ref{sec:foliationsnariai}.

The current $J^T[\psi]$ corresponding to the Killing vector field $T=\partial_t$ is defined in Section \ref{sec:vecbasictools}. The energy flux of $J^T$ can be estimated by
\begin{equation*}
J^T[\psi]\cdot n_{\Sigma_{\tau}}\sim (\partial_t\psi)^2+(1-Kx^2)(\partial_x\psi)^2+|\snabla\psi|^2+\mu^2\psi^2,
\end{equation*}
and it degenerates at the cosmological horizons $\{x=\pm \frac{1}{\sqrt{K}}\}$.

In Section \ref{sec:boundn} we consider the timelike local red-shift vector fields $N$ and $\overline{N}$. The corresponding currents $J^N[\psi]$ and $J^{\overline{N}}[\psi]$ are each positive definite and non-degenerate at one of the cosmological horizons, so
\begin{equation*}
J^N[\psi]\cdot n_{\Sigma_{\tau}}+ J^{\overline{N}}[\psi]\cdot n_{\Sigma_{\tau}} \sim (\partial_t\psi)^2+(\partial_x\psi)^2+|\snabla\psi|^2+\mu^2\psi^2.
\end{equation*}
The differences between $N$, $\overline{N}$ and $T$ are denoted $R=N-T$ and $\overline{R}=\overline{N}-T$.

Furthermore, by spherical symmetry of Nariai we can decompose the solutions $\psi$ to (\ref{eq:kgequation}) into spherical harmonic modes $\psi_l$, defined in Section \ref{sec:sphericalharmonics}.

We use the shorthand notation $\Omega^k \psi$, to denote the angular derivatives
\begin{equation*}
\Omega_1^{j_1}\Omega_2^{j_2}\Omega_3^{j_3}(\psi),
\end{equation*}
where $\Omega_i$, $i=1,2,3$ are the generators of the $SO(3)$ symmetry and $j_1+j_2+j_3=k$. Since $\Omega^k$ is a product of Killing vector fields, we can commute $\Omega^k$ with the Klein-Gordon operator.

We also consider the Cauchy problem for (\ref{eq:kgequation}) in the entire Nariai spacetime, where we work in global coordinates $(\tilde{t},\tilde{x},\theta,\phi)$. We foliate the entire spacetime by the spacelike hypersurfaces $\tilde{\Sigma}_{\tau}$, see Section \ref{sec:foliationsnariai}.

The global vector field $\tilde{N}$ is defined by,
\begin{equation*}
\tilde{N}=\frac{1}{\cosh \tilde{t}}\partial_{\tilde{t}},
\end{equation*}
with corresponding energy current
\begin{equation*}
J^{\tilde{N}}\cdot n_{\tilde{\Sigma}_{\tau}}\,d\mu_{\tilde{\Sigma}_{\tau}}\sim \left[(\partial_{\tilde{t}}\psi)^2+\frac{1}{\cosh^2 \tilde{t}}(\partial_{\tilde{x}}\psi)^2+|\snabla\psi|^2+\mu^2\psi^2\right] d\mu_{\s^2}d\tilde{x}.
\end{equation*}
We can moreover commute (\ref{eq:kgequation}) with the globally spacelike Killing vector field $X=\sqrt{K}\partial_{\tilde{x}}$.

All integrals below are with respect to the natural volume form corresponding to the induced metric. The following statements hold in Nariai:

\begin{theorem}[Boundedness in the static region of Nariai]
\label{th:boundn}
There exists a constant $C=C(e,\Lambda,\Sigma)>0$ such that for all $\mu\in \R$
\begin{equation*}
\int_{\Sigma_{\tau}} J^N[\psi]\cdot n_{\Sigma_{\tau}}+J^{\overline{N}}[\psi]\cdot n_{\Sigma_{\tau}}\leq C E_{N}[\psi],
\end{equation*}
where
\begin{align*}
E_{N}[\psi]&:=\int_{\Sigma_{0}} J^N[\psi]\cdot n_{\Sigma_{0}}+J^{\overline{N}}[\psi]\cdot n_{\Sigma_{0}}.
\end{align*}

Additionally, there exists a constant $C=C(\mu,e,\Lambda,\Sigma)>0$ such that
\begin{equation*}
\begin{split}
||\psi||_{L^{\infty}(\Sigma_{\tau})}&\leq C \Big[\sqrt{E_{N}[\psi]+E_N[T\psi]+E_{N}[R \psi]+E_{N}[\overline{R}\psi]}\\
&+||\Psi_0||_{L^{\infty}(\Sigma_0)}+||\Psi_0'||_{L^1(\Sigma_0)}+||\nabla_{\Sigma}\Psi_0||_{L^1(\Sigma_0)}\Big],
\end{split}
\end{equation*}
If $\mu\neq 0$, we can remove the $L^1$ and $L^{\infty}$ norms of $\Psi_0$ and $\Psi_0'$ on the right-hand side of the above inequality and we have that $C=\mu^{-2}\tilde{C}(e,\Lambda,\Sigma)$.
\end{theorem}

Theorem \ref{th:boundn} is proved in Propositions \ref{prop:nondegeboundednessnariai} and \ref{prop:pointboundn}.

\begin{theorem}[Decay in the static region of Nariai]
\label{th:decn}
There exists a constant $C=C(e,\Lambda,\Sigma)>0$ such that for all $\mu \in \R$ and $k\in \N$ we can estimate
\begin{equation*}
\int_{\Sigma_{\tau}} J^{N}[\psi]\cdot n_{\Sigma_{\tau}}+J^{\overline{N}}[\psi]\cdot n_{\Sigma_{\tau}}\leq C\tau^{-k}\sum_{|m|\leq k}E_N[\Omega^m \psi].
\end{equation*}

Additionally, if we restrict $\psi=\psi_{l\geq 1}$ in the $\mu=0$ case, there exists a constant $C_k=C_k(k,e,\Lambda,\Sigma,\mu)>0$, such that for all $\mu \in \R$
\begin{equation*}
||\psi||_{L^{\infty}(\Sigma_{\tau})}\leq C_k\tau^{-\frac{k}{2}}\sqrt{\sum_{|m|\leq k}E_N[\Omega^m \psi]+E_{N}[\Omega^m T\psi]+E_N[\Omega^m R \psi]+E_{N}[\Omega^m \overline{R}\psi]}.
\end{equation*}
If $\mu\neq 0$, we have that $C_k=\mu^{-2}\tilde{C}_k(e,\Lambda,\Sigma)$.

Moreover, there exist constants $C=C(\mu,e,\Lambda,\Sigma)>0$ and $\tilde{c}(e,\Lambda,\Sigma)>0$, such that for $c(l)= \frac{\tilde{c}}{l(l+1)}$ when $l\geq 1$ and $c(0)=\tilde{c}$, we can estimate
\begin{align*}
\int_{\Sigma_{\tau}} J^{N}[\psi_l]\cdot n_{\Sigma_{\tau}}+J^{\overline{N}}[\psi_l]\cdot n_{\Sigma_{\tau}}&\leq e^{-c(l)\tau}E_N[\psi_l],\\
||\psi_l||_{L^{\infty}(\Sigma_{\tau})}&\leq Ce^{-\frac{c(l)}{2}\tau}\sqrt{E_N[\psi_l]+E_{N}[ T\psi_l]+E_N[R \psi_l]+E_{N}[ \overline{R}\psi_l]},
\end{align*}
where in the second estimate we need to take $l\geq 1$ if $\mu=0$ and we have that $C=\mu^{-2}\tilde{C}(e,\Lambda,\Sigma)$, if $\mu\neq 0$.
\end{theorem}

Theorem \ref{th:decn} is proved in Propositions \ref{prop:energydecaynariai}, \ref{prop:singlemodedec} and \ref{prop:pointwisedecnariai}.

\begin{theorem}[Global boundedness in Nariai]
\label{th:globboundn}
For any $\tau\geq 0$, we can estimate
\begin{equation*}
\int_{\tilde{\Sigma}_{\tau}}J^{\tilde{N}}[\psi]\cdot n_{\tau}\leq \int_{\tilde{\Sigma}}J^{\tilde{N}}[\psi]\cdot n_{\tilde{\Sigma}}=:\tilde{E}_N[\psi].
\end{equation*}

Moreover, there exists a constant $C=C(K,\mu)>0$ such that
\begin{equation*}
\begin{split}
\label{eq:globalnariaipointwise}
||\psi||^2_{L^{\infty}(\tilde{\Sigma}_{\tau})}&\leq C \sum_{|k|+|m|\leq 2, |m|\leq 1} \tilde{E}_N[\Omega^m X^k \psi]+C\Big(||\Psi_0||_{L^{\infty}(\tilde{\Sigma}_{\tau_0})}^2\\
&+||\Psi'_0||_{L^1(\tilde{\Sigma}_{\tau_0})}^2+||\partial_{\tilde{x}}\Psi_0||^2_{L^1(\tilde{\Sigma}_{\tau_0})}\Big),
\end{split}
\end{equation*}
where in the $\mu\neq 0$ case $C=\tilde{C}(K)r_0^{-2}\mu^{-2}$, and we can remove the $L^1$ and $L^{\infty}$ norms of $\Psi_0$ and $\Psi_0'$ on the right-hand side of the above inequality.
\end{theorem}

Theorem \ref{th:globboundn} is proved in Proposition \ref{prop:eboundglobalnariai} and Corollary \ref{cor:pointboundglobalnariai}.

\begin{theorem}[Global decay in Nariai]
\label{th:globdecn}
Let $\psi=\psi_{l\geq 1}$ if $\mu=0$. Then there exist constants $C=C(K,r_0,\mu)>0$ and $c=c(K,\mu,r_0)>0$, such that
\begin{equation*}
\label{eq:finaledecglobalnariaithm}
\int_{\tilde{\Sigma}_{\tau}}J^{\tilde{N}}[\psi]\cdot n_{\tau}\leq Ce^{-c\tau}\tilde{E}_N[\psi],
\end{equation*}
where in the case $\mu\neq 0$,  $C=\tilde{C}(r_0,K)e^{\tilde{c}(K)\tau_0}$, with \newline $\tau_0>\max\{1,\log (2\sqrt{2}\sqrt{K}|\mu|^{-1})\}$ and $c=(1+\mu^{-2}+\mu^2)^{-1}\tilde{c}(K,r_0)$.

Moreover,
\begin{equation*}
||\psi||^2_{L^{\infty}(\tilde{\Sigma}_{\tau})}\leq Ce^{-c\tau} \sum_{|k|+|m|\leq 2, |m|\leq 1} \tilde{E}_N[\Omega^m X^k \psi].
\end{equation*}
\end{theorem}

Theorem \ref{th:globdecn} is proved in Corollaries \ref{prop:globedecnariai} and \ref{prop:globedecnariaipointwise}.

\begin{theorem}[Non-decay in Nariai for $\mu=0$]
\label{th:nondecn}
Let $\mu=0$. There exists a constant 
\begin{equation*}
\uline{\psi}=\uline{\psi}(e,\Lambda,\Sigma_0,\Psi_0,\Psi'_0),
\end{equation*}
such that in $\mathcal{R}$
\begin{equation*}
\psi(u,v,\theta,\phi) \to \uline{\psi},
\end{equation*}
as $u\to \infty$ and $v\to \infty$. For generic $\Psi_0$ and $\Psi_0'$, $\uline{\psi}$ is non-vanishing.
\end{theorem}

Theorem \ref{th:nondecn} is proved in Section \ref{sec:nondec}.

\subsubsection{Results in $dS_n$}
Theorem \ref{th:globdecn} applied to $\psi=\psi_0$ implies in particular uniform global decay of solutions to (\ref{eq:kgequation}) with $\mu\neq 0$ in $dS_2$. The results can in fact be generalised to obtain uniform global decay of solutions $\psi$ to (\ref{eq:kgequation}) with $\mu\neq 0$ in $dS_n$, with $n\geq 2$. 

Let $(\tilde{t},\theta_1,\ldots,\theta_{n-2},\phi)$ be global coordinates on $dS_n$, where $\tilde{t}\in \R$ and $(\theta_1,\ldots,\theta_{n-2},\phi)$ are the standard $n$-spherical polar coordinates on $\s^{n-1}$, in which the metric is given by
\begin{equation*}
g_{dS_n}=K^{-1}(-d\tilde{t}^2+\cosh^2\tilde{t} \slashed{g}_{\s^{n-1}}),
\end{equation*}
where $K>0$. See for example \cite{hael}. We now take the global red-shift vector field to be 
\begin{equation*}
\tilde{N}=\frac{1}{\cosh^{(n-1)}\tilde{t}}\partial_{\tilde{t}}.
\end{equation*}

\begin{theorem}[Decay in $dS_n$ for $\mu\neq 0$]
\label{th:decaydsn}
Let $s\in \N$ be the smallest integer satisfying $s>\frac{n-1}{2}$. Then there exists constants $C=C(n,K,\mu)>0$ and $c=c(n,K,\mu)>0$ such that
\begin{equation*}
|\psi|^2(\tilde{t},\theta_1,\ldots,\theta_{n-2},\phi)\leq Ce^{-c\tilde{t}} \sum_{|k|\leq s} \tilde{E}[\Omega^k \psi],
\end{equation*}
where $\Omega_i$, $i=1,\ldots, n$ are the generators of $(n-1)$-spherical symmmetry, $C=\tilde{C}(n,K)e^{\tilde{c}(n,K)\tau_0}$, with $\tau_0>\max\{1,\log (2\sqrt{2}\sqrt{K}\sqrt{\frac{n}{2}}|\mu|^{-1})\}$ and $c=(1+\mu^{-2})^{-1}\tilde{c}(n,K)$ and
\begin{equation*}
\tilde{E}[\psi]:=\int_{\{\tilde{t}=0\}}\sqrt{K}J^{\tilde{N}}[\psi]\cdot \partial_{\tilde{t}}.
\end{equation*}
\end{theorem}

Theorem \ref{th:decaydsn} is proved in Corollary \ref{prop:pointwisedecdsn}.

\subsection{Results in Pleba\'nski-Hacyan}
We also consider the Cauchy problem for (\ref{eq:kgequation}) in PH. We work in $(t,x,\theta,\phi)$ coordinates, where $(t,x)$ are the rectilinear coordinates on $\R^{1+1}$.

It is convenient to consider double-null coordinates $(u,v,\theta,\phi)$, where $u=t-x$ and $v=t+x$. We can express $\Sigma$ as a graph in double-null coordinates
\begin{equation*}
\begin{split}
\Sigma&=\{(u_{\Sigma}(v),v,\theta,\phi),\,|\, v\in \R,\, (\theta,\phi)\in \s^2\}\\
&=\{(u,v_{\Sigma}(u),\theta,\phi),\,|\, u\in \R,\, (\theta,\phi)\in \s^2\}.
\end{split}
\end{equation*}
  
We foliate $\mathcal{M}$ by $\Sigma_{\tau}$, see Section \ref{sec:foliationsph}. The current $J^T[\psi]$ corresponding to the vector field $T=\partial_t$ is defined in Section \ref{sec:vecbasictools}. The energy flux of $J^T$ can be estimated by
\begin{equation*}
J^T[\psi]\cdot n_{\Sigma_{\tau}}\sim (\partial_t\psi)^2+(\partial_x\psi)^2+|\snabla\psi|^2+\mu^2\psi^2.
\end{equation*}

We can foliate the interior of a lightcone $\mathcal{C}$ by hyperboloids $\mathcal{H}_{s}$, see Section \ref{sec:foliationsph}. A natural Killing vector field to employ in the hyperboloidal foliation is the boost vector field $Y=x\partial_t+t\partial_x$.

All integrals below are with respect to the natural volume form corresponding to the induced metric. The following statements hold in  Pleba\'nski-Hacyan:

\begin{theorem}[Boundedness in Pleba\'nski-Hacyan]
\label{th:boundph}
For all $\mu\in \R$ we have that
\begin{equation*}
\int_{\Sigma_{\tau}} J^T[\psi]\cdot n_{\Sigma_{\tau}}\leq E_{PH}[\psi],
\end{equation*}
where
\begin{equation*}
E_{PH}[\psi]:=\int_{\Sigma} J^T[\psi]\cdot n_{\Sigma_{0}}.
\end{equation*}

Additionally, there exists a constant $\tilde{C}=\tilde{C}(\mu,\Lambda,\Sigma)>0$ such that
\begin{equation*}
||\psi||_{L^{\infty}(\Sigma_{\tau})}\leq \tilde{C}\left[\sqrt{E_{PH}[\psi]+E_{PH}[T \psi]}+||\Psi_0||_{L^{\infty}(\Sigma)}+||\Psi_0'||_{L^1(\Sigma)}+||\nabla_{\Sigma}\Psi_0||_{L^1(\Sigma)}\right],
\end{equation*}
If $\mu\neq 0$, we can remove the $L^1$ and $L^{\infty}$ norms of $\Psi_0$ and $\Psi_0'$ on the right-hand side of the above inequality and we have that $C=\mu^{-2}\tilde{C}(\Lambda,\Sigma)$.
\end{theorem}

Theorem \ref{th:boundph} is proved in Proposition \ref{prop:boundednessph}.

\begin{theorem}[Integrated local energy decay]
\label{th:iledph}
For any $x_0>0$, there exists a constant $C=C(\Lambda,\Sigma,x_0)>0$ such that
\begin{equation*}
\begin{split}
\int_{0}^{\infty} \left( \int_{\Sigma_{\tau}\cap\{|x|\leq x_0\}} J^T[\psi]\cdot n_{{\Sigma}_{\tau}}\right)\,d\tau&\leq C \sum_{|k|\leq 1} E_{PH}[\Omega^k\psi].
\end{split}
\end{equation*}
\end{theorem}

Theorem \ref{th:iledph} is proved in Proposition \ref{prop:iledph}

\begin{theorem}[Decay in Pleba\'nski-Hacyan]
\label{th:decph}
There exists a constant $C=C(\mu,\Lambda,\Sigma)>0$ such that
\begin{equation*}
||\psi||_{L^{\infty}(\Sigma_{\tau})}\leq C \tau^{-\frac{1}{2}}\sqrt{\sum_{|k|\leq 1}E_{PH}[\Omega^k\psi]+E_{PH}[\Omega^kY\psi]},
\end{equation*}
where we need to restrict $\psi=\psi_{l\geq 1}$ if $\mu=0$ and we have that $C=\mu^{-2}\tilde{C}(\Lambda,\Sigma)$ if $\mu\neq 0$.
\end{theorem}

Theorem \ref{th:decph} is proved in Proposition \ref{prop:decph}.

\begin{theorem}[Non-decay in Pleba\'nski-Hacyan for $\mu=0$]
\label{th:nondecph}
Let $\mu=0$. On $\mathcal{N}^+_A$, $\psi$ attains the values
\begin{equation*}
 \psi|_{\mathcal{N}^+_A}(u)=\psi_0(u,v_{\Sigma}(u))+\int_{v_{\Sigma}(u)}^{\infty} \partial_v\psi_0(u_{\Sigma}(\bar{v}),\bar{v})\,d\bar{v}.
\end{equation*}
Similarly, on $\mathcal{N}^+_B$ $\psi$ attains the values
\begin{equation*}
 \psi|_{\mathcal{N}^+_B}(v)=\psi_0(u_{\Sigma}(v),v)+\int_{u_{\Sigma}(v)}^{\infty} \partial_u\psi_0(\bar{u},v_{\Sigma}(\bar{u}))\,d\bar{u}.
\end{equation*}
For generic initial $\Psi_0$ and $\Psi_0'$, the above expressions converge to non-zero constants as $u\to \infty$ or $v \to \infty$.
\end{theorem}

Theorem \ref{th:nondecph} is proved in Section \ref{sec:nondec}.

\section{Uniform boundedness in the static region of Nariai}
\label{sec:boundn}
In this section we will show that we can obtain uniform boundedness results in the static region, without making use of the Killing vector field $X$, but by using the Killing vector field $T$ together with the local red-shift along the cosmological horizons bounding $\mathcal{R}$. We will moreover see that we can prove uniform boundedness of solutions to (\ref{eq:kgequation}) independently of an integrated local energy decay statement.
\subsection{Energy boundedness}
First, recall that the control of the $\dot{H}^1(\Sigma)$ norm of $\psi$ by $J^T[\psi]\cdot n_{\Sigma}$ degenerates at the horizons
\begin{equation*}
J^T[\psi]\cdot n_{\Sigma}\sim (\partial_t\psi)^2+(\partial_x\psi)^2+(1-Kx^2)\left(|\snabla\psi|^2+\mu^2\psi^2\right).
\end{equation*}

Moreover, $J^T$ immediately provides uniform boundedness of a degenerate energy after applying Stokes' Theorem in the region $\mathcal{R}(0,\tau)$, 
\begin{equation*}
\int_{\Sigma_{\tau}}J^T[\psi]\cdot n_{\tau}\leq \int_{\Sigma_{0}}J^T[\psi]\cdot n_{0},
\end{equation*}
where we used that $\textnormal{div}\, J^T[\psi]=0$.

In order to obtain a non-degenerate energy boundedness statement, we introduce the red-shift vector fields $N$ and $\overline{N}$. Either $N$ or $\bar{N}$ is timelike at a cosmological horizon. Therefore we can estimate
\begin{equation*}
J^N[\psi]\cdot n_{\Sigma}+J^{\overline{N}}[\psi]\cdot n_{\Sigma}\sim (\partial_t\psi)^2+(\partial_x\psi)^2+|\snabla\psi|^2+\mu^2\psi^2.
\end{equation*}

The theorem below is proved in a more general setting in \cite{dafrod5} for the $\mu=0$ case, but it immediately holds for $\mu\neq 0$.
\begin{theorem}[The Dafermos-Rodnianski Red-Shift Theorem]
\label{thm:redshift}
There exist timelike vector fields $N$ and $\overline{N}$ in the static region $\mathcal{R}$ of Nariai, with positive constants $B_1$ and $B_2$ and bounds $0<x_1<x_0<\sqrt{K}$, such that
\begin{itemize}
\item[(i)]
\begin{align*}
K^N[\psi] &\geq B_1 J^N[\psi]\cdot n_{\tau}\quad x\leq -x_0,\\
K^{\overline{N}}[\psi] &\geq B_1 J^{\overline{N}}[\psi]\cdot n_{\tau}\quad x\geq x_0.
\end{align*}
\item[(ii)]
\begin{align*}
N&=T\quad x\geq -x_1,\\
\overline{N}&=T\quad x \leq x_1.
\end{align*}
\item[(iii)]
\begin{align*}
|K^N[\psi]|&\leq B_2 J^T[\psi]\cdot n_{\tau},  \:\textnormal{and}\:  J^N[\psi]\cdot n_{\tau}\sim J^T[\psi]\cdot n_{\tau} &-x_0&\leq x \leq -x_1,\\
|K^{\overline{N}}[\psi]|&\leq B_2 J^T[\psi]\cdot n_{\tau},  \:\textnormal{and}\:  J^{\overline{N}}[\psi]\cdot n_{\tau}\sim J^T[\psi]\cdot n_{\tau} &x_1&\leq x \leq x_0.
\end{align*}
\end{itemize}
\end{theorem}

The red-shift effect results in a non-degenerate energy boundedness statement.
\begin{proposition}
\label{prop:nondegeboundednessnariai}
There exists a uniform constant $C>0$, such that
\begin{equation}
\label{eq:nondegeboundednessnariai}
\int_{\Sigma_{\tau}}J^N[\psi]\cdot n_{\tau}+J^{\overline{N}}[\psi]\cdot n_{\tau}\leq C \int_{\Sigma_0} J^N[\psi]\cdot n_{0}+J^{\overline{N}}[\psi]\cdot n_{0}.
\end{equation}
\end{proposition}
\begin{proof}
We apply Stokes' theorem to the energy flux with respect to $N$. The $\overline{N}$ case follows analogously. Let $0\leq \tau'<\tau$, then by the Red-shift Theorem
\begin{equation*}
\begin{split}
\int_{\Sigma_{\tau}} J^N\cdot n_{\tau} &\leq \int_{\Sigma_\tau'}J^N\cdot n_{\tau'}-\int_{\mathcal{R}(\tau',\tau)}K^N\\
&\leq \int_{\Sigma_{\tau'}}J^N\cdot n_{{\tau'}}+C \int_{\tau'}^{\tau} \int_{\Sigma_{\bar{\tau}}\cap \{-x_0 \leq x \leq -x_1\}}J^T\cdot n_{\tau}\,d\bar{\tau}\\
&-C \int_{\tau'}^{\tau}\int_{\Sigma_{\bar{\tau}}\cap \{x \leq -x_0\}}J^N\cdot n_{\tau}\,d\bar{\tau}
\end{split}
\end{equation*}
We add $\int_{\tau'}^{\tau}\int_{\Sigma_{\bar{\tau}}} J^N\cdot n_{\bar{\tau}}\,d\bar{\tau}$ to both sides and rearrange the terms to obtain,
\begin{equation*}
\begin{split}
\int_{\Sigma_{\tau}} J^N\cdot n_{\tau}-\int_{\Sigma_{\tau'}} J^N\cdot n_{\tau'} +\int_{\tau'}^{\tau}\int_{\Sigma_{\bar{\tau}}} J^N\cdot n_{\bar{\tau}}\,d\bar{\tau}&\leq C \int_{\tau'}^{\tau}\int_{\Sigma_{\bar{\tau}}} J^T\cdot n_{\bar{\tau}}\,d\bar{\tau}\\
&\leq C(\tau-\tau')\int_{\Sigma_{\tau'}} J^T\cdot n_{\tau'},
\end{split}
\end{equation*}
where we used the degenerate energy bound in the last equality. The above inequality is of the form,
\begin{equation*}
f(t)-f(t')+\int_{t'}^t f(s)\,ds\leq C (t-t')D_0.
\end{equation*}
By dividing by $t-t'$ and taking the limit $t\to t'$ we obtain the differential inequality
\begin{equation*}
\frac{d}{dt}(tf(t))\leq CD_0
\end{equation*}
and therefore $f(t)\leq CD_0$. We can conclude that
\begin{equation*}
\int_{\Sigma_{\tau}}J^N[\psi]\cdot n_{\tau}\leq C \int_{\Sigma_0} J^N[\psi]\cdot n_{0}.\qedhere
\end{equation*}
\end{proof}

\subsection{Pointwise boundedness}
We use standard Sobolev and elliptic estimates on $\Sigma_{\tau}$ to obtain pointwise estimates from the energy estimate in Proposition \ref{prop:nondegeboundednessnariai}.
\begin{proposition}
\label{prop:pointboundn}
In Nariai there exists a uniform constant $C>0$ such that
\begin{equation*}
\begin{split}
||\psi||_{L^{\infty}(\Sigma_{\tau})}&\leq C \bigg(||\Psi_0||_{L^{\infty}(\Sigma)}+||\Psi'_0||_{L^{1}(\Sigma)}+||\partial_x\Psi_0||_{L^1(\Sigma)}\\
&+\sqrt{E_{N}[\psi]+E_N[T\psi]+E_N[R\psi]+E_N[\overline{R}\psi]}\bigg),
\end{split}
\end{equation*}
where
\begin{equation*}
E_{N}[\psi]= \int_{\Sigma_0}\left\{J^N[\psi]\cdot n_{0}+J^{\overline{N}}[\psi]\cdot n_{0}\right\}.
\end{equation*}
If $\mu\neq 0$, we can leave out the norms of $\Psi_0$ and $\Psi'_0$ everywhere above.
\end{proposition}
\begin{proof}
By compactness of $\Sigma_0$, all metric components and their derivatives are bounded, which justifies the use of the following elliptic estimate, see for example the Appendix of \cite{are1} for a derivation,
\begin{equation}
\label{est:standardelliptic}
||g^{ij}(\nabla^2\psi)_{ij}||_{L^2(\Sigma_0)}^2\leq C \int_{\Sigma_0}\left\{ J^N[\psi]\cdot n_{\Sigma}+J^N[N\psi]\cdot n_{\Sigma}+J^{\overline{N}}[\overline{N}\psi]\cdot n_{\Sigma}\right\},
\end{equation}
where $i,j$ denote components along $\Sigma$. We can in fact conclude that
\begin{equation*}
\begin{split}
||\psi||^2_{\dot{H}^1(\Sigma_0)}+||\psi||^2_{\dot{H}^2(\Sigma_0)} &\leq ||\psi||^2_{\dot{H}^1(\Sigma_0)}+C ||g^{ij}(\nabla^2\psi)_{ij}||_{L^2(\Sigma_0)}^2\\
&\leq C \int_{\Sigma_0}\left\{ J^N[\psi]\cdot n_{\Sigma}+J^N[N\psi]\cdot n_{\Sigma}+J^{\overline{N}}[\overline{N}\psi]\right\}.
\end{split}
\end{equation*}

Together with the Sobolev inequality in Proposition \ref{sobs2} and the Poincar\'e inequality in Proposition \ref{poins2}, we get for $\mu=0$ and $\psi=\psi_{\geq1}$
\begin{equation*}
||\psi_{\geq1}||_{L^{\infty}(\Sigma_0)}^2\leq C  \int_{\Sigma_0}\left\{ J^N[\psi]\cdot n_{\Sigma}+J^N[N\psi]\cdot n_{\Sigma}+J^{\overline{N}}[\overline{N}\psi]\right\},
\end{equation*}
where $C>0$ is a uniform constant. 

If $\mu \neq 0$, we no longer need Proposition \ref{poins2} and we find that 
\begin{equation*}
||\psi ||_{L^{\infty}(\Sigma_0)}^2\leq C \int_{\Sigma_0}\left\{ J^N[\psi]\cdot n_{\Sigma}+J^N[N\psi]\cdot n_{\Sigma}+J^{\overline{N}}[\overline{N}\psi]\right\}.
\end{equation*}
By construction, $\Sigma_{\tau}$ is isometric to $\Sigma_0$ for all $\tau\geq 0$, so the above estimates holds with $\Sigma_0$ replaced by $\Sigma_{\tau}$ and with $C$ unchanged. We can moreover estimate
\begin{equation*}
J^N[\psi]\cdot n_{\Sigma}\leq C \left[J^N[T\psi]\cdot n_{\Sigma}+J^N[R\psi]\cdot n_{\Sigma}\right],
\end{equation*}
where $R=N-T$ and $C>0$ is a uniform constant. Since $[\square_g,T]=0$, by the Killing property of $T$, the estimates in Proposition \ref{prop:nondegeboundednessnariai} hold for $T\psi$ replacing $\psi$.

However, when commuting $\square_g$ with $Y$, we have to estimate an additional spacetime integral of the error term $R(\psi)\square_g(R \psi)$, when replacing $\psi$ with $R\psi$ in Proposition \ref{prop:nondegeboundednessnariai}. In \cite{dafrod5} it is shown in a very general setting that the error terms in the region $\{|x|\leq x_0\}$ can be absorbed by the remaining spacetime integrals of $K^N[R\psi]$, $K^N[\psi]$ and $K^N[T\psi]$ if $x_0$ is suitably large, relying on the positivity of the surface gravity to infer that the term proportional to $(R^2\psi)^2$ (that does not come with a smallness constant) has a favourable sign. We can therefore conclude that 
\begin{equation*}
\int_{\Sigma_{\tau}}J^N[N\psi]\cdot n_{\tau}\leq C \int_{\Sigma_0} J^N[T\psi]\cdot n_{0}+J^N[R\psi]\cdot n_{0},
\end{equation*}
where $C>0$ is a uniform constant. A similar estimate holds for $\overline{N}$ replacing $N$ and $\overline{R}=\overline{N}-T$.

We can conclude in the case that $\mu\neq 0$,
\begin{equation}
\label{ellnai1}
\begin{split}
||\psi||_{L^{\infty}(\Sigma_{\tau})}^2&\leq C\int_{\Sigma_{\tau}}J^N[\psi]\cdot n_{{\tau}}+J^{\overline{N}}[\psi]\cdot n_{\tau} +J^N[N \psi]\cdot n_{{\tau}}+J^{\overline{N}}[\overline{N} \psi]\cdot n_{{\tau}} \\
&\leq C(E_N[\psi]+E_N[T\psi]+E_N[R\psi]+E_N[\overline{R}\psi]).
\end{split}
\end{equation}

Similarly, for $\mu=0$ we conclude
\begin{equation}
\label{ellnai}
\begin{split}
||\psi_{\geq1}||_{L^{\infty}(\Sigma_{\tau})}^2\leq C(E_N[\psi]+E_N[T\psi]+E_N[R\psi]+E_N[\overline{R}\psi])
\end{split}
\end{equation}
where we used the orthonormality of harmonic modes:
\begin{equation*}
\int_{\s^2}\psi_l\psi_l'= \delta_{ll'}\int_{\s^2}\psi_l^2,
\end{equation*}
where $\delta_{l l'}$ is the Kronecker delta.

We are left with estimating $|\psi_0|$ in the $\mu=0$ case. By (\ref{dalm1}) we find that
\begin{equation*}
||\psi_0||_{L^{\infty}(\Sigma_{\tau})}^2\leq C\left[||\Psi_0||_{L^{\infty}(\Sigma_{0})}^2+||\Psi'_0||_{L^{1}(\Sigma_{0})}^2+||\partial_x\Psi_0||^2_{L^1(\Sigma_0)}\right],
\end{equation*}
where $C>0$ depends on the choice of $\Sigma$.
Consequently, we can estimate
\begin{equation}
\begin{split}
||\psi||_{L^{\infty}(\Sigma_{\tau})}^2&\leq C \bigg(||\Psi_0||_{L^{\infty}(\Sigma_{0})}^2+||\Psi'_0||_{L^{1}(\Sigma_{0})}^2+||\partial_x\Psi_0||^2_{L^1(\Sigma_0)}\\
&+E_{N}[\psi]+E_N[T\psi]+E_{N}[R\psi]+E_{N}[\overline{R}\psi]\bigg). \qedhere
\end{split}
\end{equation}
\end{proof}
We have now proved Theorem \ref{th:boundn}.

\section{Uniform decay in the static region of Nariai}
\label{sec:decn}
We obtain energy decay in $\mathcal{R}$ by first showing that an integrated local energy decay statement holds.
\subsection{Integrated energy decay}
Recall from Section \ref{sec:staticcoordinates} that
\begin{equation*}
x_*(x)=\frac{1}{\sqrt{K}}\arctanh \left( \sqrt{K}x\right).
\end{equation*}
The metric in the region $\mathcal{R}$ can be expressed in $(t,x_*,\theta,\phi)$ coordinates,
\begin{equation*}
g=D(-dt^2+dx_*^2)+r_0^2\slashed{g}_{\s^2},
\end{equation*}
where $D(x_*):=(1-Kx^2)$.

The compatible current $K^V$ to $J^V$, with $V=f(x_*)\partial_{x_*}$, is given by
\begin{equation*}
K^V[\psi]=\frac{1}{2}(1-Kx^2)^{-1}f'\left[(\partial_t\psi)^2 +(\partial_{x_*}\psi)^2\right]+\left[Kx f-\frac{1}{2}f'\right]\left(|\snabla \psi|^2+\mu^2\psi^2\right).
\end{equation*}
We need to have $f'>0$ in order to control $(\partial_t\psi)^2$ and $(\partial_{x_*}\psi)^2$. In that case, we necessarily lose control of the remaining terms at the photon sphere $x=0$. We therefore need to modify the current $J^V$.

Consider the first modified current,
\begin{equation*}
J^{V,1}_{\alpha}:=J^V_{\alpha}+\frac{f'}{2}\psi\partial_{\alpha}\psi-\frac{1}{4}\psi^2\partial_{\alpha}f'.
\end{equation*}
The corresponding first modified compatible current is given by,
\begin{equation*}
\begin{split}
K^{V,1}:&=\nabla^{\alpha}J^{V,1}_{\alpha}=K^V+\frac{1}{2}(\partial^{\alpha}f')\psi\partial_{\alpha}\psi+\frac{f'}{2}\partial^{\alpha}\psi\partial_{\alpha}\psi+\frac{f'}{2}\psi \square_g\psi\\
&-\frac{1}{2}\psi\partial^{\alpha}\psi \partial_{\alpha}f'-\frac{1}{4}\psi^2\square_g f'\\
&=K^V+\frac{f'}{2}g^{\alpha \beta}\partial_{\alpha}\psi\partial_{\beta}\psi-\frac{1}{4D}f''' \psi^2+\frac{\mu^2}{2}f'\psi^2\\
&=\frac{1}{D}f'(\partial_{x_*}\psi)^2+Kxf\left(|\snabla \psi|^2+\mu^2\psi^2\right)-\frac{1}{4D}f''' \psi^2.
\end{split}
\end{equation*}
We now have potential non-negativity of the terms above, for a suitable choice of $f$. However, we have lost control over $(\partial_t\psi)^2$ and the control over $|\snabla\psi|^2$ still degenerates at $x=0$. This is a manifestation of the photon sphere $\{x=0\}$ containing trapped null geodesics. 

Observe that
\begin{equation}
\label{eq:poincares2fixedl}
\int_{\s^2}|\snabla \psi_l|^2\,d\mu_{\s^2}=\frac{l(l+1)}{r_0^2}\int_{\s^2} \psi_l^2\,d\mu_{\s^2}.
\end{equation}
We will first restrict to single harmonic modes $\psi=\psi_l$ with $l\geq 1$, in order to control $|\snabla \psi_l|^2$ by $\psi_l^2$ at the cost of a factor $l(l+1)$, which will result in the loss of an angular derivative in the final estimate. 

In the regime $|x|\geq x_0$ we can use the compatible red-shift currents $K^N$ and $K^{\overline{N}}$ to control badly signed error terms. It is essential that these error terms can be taken arbitrarily small compared to the terms with a good sign in the region $x_1\leq |x|\leq x_0$, as the red-shift currents $K^N$ and $K^{\overline{N}}$ fail to be non-negative definite in this region. In order to control boundary currents, we will also cut off $f$ at a large $x_*$, so that $J^{V,1}$ vanishes for suitably large $x_*$.
\begin{lemma}
\label{prop:mainiledn}
Restrict $l\geq 1$. Let $\alpha>x_*(x_0)$ be a suitably large constant. Then there exists a function $f: \R \to \R$ such that the corresponding current $K^{V,1}$ is non-negative when $|x_*|\leq \alpha$ and vanishes for $|x_*|>\frac{e}{2}\alpha$. 

In the region $|x|\leq x_0$ we can control
\begin{equation*}
\int_{\s^2}K^{V,1}[\psi_l] D d\mu_{\s^2} \geq C\alpha^{-3}\int_{\s^2} (\partial_{x}\psi_l)^2+\frac{1}{l(l+1)}|\snabla \psi_l|^2+\mu^2\psi_l^2 d\mu_{\s^2},
\end{equation*}
where $C>0$ is a constant independent of $l$ and $\alpha$.

In the region $x_1\leq|x|\leq x_0$ we can moreover control
\begin{equation*}
\int_{\s^2}K^{V,1}[\psi_l] D d\mu_{\s^2} \geq C\alpha^{-1}\int_{\s^2} (\partial_{x}\psi_l)^2+|\snabla \psi_l|^2+\mu^2\psi_l^2 d\mu_{\s^2},
\end{equation*}
where $C>0$ is a constant independent of $l$ and $\alpha$.

Furthermore, for $|x_*|\in [\alpha,\frac{e}{2}\alpha]$, we can control
\begin{equation*}
\int_{\s^2} K^{V,1}[\psi_l] D d\mu_{\s^2}\geq\ -C\alpha^{-3}\int_{\s^2}  (\partial_{x}\psi_l)^2+|\snabla \psi_l|^2+\mu^2\psi_l^2 d\mu_{\s^2},
\end{equation*}
where $C>0$ is a constant independent of $l$ and $\alpha$.

If $l=0$, the above estimates hold without the $|\snabla \psi_l|^2$ term.
\end{lemma}
\begin{proof}
We define $\tilde{f}:(-e\alpha,\infty)\to \R$ by
\begin{equation*}
\tilde{f}(x_*)=\left(\frac{x_*}{\alpha}+e\right)\log\left(\frac{x_*}{\alpha}+e\right)-\left(\frac{x_*}{\alpha}+e\right).
\end{equation*}
Then $\tilde{f}(x_*)\geq 0$ for $x_*\geq 0$ and $\tilde{f}(x_*)\leq 0$ for $-e\alpha<x_*<0$, so $x \tilde{f}(x_*)\geq 0$ for $x_*>-e\alpha$.

Moreover,
\begin{align*}
\tilde{f}'(x_*)&=\alpha^{-1}\log\left(\frac{x_*}{\alpha}+e\right),\\
\tilde{f}''(x_*)&= \alpha^{-1}(x_*+e\alpha)^{-1},\\
\tilde{f}'''(x_*)&=-\alpha^{-1} \left(x_*+e\alpha\right)^{-2}.
\end{align*}
In particular, $\tilde{f}'>0$ and $\tilde{f}'''<0$ for $x_*>-e\alpha$.

Let $\chi: \R \to \R$ be the function defined by
\begin{equation*}
\chi(z)=(\mathbf{1}_{\geq \alpha}*\eta_{\alpha,R})(z),
\end{equation*}
where $\mathbf{1}_{z\geq \alpha}$ is an indicator function and $\eta_{\alpha,R}(z):=R^{-1}\eta(\frac{z-\alpha}{R})$, where $\eta$ is the bump function defined in Lemma \ref{lmbump}. Take $R:=\frac{e\alpha}{2}-\alpha>0$. Therefore, $\chi$ is a cut-off function satisfying $\chi(|x_*|)=1$ for $|x_*|\leq\alpha$ and $\chi(x_*)=0$ for $|x_*|\geq \frac{e}{2}\alpha$.

Moreover, we can estimate
\begin{equation}
\label{est:cutoffest}
\sup_{x\in \R} |\chi^{(k)}(x)|\leq R^{-k}\sup_{y\in (0,1)}|\eta^{(k)}(y)|\leq C_k R^{-k}\leq C_k\alpha^{-k}.
\end{equation}

We define $f: \R \to \R$ by $f(x_*)=\tilde{f}(x_*)\chi(|x_*|)$ for $|x_*|\leq \frac{e}{2}\alpha$ and $f(x_*)=0$ for $|x_*|\geq\frac{e}{2}\alpha$. By applying the chain rule when differentiating $f$, we obtain
\begin{equation*}
\begin{split}
DK^{V,1}&=\chi\left[\tilde{f}'D^2(\partial_{x}\psi)^2+KxD\tilde{f}\left(|\snabla \psi|^2+\mu^2\psi^2\right)-\frac{\tilde{f}'''}{4}\psi^2\right]+D^2\tilde{f}\chi' (\partial_x\psi)^2\\
&-\frac{1}{4}\left[3\tilde{f}''\chi'+3\tilde{f}'\chi''+\tilde{f}\chi'''\right]\psi^2.
\end{split}
\end{equation*}
From the above expression, it follows that in the region $|x|\leq x_0$, we can estimate
\begin{equation*}
K^{V,1} \geq C\alpha^{-3}\left[(\partial_x\psi_l)^2+\psi_l^2\right],
\end{equation*}
where $C=C(x_0,K,r_0,\mu)>0$ is a constant that is independent of $\alpha$. 

Away from $x=0$, in the region $x_1\leq |x|\leq x_0$ we can in fact obtain an estimate that is less degenerate in $\alpha$,
\begin{equation*}
K^{V,1} \geq C\alpha^{-1}\left[(\partial_x\psi_l)^2+|\snabla \psi_l|^2+\mu^2\psi_l^2\right].
\end{equation*}

Now consider the region $\alpha<|x_*|<\frac{e \alpha}{2}$. 

The terms multiplied by $\chi$ are non-negative. Moreover, we can estimate in the region $\alpha\leq |x_*|\leq \frac{e\alpha}{2}$\begin{equation*}
|D^2\tilde{f}|\leq C(1-\tanh^2(\sqrt{K} \alpha))\leq C \alpha^{-3}.
\end{equation*}
Together with the bounds for the derivatives of $\chi$ in (\ref{est:cutoffest}), we can therefore estimate in the region $\alpha\leq |x_*|\leq \frac{e\alpha}{2}$,
\begin{equation*}
K^{V,1}\geq -C\alpha^{-3}\left[(\partial_{x}\psi_l)^2+|\snabla\psi_l|^2+\psi_l^2\right].
\end{equation*}

We integrate over $\s^2$ and apply (\ref{eq:poincares2fixedl}) to obtain the statements in the lemma.
\end{proof}

In order to estimate $(\partial_t\psi_l)^2$ we will use the fact that we have control over $|\snabla\psi_l|^2$ by the proposition above. We consider a different vector field multiplier $W=g(x_*)\partial_{x_*}$ and use the unmodified compatible current. 

\begin{lemma}
\label{prop:auxiledn} 
In the region $|x|\leq x_0$ there exists a function $g:\R\to \R$ such that we can estimate for all $l\geq 0$
\begin{equation*}
\int_{\s^2}(\partial_t \psi_l)^2+(\partial_{x*}\psi_l)^2\,d\mu_{\s^2}\leq C \int_{\s^2} |\snabla\psi_l|^2+\mu^2\psi_l^2\,d\mu_{\s^2}+C\int_{\s^2}K^W[\psi_l] D\,d\mu_{\s^2},
\end{equation*}
where $C=C(x_0)>0$. Moreover, for $x_0$ suitably large $K^W$ is positive definite when $|x|\geq x_0$.
\end{lemma}
\begin{proof}
Define
\begin{align*}
g(x_*)&=\alpha^{-1}\arctan (\alpha^{-1}x_*),\\
g'(x_*)&=\frac{1}{x_*^2+\alpha^2}.
\end{align*}
Then
\begin{equation*}
K^W[\psi]=\frac{1}{2}D^{-1}g'\left[(\partial_t\psi)^2 +(\partial_{x_*}\psi)^2\right]+\left[Kx g-\frac{1}{2}g'\right]\left(|\snabla \psi|^2+\mu^2\psi^2\right).
\end{equation*}
We use that $g'>0$ is bounded for all $x_*$, together with Lemma \ref{prop:mainiledn}, to arrive at the required estimate. Moreover, for $x_0$ and $\alpha$ suitably large, $K^W\geq 0$ outside $-x_0\leq x \leq x_0$, because the $g$ term dominates the $g'$ term.
\end{proof}

\begin{proposition}
There exists a  uniform constant $C>0$, such that we can bound for all $l\geq 1$
\begin{equation*}
\begin{split}
&\int_{\tau'}^{\tau} \int_{\Sigma_{\bar{\tau}}\cap\{-x_0\leq x \leq x_0\}}\left(xf(x)+\frac{1}{l(l+1)}\right)\left[(\partial_t \psi_l)^2+(\partial_{x*}\psi_l)^2+|\snabla\psi_l|^2\right]+\mu^2\psi_l^2\,d\bar{\tau}\\
&\leq C\int_{\mathcal{R}(\tau',\tau)}K^{V,1}[\psi_l]+K^W[\psi_l]+ C\int_{\mathcal{R}(\tau',\tau)} K^N[\psi_l]+K^{\overline{N}}[\psi_l].
\end{split}
\end{equation*}

In the $l=0$ case, the above estimate holds without the factor $\left(xf(x)+\frac{1}{l(l+1)}\right)$.
\end{proposition}
\begin{proof}
Combining the estimates of Lemma \ref{prop:mainiledn} and Lemma \ref{prop:auxiledn} gives in particular
\begin{equation}
\label{est:crucialiledestimate}
\begin{split}
&\int_{\tau'}^{\tau} \int_{\Sigma_{\bar{\tau}}\cap\{-x_0\leq x \leq x_0\}}\left[xf(x)+\frac{1}{l(l+1)}\right]\left[(\partial_t \psi_l)^2+(\partial_{x*}\psi_l)^2+|\snabla\psi_l|^2+\mu^2\psi_l^2\right]\,d\bar{\tau}\\
&+\alpha^2\int_{\tau'}^{\tau} \int_{\Sigma_{\bar{\tau}}\cap\{x_1\leq |x| \leq x_0\}}(\partial_t \psi_l)^2+(\partial_{x*}\psi_l)^2+|\snabla\psi_l|^2+\mu^2\psi_l^2\,d\bar{\tau}\\
&\leq C\alpha^3 \int_{\mathcal{R}(\tau',\tau)\cap\{ |x_*|\leq \alpha\}}K^{V,1}+\alpha^{-3}K^W,
\end{split}
\end{equation}
where the factor $D$ appearing in front of $K^{V,1}$ and $K^W$ on the right hand side is absorbed into the volume form, $\sqrt{\det g}=Dr_0^2\sin\theta$.

Note that the factor $\frac{1}{l(l+1)} \to 0$ as $l\to \infty$, so the estimate degenerates in the high angular frequency limit at $x=0$. If our integration domain excludes a region around the photon sphere $x=0$, we can in fact drop the $\frac{1}{l(l+1)}$ factor.

Since we want to bound the right-hand side by suitable boundary currents, we would like to extend the integration domain on the right-hand side to the entire region $\mathcal{R}(\tau',\tau)$. As a result of the estimates in the region $\alpha\leq |x|\leq \frac{e}{2}\alpha$ in Lemma \ref{prop:mainiledn} and Lemma \ref{prop:auxiledn}, we have that
\begin{equation}
\label{est:preiledn1}
\begin{split}
 \alpha^3\int_{\mathcal{R}(\tau',\tau)\cap\{|x_*|\leq \alpha\} }K^{V,1} &\leq \alpha^3 \int_{\mathcal{R}(\tau',\tau)}K^{V,1}\\
 &+C\int_{\tau'}^{\tau} \int_{\Sigma_{\bar{\tau}}\cap\{-\frac{e}{2}\alpha<x_*\leq \alpha\}} J^N\cdot n_{\tau}\,d\bar{\tau}\\
 &+C\int_{\tau'}^{\tau} \int_{\Sigma_{\bar{\tau}}\cap\{\alpha<x_*\leq \frac{e}{2}\alpha\}} J^{\overline{N}}\cdot n_{\tau}\,d\bar{\tau}.
 \end{split}
\end{equation}
By the Red-Shift Theorem, we can further estimate
\begin{equation*}
\begin{split}
\int_{\tau'}^{\tau} \int_{\Sigma_{\bar{\tau}}\cap\{-\frac{e}{2}\alpha<x_*<-\alpha\}} J^N\cdot n_{\tau}\,d\bar{\tau} &\leq \frac{1}{B_1}\int_{\mathcal{R}(\tau',\tau)} K^N\\
&+ \frac{B_2}{B_1}\int_{\tau'}^{\tau} \int_{\Sigma_{\bar{\tau}}\cap\{-x_0<x<-x_1\}} J^N\cdot n_{\tau}\,d\bar{\tau}.
\end{split}
\end{equation*}
We can absorb the $ J^N\cdot n_{\tau}$ integral into the second integral on the left-hand side of (\ref{est:crucialiledestimate}) by taking $\alpha>0$ suitably large. We can estimate the error terms in $[\alpha,\frac{e}{2}\alpha]$ analogously, employing $J^{\overline{N}}\cdot n_{\tau}$.

We can now conclude that
\begin{equation*}
\begin{split}
&\int_{\tau'}^{\tau} \int_{\Sigma_{\bar{\tau}}\cap\{-x_0\leq x \leq x_0\}}\left[xf(x)+\frac{1}{l(l+1)}\right]\left[(\partial_t \psi_l)^2+(\partial_{x*}\psi_l)^2+|\snabla\psi_l|^2+\mu^2\psi_l^2\right]\,d\bar{\tau}\\
&\leq C\int_{\mathcal{R}(\tau',\tau)}K^{V,1}+K^W+ C\int_{\mathcal{R}(\tau',\tau)} K^N+K^{\overline{N}}. \qedhere
\end{split}
\end{equation*}
\end{proof}

\begin{proposition}
\label{prop:iledn}
The following spacetime estimate holds,
\begin{equation*}
\int_{\tau'}^{\tau} \int_{\Sigma_{\bar{\tau}}\cap\{-x_0\leq x \leq x_0\}} J^N[\psi]+J^{\overline{N}}[\psi]\cdot n_{\tau}\,d\bar{\tau}\leq C\sum_{|k|\leq 1} \int_{\Sigma_{\tau'}} J^N[\Omega^k\psi]+J^{\overline{N}}[\Omega^k\psi]\cdot n_{\tau'},
\end{equation*}
where $C>0$ is a uniform constant.
\end{proposition}
\begin{proof}
By boundedness of the function $g$, we can estimate for $n$ a causal normal vector field,
\begin{equation*}
|J^W\cdot n|\leq C (J^N\cdot n+J^{\overline{N}}\cdot n).
\end{equation*}
Moreover, by construction $J^{V,1}$ vanishes for $x_*$ outside a bounded interval, so $J^{V,1}\cdot n=0$, for $n=n_{\mathcal{C}^+}$ and $n=n_{\overline{\mathcal{C}}^+}$.  Moreover, by (\ref{eq:poincares2fixedl}) we can also estimate
\begin{equation*}
|J^{V,1}\cdot n_{\tau}|\leq C J^T\cdot n_{\tau},
\end{equation*}
if $\mu\neq 0$. If $\mu=0$, this estimate holds only for $l\geq 1$. In the $l=0$ case there is no need for the current $J^{V,1}$ and we can control all derivatives by solely employing $J^W$ and $K^W$.

Together with the non-degenerate energy boundedness statement we conclude by Stokes' theorem that
\begin{equation*}
\int_{\mathcal{R}(\tau',\tau)}K^{V,1}+K^W+ C\int_{\mathcal{R}(\tau',\tau)} K^N+K^{\overline{N}} \leq  C\int_{\Sigma_{\tau'}} J^N[\psi_l]\cdot n_{\tau'}+J^{\overline{N}}[\psi_l]\cdot n_{\tau'}.
\end{equation*}
In order to obtain an $l$ independent estimate, we apply (\ref{eq:poincares2fixedl}) once more. Moreover, Lemma \ref{prop:auxiledn} gives the required estimate for $l=0$. This implies that for all $l$
\begin{equation*}
\int_{\tau'}^{\tau} \int_{\Sigma_{\bar{\tau}}\cap\{-x_0\leq x\leq x_0\}} J^N[\psi_l]+J^{\overline{N}}[\psi_l]\cdot n_{\tau}\,d\bar{\tau}\leq C\sum_{|k|\leq 1}\int_{\Sigma_{\tau'}} J^N[\Omega^k\psi_l]+J^{\overline{N}}[\Omega^k\psi_l]\cdot n_{\tau'},
\end{equation*}
By the orthonormality property of $\psi_l$ and the independence of $l$ in the uniform constants in the above estimate, the estimate holds for general $\psi$.
\end{proof}

\begin{corollary}
\label{cor:iednfin}
The following spacetime estimate holds
\begin{equation}
\label{est:iednfin}
\int_{\tau'}^{\tau} \int_{\Sigma_{\bar{\tau}}} J^N[\psi]\cdot n_{\tau}+J^{\overline{N}}[\psi]\cdot n_{\tau}\,d\bar{\tau}\leq C\sum_{|k|\leq 1} \int_{\Sigma_{\tau'}} J^N[\Omega^k\psi]\cdot n_{\tau'}+J^{\overline{N}}[\Omega^k\psi]\cdot n_{\tau'},
\end{equation}
where $C>0$ is a uniform constant.
\end{corollary}
\begin{proof}
This follows directly from Proposition \ref{prop:iledn} combined with the Red-Shift Theorem in the regions $\{x\geq x_0\}$ and $\{x<-x_0\}$.
\end{proof}

\subsection{Energy decay}
The integrated energy decay statement (\ref{est:iednfin}) implies energy decay by an application of the pigeonhole principle.
\begin{proposition}
\label{prop:energydecaynariai}
\begin{equation*}
\int_{\Sigma_{\tau}} J^N[\psi]\cdot n_{\tau}+ J^{\overline{N}}[\psi]\cdot n_{\tau} \leq C_k \tau^{-k} \sum_{|m|\leq k}E_{N}[\Omega^m \psi],
\end{equation*}
where $C_k>0$ is a uniform constant that depends on $k\in \N$.
\end{proposition}
\begin{proof}
We will first prove the following statement by induction: there exists a dyadic sequence $\{\tau_j\}$ such that
\begin{equation*}
\int_{\Sigma_{\tau_j}} J^N[\psi]\cdot n_{\tau_j}+J^{\overline{N}}[\psi]\cdot n_{\tau_j} \leq C \tau_j^{-k} \sum_{|m|\leq k}\int_{\Sigma_0}J^N[\Omega^m\psi]\cdot n_{0}+J^{\overline{N}}[\Omega^m\psi]\cdot n_{0},
\end{equation*}
where $C=C(k)$.

The $k=0$ case trivially holds by the non-degenerate energy bound. Now suppose there exists a dyadic sequence $\{\tau_j\}$, such that the above statement holds for all $m\leq k$. By (\ref{est:iednfin}) we can estimate
\begin{equation*}
\begin{split}
\int_{\tau_{j}}^{\tau_{j+1}} \int_{\Sigma_{\bar{\tau}}} J^N[\psi]\cdot n_{\tau}+J^{\overline{N}}[\psi]\cdot n_{\tau}\,d\bar{\tau}&\leq C_0\sum_{|m|\leq 1} \int_{\Sigma_{\tau_j}} J^N[\Omega^m\psi]\cdot n_{\tau_j}+J^{\overline{N}}[\Omega^m\psi]\cdot n_{\tau_j}\\
&\leq C_0 C_k\tau_j^{-k}\sum_{|m|\leq k+1}\int_{\Sigma_0}J^N[\Omega^m\psi]\cdot n_{0}\\
&+J^{\overline{N}}[\Omega^m\psi]\cdot n_{0}.
\end{split}
\end{equation*}
By the pigeonhole principle, there must therefore exist a $\tau_j\leq \tau_j'\leq \tau_{j+1}$, such that
\begin{equation*}
\int_{\Sigma_{\tau_j'}} J^N[\psi]\cdot n_{\tau_j'}+J^{\overline{N}}[\psi]\cdot n_{\tau_j'} \leq C C_k\tau_j^{-k}\tau_j'^{-1} \sum_{|m|\leq k+1}\int_{\Sigma_0}J^N[\Omega^m\psi]\cdot n_{0}+J^{\overline{N}}[\Omega^m\psi]\cdot n_{0}.
\end{equation*}
In particular, the sequence $\{{\tau'}_{2j+1}\}$ is also dyadic and we obtain
\begin{equation*}
\begin{split}
\int_{\Sigma_{{\tau'}_{2j+1}}} J^N[\psi]\cdot n_{\tau'_{2j+1}}+J^{\overline{N}}[\psi]\cdot n_{{\tau'}_{2j+1}}& \leq CC_k {\tau'}_{2j+1}^{-k-1} \sum_{|m|\leq k+1}\int_{\Sigma_0}J^N[\Omega^m\psi]\cdot n_{0}\\
&+J^{\overline{N}}[\Omega^m\psi]\cdot n_{0}.
\end{split}
\end{equation*}

Let $\tau\geq 0$ be arbitrary. Then there exists a $j\in \N$ such that $\tau_j\leq \tau \leq \tau_{j+1}$. By non-degenerate energy boundedness and $\tau\sim \tau_j$, we conclude
\begin{equation*}
\begin{split}
\int_{\Sigma_{\tau}} J^N[\psi]\cdot n_{\tau}+J^{\overline{N}}[\psi]\cdot n_{\tau} &\leq C\int_{\Sigma_{\tau_j}} J^N[\psi]\cdot n_{\tau_j}+J^{\overline{N}}[\psi]\cdot n_{\tau_j}\\
&\leq C\tau^{-k} \sum_{|m|\leq k+1}\int_{\Sigma_0}J^N[\Omega^m\psi]\cdot n_{0}+J^{\overline{N}}[\Omega^m\psi]\cdot n_{0}. \qedhere
\end{split}
\end{equation*}
\end{proof}
If we restrict to single modes $\psi=\psi_l$ (or to low angular frequencies), $l\geq 1$, we can obtain a stronger energy decay statement, with a uniform constant that depends on $l$. We will need a Gr\"onwall-type lemma.
\pagebreak
\begin{lemma}
\label{lm:gronwall}
Let $f:\R \to [0,\infty)$ be a function that satisfies the inequalities
\begin{itemize}
\item[(i)]
\begin{equation*}
\int_{t'}^{t} f(\bar{t}) \,d\bar{t}\leq  C_1 f(t'),
\end{equation*}
\item[(ii)]
\begin{equation*}
f(t)\leq C_2 f(t'),
\end{equation*}
\end{itemize}
for all $0\leq t'<t$ and $C_{1,2}>0$ constants. Then there exists a numerical constant $c>0$ such that
\begin{equation*}
f(t)\leq C_2 f(0)e^{-\frac{c}{C_1C_2}t}.
\end{equation*}
\end{lemma}
\begin{proof}
Let and consider a sequence $\{t_i\}$ such that $t_0=0$ and $2C_1C_2\leq t_i-t_{i-1}\leq  4C_1C_2$. By the pigeonhole principle, we can take $t_{i+1}$ to satisfy
\begin{equation*}
f(t_{i+1})\leq \frac{1}{2C_1C_2}\int_{t_i+2C_1C_2}^{t_i+4C_1C_2}f(\bar{t})\,d\bar{t}.
\end{equation*}
By the assumptions (i) and (ii) of the lemma we can therefore estimate
\begin{equation*}
f(t_{i+1})\leq \frac{1}{2C_2C_1}C_2f(t_i+2C_1C_2)\leq \frac{1}{2}f(t_i).
\end{equation*}
Consequently, we have for all $i\in \N$ that
\begin{equation*}
f(t_i)\leq f(0)e^{-i\log 2}.
\end{equation*}

Now let $t\geq 0$. There exists a $j\in \N$ such that $t\in [t_j,t_{j+1})$. Moreover, we can estimate
\begin{equation*}
j\geq \frac{t}{4 C_1C_2}.
\end{equation*}
Invoking assumption (ii), we therefore arrive at the estimate
\begin{equation*}
f(t)\leq C_2f(t_j)\leq C_2 f(0)e^{-\frac{\log 2}{4C_1C_2}t}. \qedhere
\end{equation*}
\end{proof}

\begin{proposition}
\label{prop:singlemodedec}
Let $\psi=\psi_l$, with $l\geq 0$, then
\begin{equation*}
\int_{\Sigma_{\tau}} J^N[\psi_l]\cdot n_{\tau} \leq Ce^{-c(l)\tau} \int_{\Sigma_0}J^N[\psi_l]\cdot n_{\tau}+J^{\overline{N}}[\psi_l]\cdot n_{\tau},
\end{equation*}
where $C>0$ is a uniform constant, independent of $l$ and $c(l)= \frac{\tilde{C}}{l(l+1)}$ for $l\geq 1$ and $c(0)=\tilde{C}$, where $\tilde{C}>0$ is independent of $l$.
\end{proposition}
\begin{proof}
The proposition follows directly from the integrated energy decay statement for $\psi_l$ with a factor $l(l+1)$ on the right-hand side for $l\geq 1$, and Lemma \ref{lm:gronwall}.
\end{proof}

\subsection{Pointwise decay}
We can straightforwardly obtain pointwise decay from the energy decay statements in Propositions \ref{prop:energydecaynariai} and \ref{prop:singlemodedec}.
\begin{proposition}
\label{prop:pointwisedecnariai}
If we restrict to $l\geq 1$ in the $\mu=0$ case, we obtain the following decay of solutions to (\ref{eq:kgequation}), with suitably regular initial data, in regions $\mathcal{R}$:
\begin{equation*}
||\psi||^2_{L^{\infty}(\Sigma_{\tau})}\leq C\tau^{-k}\sum_{|m|\leq k}E_N[\Omega^m \psi]+E_N[\Omega^m T \psi]+E_N[\Omega^m R \psi]+E_N[\Omega^m \overline{R} \psi].
\end{equation*}
where $R=N-T$ and $\overline{R}=\overline{N}-T$, and $C>0$ is a uniform constant.

Moreover, if we restrict to a single mode $\psi=\psi_l$, there exist uniform constants $C>0$ and $\tilde{C}>0$, such that for $c(l)= \frac{\tilde{C}}{l(l+1)}$ when $l\geq 1$ and $c(0)=\tilde{C}$, we can estimate
\begin{align*}
||\psi_l||^2_{L^{\infty}(\Sigma_{\tau})}&\leq Ce^{-{c(l)}\tau}\left(E_N[\psi_l]+E_N[ T \psi_l]+E_{N}[ Y\psi_l]+E_{N}[ \overline{Y}\psi_l]\right),
\end{align*}
where in the second estimate we need to take $l\geq 1$ if $\mu=0$.
\end{proposition}
\begin{proof}
The statements of the proposition follow from the energy estimates in Propositions \ref{prop:energydecaynariai} and \ref{prop:singlemodedec} by commuting $\square_g$ with $Y$, as in Proposition \ref{prop:pointboundn}. The additional error terms are spacetime integrals of $N(\psi)\square_g(Y\psi)$ and can be absorbed into $K^N[Y\psi]$ and $K^N[T\psi]$ for suitably large $x_1<x_0$. This can be repeated for $\overline{Y}$.

The restriction $l\geq 1$ arises in the $\mu=0$ case because $\psi_0$ is uniformly bounded but does not decay, see Section \ref{sec:nondec}.
\end{proof}

We have now proved Theorem \ref{th:decn}.
\section{Uniform boundedness in Nariai via a global method}
\label{sec:globalboundnariai}
In this section we work in the global coordinates $(\tilde{t},\tilde{x},\theta,\phi)$ on Nariai, introduced in Section \ref{sec:narstco}. Consider the timelike vector field
\begin{equation*}
\tilde{N}=f(\tilde{t})\partial_{\tilde{t}}.
\end{equation*}
The corresponding energy current through the foliation $\tilde{\Sigma}_{\tau}$ introduced in Section \ref{sec:foliationsnariai}, is given by
\begin{equation*}
J^{\tilde{N}}[\psi]\cdot n_{\tau}=\frac{\sqrt{K}f}{2}\left[(\partial_{\tilde{t}}\psi)^2+\cosh^{-2}\tilde{t}(\partial_{\tilde{x}}\psi)^2+K^{-1}(|\snabla\psi|^2+\mu^2\psi^2)\right].
\end{equation*}
\begin{proposition}
\label{prop:eboundglobalnariai}
Let $f(\tilde{t})=\cosh^{-1}\tilde{t}$. Then for any $\tau\geq 0$,
\begin{equation}
\label{eq:globalnariaieb1}
\int_{\tilde{\Sigma}_{\tau}}J^{\tilde{N}}[\psi]\cdot n_{\tau}\leq \int_{\tilde{\Sigma}}J^{\tilde{N}}[\psi]\cdot n_{\tilde{\Sigma}}=:\tilde{E}_N[\psi].
\end{equation}

In particular, there exists a constant $C=C(K)>0$ such that
\begin{equation}
\label{eq:globalnariaieb2}
\int_{\R}\int_{\s^2}(\partial_{\tilde{t}}\psi)^2+\cosh^{-2}\tilde{t}(\partial_{\tilde{x}}\psi)^2+K^{-1}(|\snabla\psi|^2+\mu^2\psi^2)\,d\mu_{\s^2}d\tilde{x}\Big|_{t=\tau}\leq Cr_0^{-2}\tilde{E}_N[\psi].
\end{equation}
\end{proposition}
\begin{proof}
Denote $f'(\tilde{t})=\frac{df}{dt}(\tilde{t})$. The compatible current $K^{\tilde{N}}$ corresponding to $J^{\tilde{N}}$ is given by
\begin{equation*}
K^{\tilde{N}}=\frac{K}{2}(-f'+f\tanh \tilde{t})[(\partial_{\tilde{t}}\psi)^2+\cosh^{-2}\tilde{t}(\partial_{\tilde{x}}\psi)^2]+\frac{1}{2}(-f'-f\tanh\tilde{t})[|\snabla \psi|^2+\mu^2\psi^2].
\end{equation*}
With the choice $f(\tilde{t})=\cosh^{-1}\tilde{t}$, $f'=-f\tanh \tilde{t}$. We can therefore estimate
\begin{equation*}
K^{\tilde{N}}=K f\tanh\tilde{t}[(\partial_{\tilde{t}}\psi)^2+\cosh^{-2}\tilde{t}(\partial_{\tilde{x}}\psi)^2]\geq 0.
\end{equation*}
By the non-negativity of $K^{\tilde{N}}$, we can apply Stokes' theorem in the spacetime region bounded by $\tilde{\Sigma}_{\tau}$ and $\tilde{\Sigma}$ to obtain (\ref{eq:globalnariaieb1}).

Using the exponential growth of the volume form on $\tilde{\Sigma}_{\tau}$, $d\mu_{\tilde{\Sigma}_{\tau}}=K^{-\frac{1}{2}}r_0^2\cosh \tilde{t}d\tilde{x}d\mu_{\s^2}$, the estimate (\ref{eq:globalnariaieb2}) follows. 
\end{proof}

By commuting with the spacelike Killing vector field $X$, we can obtain global uniform pointwise boundedness of $\psi$.
\begin{corollary}
\label{cor:pointboundglobalnariai}
There exists a constant $C=C(K,\mu)>0$ such that
\begin{equation}
\label{eq:globalnariaipointwisedecay}
||\psi||^2_{L^{\infty}(\tilde{\Sigma}_{\tau})}\leq C \sum_{|k|+|m|\leq 2, |m|\leq 1} \tilde{E}_N[\Omega^m X^k \psi]+C\left(||\Psi_0||_{L^{\infty}(\tilde{\Sigma})}^2+||\Psi'_0||_{L^1(\tilde{\Sigma})}^2+||\partial_{\tilde{x}}\Psi_0||^2_{L^1(\tilde{\Sigma})}\right),
\end{equation}
where in the $\mu\neq 0$ case $C=\tilde{C}(K)r_0^{-2}\mu^{-2}$, with $\tilde{C}(K)>0$, and we can remove the $L^1$ and $L^{\infty}$ norms of $\Psi_0$ and $\Psi_0'$ on the right-hand side of the above inequality.\end{corollary}
\begin{proof}
Recall, $X=\sqrt{K}\partial_{\tilde{x}}$. We therefore commute with $\Omega_i$ and $X$ and apply the Sobolev estimate Proposition \ref{sobs2}, together with the Poincar\'e inequality (\ref{eq:poincares2fixedl}), to obtain (\ref{eq:globalnariaipointwisedecay}) in the case that $\psi=\psi_{l\geq 1}$ or $\mu\neq 0$. $\psi_0$ can be bounded as in Proposition \ref{prop:pointboundn}.
\end{proof}

We have now proved Theorem \ref{th:globboundn}.

\section{Uniform decay in Nariai via a global method}
\label{sec:globaldecaynariai}
In order to obtain energy decay, we modify the current $J^{\tilde{N}}$,
\begin{equation*}
J^{\tilde{N},1}_{\alpha}:=J^{\tilde{N}}_{\alpha}+\frac{h'}{2}\psi\partial_{\alpha}\psi-\frac{1}{4}\psi^2\partial_{\alpha}h',
\end{equation*}
where $h=h(\tilde{t})$ and $h'(\tilde{t})=\frac{dh}{d\tilde{t}}(\tilde{t})$. The corresponding modified compatible current is defined as follows:
\begin{equation}
\label{eq:globalnariaimodk}
K^{\tilde{N},1}:=\nabla^{\alpha}J^{\tilde{N},1}_{\alpha}=K^N+\frac{h'}{2}g^{\alpha \beta}\partial_{\alpha}\psi\partial_{\beta}\psi-\frac{1}{4}\psi^2\square_gh',
\end{equation}
where
\begin{equation*}
\square_gh'=-K(h'''+h''\tanh\tilde{t}).
\end{equation*}
We will show in the proof of the proposition below that the modified energy current $J^{\tilde{N},1}$ globally satisfies an estimate analogous to (i) of Theorem \ref{thm:redshift} for the local red-shift current $J^N\cdot n_{\tau}$,
\begin{equation*}
\int_{\s^2} K^{\tilde{N},1}[\psi_{l\geq 1}]d\mu_{\s^2}\geq C \int_{\s^2}J^{\tilde{N}}[\psi_{l\geq 1}]\cdot n_{\tau}d\mu_{\s^2}.
\end{equation*}

\begin{proposition}
\label{prop:globedecnariaipsigeq1}
Let $\psi=\psi_{l\geq 1}$. Moreover, take $\tau_0>\max\{\log (2\sqrt{K}r_0),1\}$. Then there exists a constant $c=c(K,\mu,r_0)>0$, such that
\begin{equation*}
\int_{\tilde{\Sigma}_{\tau}}J^{\tilde{N}}[\psi]\cdot n_{\tau}\leq e^{-c(\tau-\tau_0)}\int_{\tilde{\Sigma}_{\tau_0}}J^{\tilde{N}}[\psi]\cdot n_{\tau_0},
\end{equation*}
where in the case $\mu=0$, $c=c(K)(1+r_0^2)^{-1}>0$ and in the case $\mu\neq 0$, $c=\tilde{c}(K)(1+r_0^2+\mu^{-2}+\mu^2r_0^2)^{-1}>0$.
\end{proposition}
\begin{proof}
Let $f(t)=\cosh^{-1} \tilde{t}$. By (\ref{eq:globalnariaimodk}), we obtain the following expression for $K^{\tilde{N},1}$:
\begin{equation*}
\begin{split}
K^{\tilde{N},1}&= K\left(-f'-\frac{h'}{2}\right)(\partial_{\tilde{t}}\psi)^2+K\left(-f'+\frac{h'}{2}\right)\cosh^{-2}\tilde{t}(\partial_{\tilde{x}}\psi)^2+\frac{h'}{2}|\snabla\psi|^2\\
&+\frac{K}{4}\left(h'''+h''\tanh\tilde{t}\right)\psi^2,
\end{split}
\end{equation*}
where we have used that $f'=-f\tanh\tilde{t}$. We take $h=-f$, then
\begin{align*}
h'&=-f'=f\tanh\tilde{t},\\
-f'-\frac{h'}{2}&=-\frac{f'}{2}=\frac{f}{2}\tanh \tilde{t},\\
-f'+\frac{h'}{2}&=\frac{3f}{2}\tanh \tilde{t}.
\end{align*}

Moreover,
\begin{align*}
h''&=\frac{1-\sinh^2 \tilde{t}}{\cosh^3\tilde{t}},\\
h'''&=\frac{\sinh^3\tilde{t}-5\sinh\tilde{t}}{\cosh^4\tilde{t}},
\end{align*}
so
\begin{equation*}
h'''+h''\tanh\tilde{t}=\frac{-4\sinh\tilde{t}}{\cosh^4\tilde{t}}=-\frac{4}{\cosh^2\tilde{t}}h'.
\end{equation*}
The $\psi^2$ term is the only term that has the wrong sign. However, using that $\psi=\psi_{l\geq 1}$ and applying the Poincar\'e inequality (\ref{eq:poincares2fixedl}), we can control
\begin{equation*}
\begin{split}
\int_{\s^2} \frac{K}{\cosh^2\tilde{t}}h'\psi^2\,d\mu_{\s^2}&\leq \frac{K r_0^2}{ \cosh^2\tilde{t}}\int_{\s^2}\frac{h'}{2}|\snabla\psi|^2\,d\mu_{\s^2}\\
&\leq \frac{1}{4}\int_{\s^2}\frac{h'}{2}|\snabla\psi|^2\,d\mu_{\s^2},
\end{split}
\end{equation*}
if we require $\cosh \tau_0\geq 2\sqrt{K} r_0$. This follows in particular from the lower bound $\tau_0>\max\{\log (2\sqrt{K}r_0),1\}$. Hence,
\begin{equation*}
\int_{\s^2} J^{\tilde{N}}[\psi]\cdot n_{\tau}\leq \frac{C}{\tanh \tau_0} \int_{\s^2} K^{\tilde{N},1}[\psi]\leq C\int_{\s^2} K^{\tilde{N},1},
\end{equation*}
where in the case $\mu=0$, $C>0$ is a numerical constant and in the case $\mu\neq 0$, $C=\tilde{C}(K)(\mu^2r_0^2+1)$, with $\tilde{C}(K)>0$.

We can estimate
\begin{equation*}
\begin{split}
\int_{\s^2}J^{\tilde{N},1}\cdot n_{\tau}\,d\mu_{\s^2}&=\int_{\s^2}J^{\tilde{N}}\cdot n_{\tau}-\sqrt{K}\frac{f \tanh \tilde{t}}{2}\psi\partial_{t}\psi-\frac{\sqrt{K}}{4}\psi^2 \frac{1-\sinh^2 \tilde{t}}{\cosh^2\tilde{t}}f\,d\mu_{\s^2}\\
&\leq C \int_{\s^2}J^{\tilde{N}}\cdot n_{\tau}\,d\mu_{\s^2},
\end{split}
\end{equation*}
where $C=\tilde{C}(K)(r_0^2+1)>0$ if $\mu=0$ and $C=\tilde{C}(K)(\mu^{-2}+1)>0$ if $\mu\neq 0$.

Applying Stokes' theorem in the spacetime region bounded by $\tilde{\Sigma}_{\tau}$ and $\tilde{\Sigma}_{\tau_0}$, together with the boundedness statement of Proposition \ref{prop:eboundglobalnariai} therefore results in the estimate
\begin{equation*}
\int_{\tau_0}^{\tau}\int_{\tilde{\Sigma}_{\bar{\tau}}} J^{\tilde{N}}\cdot n_{\tau}\,d\bar{\tau}\leq C\int_{J^+(\tilde{\Sigma}_{\tau_0})\cap J^-(\tilde{\Sigma}_{\tau})}K^{\tilde{N},1} \leq C \int_{\tilde{\Sigma}_{\tau_0}} J^{\tilde{N}}\cdot n_{\tau_0},
\end{equation*}
where $C=\tilde{C}(K)(1+r_0^2)>0$ if $\mu=0$ and $C=\tilde{C}(K)(\mu^2r_0^2+1)(1+\mu^{-2})>0$ if $\mu\neq 0$.

We apply energy boundedness from Proposition \ref{prop:eboundglobalnariai} and Lemma \ref{lm:gronwall} to obtain the decay estimate
\begin{equation*}
\int_{\tilde{\Sigma}_{\tau}} J^{\tilde{N}}[\psi]\cdot n_{\tau}\leq e^{-c(\tau-\tau_0)}\int_{\tilde{\Sigma}_{\tau_0}} J^{\tilde{N}}[\psi]\cdot n_{\tau_0},
\end{equation*}
where in the case $\mu=0$, $c=c(K)(1+r_0^2)^{-1}>0$ and in the case $\mu\neq 0$, $c=\tilde{c}(K)(1+r_0^2+\mu^{-2}+\mu^2r_0^2)^{-1}>0$.
\end{proof}

\begin{proposition}
\label{prop:globedecnariaipsi0}
Let $\mu\neq 0$. Take $\tau_0>\max\{\log (2\sqrt{2}\sqrt{K}|\mu|^{-1}),1\}$. Then there exist a constant $c=c(K)>0$, such that
\begin{equation*}
\int_{\tilde{\Sigma}_{\tau}} J^{\tilde{N}}[\psi_0]\cdot n_{\tau} \leq e^{-\frac{c}{1+\mu^{-2}}(\tau-\tau_0)}\int_{\tilde{\Sigma}_{\tau_0}} J^{\tilde{N}}[\psi_0]\cdot n_{\tau_0}.
\end{equation*}
\end{proposition}
\begin{proof}
In the $\mu\neq 0$ case we deal with the spherically symmetric modes $\psi_0$, satisfying
\begin{equation}
\label{eq:sphkgeqlobalnariai}
\square_g \psi_0=\mu^2\psi_0,
\end{equation}
by considering the auxiliary metric
\begin{equation*}
\tilde{g}=K^{-1}(-d\tilde{t}^2+\cosh^2 \tilde{t}d\tilde{x}^2)+\tilde{r_0}^2\slashed{g}_{\s^2}.
\end{equation*}
We take $\tilde{r_0}^2:=2\mu^{-2}$. Consider $\psi(t,x,\theta,\phi)=\psi_0(t,x)Y_{11}(\theta,\phi)$. By definition of $\tilde{r_0}^2$ and the expression (\ref{eq:kqequationfixedl}), (\ref{eq:sphkgeqlobalnariai}) is equivalent to $\square_{\tilde{g}} \psi=0$, with $\psi$ the solution arising from initial data $\Psi=\Psi_0Y_{11}$, $\Psi'=\Psi_0'Y_{11}$ on $\Sigma_{\tau_0}$. 

By Proposition \ref{prop:globedecnariaipsigeq1} applied to $\psi$, with respect to the metric $\tilde{g}$, we can estimate
\begin{equation}
\label{eq:preedecpsi0globalnariai}
\begin{split}
&\int_{\R} \int_{\s^2} |Y_{11}|^2\left[(\partial_{\tilde{t}}\psi_0)^2+\cosh^{-2}\tilde{t} (\partial_{\tilde{x}}\psi_0)^2+\frac{\mu^2}{K}\psi_0^2\right]\,r_0^2\,d\mu_{\s^2}d\tilde{x}\Big|_{\tilde{t}=\tau}\\
&\leq e^{-\frac{c}{1+r_0^2}(\tau-\tau_0)}\int_{\R} \int_{\s^2} |Y_{11}|^2\left[(\partial_{\tilde{t}}\psi_0)^2+\cosh^{-2}\tilde{t} (\partial_{\tilde{x}}\psi_0)^2+\frac{\mu^2}{K}\psi_0^2\right]\,r_0^2\,d\mu_{\s^2}d\tilde{x}\Big|_{\tilde{t}=\tau_0}.
\end{split}
\end{equation}
Since $\int_{\s^2}|Y_{11}|^2\,d\mu_{\s^2}=1$ we can rewrite (\ref{eq:preedecpsi0globalnariai}) to obtain
\begin{equation*}
\int_{\tilde{\Sigma}_{\tau}} J^{\tilde{N}}[\psi_0]\cdot n_{\tau} \leq e^{-\frac{c}{1+\mu^{-2}}(\tau-\tau_0)}\int_{\tilde{\Sigma}_{\tau_0}} J^{\tilde{N}}[\psi_0]\cdot n_{\tau_0},
\end{equation*}
where $c=c(K)>0$.
\end{proof}

\begin{corollary}
\label{prop:globedecnariai}
Let $\psi=\psi_{l\geq 1}$ if $\mu=0$. Then there exist constants $C=C(K,r_0,\mu)>0$ and $c=c(K,\mu,r_0)>0$, such that
\begin{equation}
\label{eq:finaledecglobalnariai}
\int_{\tilde{\Sigma}_{\tau}}J^{\tilde{N}}[\psi]\cdot n_{\tau}\leq Ce^{-c\tau}\tilde{E}_N[\psi],
\end{equation}
where in the case $\mu\neq 0$, $C=\tilde{C}(r_0,K)e^{\tilde{c}(K)\tau_0}$, with \newline$\tau_0>\max\{1,\log (2\sqrt{2}\sqrt{K}|\mu|^{-1})\}$ and $c=(1+\mu^{-2}+\mu^2)^{-1}\tilde{c}(K,r_0)$.
\end{corollary}
\begin{proof}
By the orthogonality of the modes $\psi_{l\geq1}$ and $\psi_0$, we can combine the estimates in Propositions \ref{prop:globedecnariaipsigeq1}, \ref{prop:globedecnariaipsi0} and \ref{eq:globalnariaieb2} to obtain (\ref{eq:finaledecglobalnariai}).
\end{proof}

\begin{remark}
It may seem that the energy decay statement of Proposition \ref{prop:globedecnariai} is in contradiction with the results of Sbierski \cite{sbier1} mentioned in Section \ref{sec:trapping}, as there is no loss of derivatives due to the trapping of null geodesics. However, these results require in particular that $g(N,N)\leq -c$ for some uniform constant $c>0$, where $N$ is a timelike vector field used to define the energy current corresponding to a foliation of spacelike hypersurfaces. See for example Theorem 2.42 of \cite{sbier1}. We have that $g(\tilde{N},\tilde{N})\to 0$ as $t\to \infty$, so the results of Sbierski do not apply to $\tilde{E}_N$. 

The energy $E_N$ defined in Section \ref{sec:boundn} does however satisfy the requirements needed for the results of Sbierski to hold, so there must be a loss of derivatives in the corresponding energy decay statement, as seen in Proposition \ref{prop:energydecaynariai}.
\end{remark}

\begin{corollary}
\label{prop:globedecnariaipointwise}
Let $\psi=\psi_{l\geq 1}$ if $\mu=0$. Then there exists constants $C=C(K,r_0,\mu)>0$ and $c=c(K,\mu,r_0)>0$ such that
\begin{equation*}
||\psi||^2_{L^{\infty}(\tilde{\Sigma}_{\tau})}\leq Ce^{-c\tau} \sum_{|k|+|m|\leq 2, |m|\leq 1} \tilde{E}_N[\Omega^m X^k \psi],
\end{equation*}
where in the case $\mu\neq 0$,  $C=\tilde{C}(r_0,K)e^{\tilde{c}(K)\tau_0}$, with \newline $\tau_0>\max\{1,\log (2\sqrt{2}\sqrt{K}|\mu|^{-1})\}$ and $c=(1+\mu^{-2}+\mu^2)^{-1}\tilde{c}(K,r_0)$.
\end{corollary}
\begin{proof}
The pointwise estimate follows from Corollary \ref{prop:globedecnariai}, by applying Sobolev estimates on $\R\times \s^2$ as in Proposition \ref{cor:pointboundglobalnariai}.
\end{proof}

We have now proved Theorem \ref{th:globdecn}.

\begin{remark}
We can compare the pointwise decay result from Proposition \ref{prop:pointwisedecnariai} to the result of Corollary \ref{prop:globedecnariaipointwise}. Using the relation (\ref{eq:tildetvst}) between $t$, defined in the static region, and $\tilde{t}$, defined globally, we can replace $e^{-c\tau}$ by $e^{\tilde{c}|x_*|}e^{-\tilde{c}t}$ in the static region, where $\tilde{c}>0$ is a uniform constant. The factor $e^{\tilde{c}|x_*|}$ blows up as we approach the cosmological horizons. Therefore, while Corollary \ref{prop:globedecnariaipointwise} gives a stronger decay rate at a fixed $x$ compared to Proposition \ref{prop:pointwisedecnariai}, the decay estimate is not uniform in $x$.

Moreover, unlike Corollary \ref{prop:globedecnariaipointwise}, Proposition \ref{prop:pointwisedecnariai} does not rely on the geometry outside the static region.
\end{remark}

\section{Uniform decay in $n$-dimensional de Sitter space}
\label{sec:decaydsn}
The methods of Sections \ref{sec:globalboundnariai} and \ref{sec:globaldecaynariai} can be generalised to spacetimes of the form $dS_n\times_{r_0} \s^2$, where $n\geq 2$, to obtain in particular uniform decay of solutions to (\ref{eq:kgequation}) with $\mu\neq 0$ on a $dS_n$ background.

Choose $r_0^2:=2\mu^{-2}$ and consider $dS_n\times_{r_0} \s^2$ in global coordinates \newline $(\tilde{t},\theta_1,\ldots,\theta_{n-2},\phi,\vartheta,\varphi)$, where $(\theta_1,\ldots,\theta_{n-2},\phi)$ are coordinates on $\s^{n-1}$ and $(\vartheta,\varphi)$ are coordinates on $\s^2$. The metric is then given by
\begin{equation*}
g=K^{-1}(-d\tilde{t}^2+\cosh^2 \tilde{t}\slashed{g}_{\s^{n-1}})+r_0^2\slashed{g}_{\s^2}.
\end{equation*}

Let $\tilde{\Sigma}=\{\tilde{t}=0\}$, $\tilde{\Sigma}_{\tau}:=\{\tilde{t}=\tau\}$ and $n_{\tau}=n_{\tilde{\Sigma}_{\tau}}$ the corresponding future-directed normal. Consider the timelike vector field
\begin{equation*}
\tilde{N}=f(\tilde{t})\partial_{\tilde{t}}.
\end{equation*}

We consider (\ref{eq:waveequation}) on $dS_n\times_{r_0} \s^2$. The corresponding energy current through the foliation $\tilde{\Sigma}_{\tau}$ is given by
\begin{equation*}
J^{\tilde{N}}[\psi]\cdot n_{\tau}=\frac{\sqrt{K}f}{2}\left[(\partial_{\tilde{t}}\psi)^2+|\snabla \psi|^2+K^{-1}r_0^{-2}|\mathring{\snabla}\psi|^2\right],
\end{equation*}
where $|\snabla \psi|^2=\cosh^{-2}\tilde{t}|\snabla_{\s^{n-1}}\psi|^2$ with $\snabla_{\s^{n-1}}$ denoting the derivative restricted to $\s^{n-1}$, and $\mathring{\snabla}$ denotes the derivative restricted to $\s^2$.

\begin{proposition}
\label{prop:ebounddsn}
Let $f(\tilde{t})=\cosh^{-(n-1)}\tilde{t}$. Then for any $\tau\geq 0$,
\begin{equation*}
\int_{\tilde{\Sigma}_{\tilde{\tau}}}J^{\tilde{N}}[\psi]\cdot n_{\tau}\leq \int_{\tilde{\Sigma}}J^{\tilde{N}}[\psi]\cdot n_{\tilde{\Sigma}}=:\tilde{E}[\psi].
\end{equation*}
\end{proposition}
\begin{proof}
The proof proceeds in the same way as the proof of Proposition \ref{prop:eboundglobalnariai}, with $K^{\tilde{N}}$ for general $n$ given by
\begin{equation*}
\begin{split}
K^{\tilde{N}}&=\frac{K}{2}(-f'+f(n-1)\tanh \tilde{t})(\partial_{\tilde{t}}\psi)^2+\frac{K}{2}(-f'-f(n-1)\tanh\tilde{t}+2)|\snabla\psi|^2\\
&+\frac{1}{2}(-f'-f(n-1)\tanh\tilde{t})r_0^{-2}|\mathring{\snabla}\psi|^2.
\end{split}
\end{equation*}

By choosing $f(\tilde{t})=\cosh^{-(n-1)}\tilde{t}$, $K^{\tilde{N}}\geq 0$ easily follows.
\end{proof}

As in Nariai, we can perform a decomposition of $\psi$ into spherical harmonics on $\s^2$,
\begin{equation*}
\psi=\sum_{l\geq 0}\psi_l,
\end{equation*}
where
\begin{equation*}
\psi_l(\tilde{t},\theta_1,\ldots,\theta_{n-2},\phi,\vartheta,\varphi):=\sum_{m=-l}^l \psi_{m,l}(\tilde{t},\theta_1,\ldots,\theta_{n-2},\phi)Y^{m,l}(\vartheta,\varphi).
\end{equation*}

\begin{proposition}
\label{prop:dsndecay1}
Let $\psi=\psi_{l\geq 1}$. Moreover, take $\tau_0\geq \max\{\log\left(2\sqrt{K}r_0\sqrt{\frac{n}{2}}\right),1\}$. Then there exists a constant $c=c(K,n)>0$, such that
\begin{equation*}
\int_{\tilde{\Sigma}_{\tau}}J^{\tilde{N}}[\psi]\cdot n_{\tau}\leq e^{-\frac{c}{1+r_0^2}(\tau-\tau_0)}\int_{\tilde{\Sigma}_{\tau_0}}J^{\tilde{N}}[\psi]\cdot n_{\tau_0}.
\end{equation*}
\end{proposition}
\begin{proof}
As in Nariai, we define the modified current $J^{\tilde{N}}$,
\begin{equation*}
J^{\tilde{N},1}_{\alpha}:=J^{\tilde{N}}_{\alpha}+\frac{h'}{2}\psi\partial_{\alpha}\psi-\frac{1}{4}\psi^2\partial_{\alpha}h',
\end{equation*}
where $h=h(\tilde{t})$ and $h'(\tilde{t})=\frac{dh}{d\tilde{t}}(\tilde{t})$. The corresponding modified compatible current is defined as follows:
\begin{equation*}
K^{\tilde{N},1}:=\nabla^{\alpha}J^{\tilde{N},1}_{\alpha}=K^N+\frac{h'}{2}g^{\alpha \beta}\partial_{\alpha}\psi\partial_{\beta}\psi-\frac{1}{4}\psi^2\square_gh',
\end{equation*}
where
\begin{equation*}
\square_gh'=-K(h'''+(n-1)h''\tanh\tilde{t}).
\end{equation*}

We take $h=-f$ from which is is easy to derive that
\begin{equation*}
h'''+(n-1)h''\tanh\tilde{t}=\frac{-2n}{\cosh^2 \tilde{t}}h'.
\end{equation*}

It is straightforward to show that
\begin{equation*}
\begin{split}
K^{\tilde{N},1}&=\frac{K}{2}(n-1)f\tanh\tilde{t}(\partial_{\tilde{t}}\psi)^2+\frac{n+1}{2}Kf\tanh\tilde{t}|\snabla\psi|^2+\frac{1}{2}(n-1)h'r_0^{-2}|\mathring{\snabla}\psi|^2\\
&-\frac{nK}{2\cosh^2\tilde{t}}h'\psi^2.
\end{split}
\end{equation*}
We can estimate
\begin{equation*}
\begin{split}
\frac{nK}{2\cosh^2\tilde{t}}h'\int_{\s^2}\psi^2\,d\mu_{\s^2}&\leq \frac{nK}{2}\frac{r_0^2}{\cosh^2\tilde{t}}\int_{\s^2}\frac{1}{2}h'r_0^{-2}|\mathring{\snabla}\psi|^2\,d\mu_{\s^2}\\
&\leq \frac{1}{4}\int_{\s^2}\frac{1}{2}h'r_0^{-2}|\mathring{\snabla}\psi|^2\,d\mu_{\s^2},
\end{split}
\end{equation*}
if we take $\tau_0\geq \max\{1,\log\left(2\sqrt{K}r_0\sqrt{\frac{n}{2}}\right)\}$.

The remainder of the proof proceeds in the same way as the proof of Proposition \ref{prop:globedecnariaipsigeq1}. 
\end{proof}

\begin{proposition}
\label{prop:dsndecay2}
There exist constants $C=C(n,K,\mu)>0$ and \newline $c=c(n,K,\mu)>0$, such that
\begin{equation}
\label{eq:finaledecdsn}
\int_{\tilde{\Sigma}_{\tau}}J^{\tilde{N}}[\psi_0]\cdot n_{\tau}\leq Ce^{-c\tau}\tilde{E}_N[\psi_0],
\end{equation}
where in the case $\mu\neq 0$, $C=\tilde{C}(n,K)e^{\tilde{c}(n,K)\tau_0}$, with \newline $\tau_0>\max\{1,\log (2\sqrt{2}\sqrt{K}\sqrt{\frac{n}{2}}|\mu|^{-1})\}$ and $c=(1+\mu^{-2})^{-1}\tilde{c}(n,K)$.
\end{proposition}
\begin{proof}
The proof proceeds in the same was as the proof of Proposition \ref{prop:globedecnariaipsi0}, using the results of Proposition \ref{prop:dsndecay1} for $\psi=\psi_0Y_{11}$ and $r_0^2=2\mu^{-2}$, together with Proposition \ref{prop:ebounddsn}.
\end{proof}

The metric $g_{dS_n}$ does not depend on the coordinates on $\s^{n-1}$ so there exist Killing vector fields $\Omega_i$, with $i=1,\ldots,n$ that generate the $(n-1)$-spherical symmetry, as in the $n=3$ case. Moreover, Proposition \ref{sphsym} can be generalised to $n\geq 2$ to obtain
\begin{equation*}
\int_{\s^{n-1}}|\snabla_{\s^{n-1}}\psi|^2\,d\mu_{\s^{n-1}}= \sum_{i=1}^{n}\int_{\s^{n-1}}(\Omega_i\psi)^2\,d\mu_{\s^{n-1}}.
\end{equation*}
We can therefore commute the wave operator $\square_{g}$ with $\Omega_i$ and apply standard Sobolev estimates on $\s^{n-1}$ to obtain pointwise decay from Proposition \ref{prop:dsndecay2}.

\begin{corollary}
\label{prop:pointwisedecdsn}
Let $s\in \N$ be the smallest integer satisfying $s>\frac{n-1}{2}$. Then there exists constants $C=C(n,K,\mu)>0$ and $c=c(n,K,\mu)>0$ such that
\begin{equation*}
||\psi_0||^2_{L^{\infty}(\tilde{\Sigma}_{\tau})}\leq Ce^{-c\tau} \sum_{|k|\leq s} \tilde{E}[\Omega^k \psi_0],
\end{equation*}
where $C=\tilde{C}(n,K)e^{\tilde{c}(n,K)\tau_0}$, with $\tau_0>\max\{1,\log (2\sqrt{2}\sqrt{K}\sqrt{\frac{n}{2}}|\mu|^{-1})\}$ and $c=(1+\mu^{-2})^{-1}\tilde{c}(n,K)$
\end{corollary}

We have now proved Theorem \ref{th:decaydsn}.

\begin{remark}
Instead of considering the entire manifold $dS_n$, we could have restricted to an open subset $dS_{n,flat}$ that can be covered by a single chart $(t',x_1',\ldots,x_{n-1}')$, where $t\in \R$ and $x'_i\in \R$ for all $i=1,\ldots,n-1$, see also Section \ref{sec:subsubdsn}. The metric in these coordinates is given by
\begin{equation*}
g=-d{t'}^2+e^{2\sqrt{\frac{\Lambda}{3}}t'}\sum_{i=1}^{n-1}dx_i^2.
\end{equation*}
By considering a foliation of constant $t'$ hypersurfaces, we can obtain exponential decay in $t'$ using the methods from this section. In particular, we need to take $\tilde{N}=f(t')\partial_{t'}$ with $f(t')=e^{-(n-1)\sqrt{\frac{\Lambda}{3}}t'}$ and $h=-f$ in the modified current $J^{\tilde{N},1}\cdot n_{\tau}$.

One can easily show that $\square_g h'=0$, which means that we do not need to take the product $dS_{n,flat}\times_{r_0}\s^2$ and restrict to $\psi_{l\geq 1}$ first, but we can consider $dS_{n,flat}$ directly.
\end{remark}

\section{Uniform boundedness in Pleba\'nski-Hacyan}
\label{sec:boundph}
We now consider (\ref{eq:kgequation}) on PH. Boundedness follows very easily in this case.
\begin{proposition}
\label{prop:boundednessph}
In PH, we can estimate
\begin{equation*}
\int_{\Sigma_{\tau}} J^T[\psi]\cdot n_{\Sigma_{\tau}}\leq E_{PH}[\psi],
\end{equation*}
where
\begin{equation*}
E_{PH}[\psi]:=\int_{\Sigma} J^T[\psi]\cdot n_{\Sigma},
\end{equation*}
and moreover, there exists a uniform constant $C>0$ such that
\begin{equation*}
\begin{split}
||\psi||_{L^{\infty}(\Sigma_{\tau})}&\leq C \Big(||\Psi_0||_{L^{\infty}(\Sigma)}+||\Psi'_0||_{L^{1}(\Sigma)}+||\partial_x \Psi_0||_{L^1(\Sigma)}\\
&+\sqrt{E_{PH}[\psi]+E_{PH}[T\psi]}\Big).
\end{split}
\end{equation*}
If $\mu\neq 0$, we can remove the $L^1$ and $L^{\infty}$ norms of $\Psi_0$ and $\Psi_0'$ on the right-hand side of the above inequality.
\end{proposition}
\begin{proof}
In PH there exists a uniformly timelike Killing vector field $T$. Non-degenerate energy boundedness is implied by energy conservation, see the comments in Section \ref{sec:vecbasictools}. Moreover, we can apply the elliptic estimate (\ref{est:standardelliptic}) with $N$ replaced by $T$ and commute $\square_g$ with $T$ to arrive at the pointwise boundedness statement.
\end{proof}

We have now proved Theorem \ref{th:boundph}.

\section{Uniform decay in Pleba\'nski-Hacyan}
\label{sec:decph}
As in the static region of Nariai, we can show that integrated local energy decay holds in PH. However, we will not use this to obtain pointwise decay. Instead, we will derive pointwise decay directly by exploiting the boost isometry in $\R^{1+1}$, generated by the Killing vector field $Y$. 
\subsection{Pointwise decay}
To obtain global pointwise decay from global energy decay with respect to the current $J^T$ in PH, we would need positivity of the energy flux associated to $J^T$, through the null boundary $\mathcal{N}_A^+$ and $\mathcal{N}_B^+$. We show below that this flux vanishes, implying instead the non-decay of the energy flux of $J^T$ through any foliation of spacelike hypersurfaces asymptoting at the null boundaries $\mathcal{N}_A^+$ and $\mathcal{N}_B^+$.
\begin{proposition}
\label{prop:vanishingfluxinfinity}
Let $u_0<-1$ and $v_0<-1$. If $\mu=0$, restrict $\psi=\psi_{l\geq 1}$. Then
\begin{align*}
\int_{\mathcal{N}_A^+\cap\{u\leq u_0\}} J^T[\psi]\cdot \partial_u=0,\\
\int_{\mathcal{N}_B^+\cap\{v\leq v_0\}} J^T[\psi]\cdot \partial_v=0,
\end{align*}
if $E_{PH}[\psi]+E_{PH}[\partial_u\psi]+ \sum_{|k|\leq 1}E_{PH}[\Omega^k \psi]$ is finite.
\end{proposition}
\begin{proof}
Without loss of generality, we only consider $\mathcal{N}_A^+$. Let $\mathcal{D}$ be the spacetime region bounded in the past by $\Sigma$, and in the future by the null hypersurfaces $\{u=u_0\}$ and the null boundary segment $\mathcal{N}^+_A\cap \{u\leq u_0\}$. In particular, $u\leq u_0$ everywhere in $\mathcal{D}$.

Define the vector fields $V$ and $W$ in $\mathcal{D}$ by
\begin{align*}
V&=\frac{v}{|u| \log^2|u|}\partial_v,\\
W&=\frac{1}{\log |u|}\partial_u.
\end{align*}
We have that
\begin{align*}
K^V[\psi]&=-\frac{2}{|u|\log^2|u|}\left(|\snabla \psi|^2+\mu^2\psi^2\right)+\textnormal{non-negative terms},\\
K^W[\psi]&=\frac{2}{|u|\log^2|u|}\left(|\snabla\psi|^2+\mu^2\psi^2\right).
\end{align*}
Consequently, $K^V[\psi]+K^W[\psi]\geq 0$.

By Stokes' theorem in $\mathcal{D}$, we then have that
\begin{equation*}
\begin{split}
\int_{\mathcal{N}_A^+\cap\{u\leq u_0\}} J^V[\psi]\cdot \partial_u+J^W[\psi]\cdot \partial_u &\leq \int_{\Sigma\cap\{u\leq u_0\}} J^V[\psi]\cdot n_{\Sigma}+J^W[\psi]\cdot n_{\Sigma}\\
&-\int_{\mathcal{D}} K^V[\psi]+ K^W[\psi]\\
&\leq  \int_{\Sigma\cap\{u\leq u_0\}} J^V[\psi]\cdot n_{\Sigma}+J^W[\psi]\cdot n_{\Sigma}.
\end{split}
\end{equation*}

By the properties of $\Sigma$ in Section \ref{sec:foliationsph}, it follows that $v\sim -u$ on $\Sigma\cap\{u\leq u_0\}$, so
\begin{equation*}
J^V[\psi]\cdot n_{\Sigma}+J^W[\psi]\cdot n_{\Sigma} \leq C J^T[\psi]\cdot n_{\Sigma}.
\end{equation*}
Moreover,
\begin{equation*}
\int_{\s^2} \psi^2 d\mu_{\s^2}\leq \frac{C}{v} \int_{\s^2} J^V[\psi]\cdot \partial_u.
\end{equation*}
In the $\mu=0$ case the above estimate follows from the Poincar\'e inequality on $\s^2$.

The statement of the proposition follows immediately, after commuting with $\partial_u$ and $\Omega_i$ and taking the limit $v\to \infty$.
\end{proof}

By homogeneity of PH, we can in fact conclude from Proposition \ref{prop:vanishingfluxinfinity} that
\begin{align*}
\int_{\mathcal{N}_A^+} J^T[\psi]\cdot \partial_u=0,\\
\int_{\mathcal{N}_B^+} J^T[\psi]\cdot \partial_v=0.
\end{align*}

Consider a lightcone $\mathcal{C}$ with the hyperboloidal foliation, see Section \ref{sec:foliationsph}. We can apply Stokes' Theorem to the region bounded in the future by $\mathcal{H}_{s}$ and in the past by $\mathcal{H}_1$, in order to obtain the estimate
\begin{equation*}
\int_{\mathcal{H}_s}J^T[\psi]\cdot n_{\mathcal{H}_s}\leq \int_{\mathcal{H}_1}J^T[\psi]\cdot n_{\mathcal{H}_1}.
\end{equation*}
Moreover, in the region $|x|\leq t$ we can estimate 
\begin{equation*}
\begin{split}
 J^T[\psi]\cdot n_{\mathcal{H}_s}&= \frac{t}{s}T_{tt}+\frac{x}{s}T_{tx}\\
&= \frac{t}{2s}\left[(\partial_t \psi)^2+(\partial_x \psi)^2+|\snabla \psi|^2+\mu^2\psi^2\right]+\frac{x}{s}(\partial_t\psi)(\partial_x\psi)\\
&\geq \frac{t}{2s}(|\snabla \psi|^2+\mu^2\psi^2).
\end{split}
\end{equation*}
Let $\Gamma$ be a Killing vector field of $(\mathcal{M},g)$. Then we can replace $\psi$ above by $\Gamma \psi$. Before proving pointwise decay of $\psi$, we need to make use of a scaled Sobolev inequality on the hyperboloid.
\begin{lemma}
\label{sobhs}
Let $U(x):=u\left(\sqrt{s^2+x^2},x\right)$, where $U$ is suitably regular such that the right-hand side below is well-defined. Then for $(t,x)\in \mathcal{H}_s$
\begin{equation}
\label{eq:sobhs1}
t|u(t,x)|^2\leq C \int_{\R} U(\bar{x})^2+(s^2+\bar{x}^2)(\partial_{\bar{x}} U)^2(\bar{x})\,d\bar{x}.
\end{equation}
\end{lemma}
\begin{proof}
We parametrise $\mathcal{H}_{s}$ by the $x$-coordinate and use that $t=\sqrt{x^2+s^2}$, with $x\in \R$. By the scaled Sobolev inequality in Lemma \ref{sobscale}, we have for $B_R(x_0)\subset \R$
\begin{equation*}
R|U(x_0)|^2\leq C \left[\int_{|y|\leq R}U^2(x_0+y)\,dy+ R^2\int_{|y|\leq R}(\partial_yU)^2(x_0+y)\,dy\right].
\end{equation*}
Now let $R=\frac{1}{2}\sqrt{s^2+x_0^2}=\frac{1}{2}t(x_0)$.
\begin{equation*}
s^2+(y+x_0)^2=s^2+x_0^2+2yx_0+y^2 \geq (2R)^2-4R|y|+y^2=(2R-|y|)^2.
\end{equation*}
By virtue of $|y|<R$ in $B_R(x_0)$, we therefore obtain
\begin{equation*}
\sqrt{s^2+(y+x_0)^2)}\geq 2 R-|y|> R.
\end{equation*}
With the above inequality, (\ref{eq:sobhs1}) gives the statement of the lemma directly.
\end{proof}

\begin{proposition}
\label{prop:decph}
If we restrict to $l\geq 1$ in the $\mu=0$ case, we obtain the following decay of solutions to (\ref{eq:kgequation}) for all $\mu\in \R$, with suitably regular initial data, 
\begin{equation*}
||\psi||^2_{L^{\infty}(\Sigma_{\tau})}\leq C \tau^{-1}\sum_{|k|\leq 1}\left\{E_{PH}[\Omega^k\psi]+E_{PH}[\Omega^kY\psi]\right\}.
\end{equation*}
\end{proposition}
\begin{proof}
We commute (\ref{eq:kgequation}) with $Y=x\partial_t+t\partial_x$, which is tangent to $\mathcal{H}_s$, because on $\mathcal{H}_s$ we have the following equality,
\begin{equation*}
(Y f)^2=\left(x\left(\frac{\partial t}{\partial x}\right)^{-1} +t\right)^2(\partial_x f)^2=4(s^2+x^2)(\partial_x f)^2.
\end{equation*}
When integrating the natural volume form induced on $\mathcal{H}_s$ we have to take into account the factor
\begin{equation*}
\sqrt{\det g|_{\mathcal{H}_S}}=\frac{s}{t}.
\end{equation*}
We can now apply Lemma \ref{sobhs} to $\psi$ and $\Omega_i \psi$ and integrate over $\s^2$ to obtain
\begin{equation*}
\begin{split}
\int_{\s^2} |\snabla \psi|^2+\mu^2\psi^2 \,d\mu_{\s^2} &=\sum_{i=1}^3\int_{\s^2} r_0^{-4}(\Omega_i \psi)^2+\mu^2\psi^2\,d\mu_{\s^2}\\
&\leq Ct^{-1} \sum_{i=1}^3\int_{\mathcal{H}_s} \frac{t}{s}[(\Omega_i\psi)^2+\mu^2\psi^2+(Y \Omega_i \psi)^2+ \mu^2 (Y\psi)^2]\\
& \leq C t^{-1} \int_{\mathcal{H}_s} \frac{t}{s}\left[|\snabla\psi|^2 + \mu^2\psi^2+ |\snabla Y\psi|^2+\mu^2 (Y\psi)^2\right]\\
& \leq C t^{-1} \int_{\mathcal{H}_s} \left\{J^T[\psi]\cdot n_{\mathcal{H}_s}+J^T[Y \psi]\cdot n_{\mathcal{H}_s}\right\},
\end{split}
\end{equation*}
where we have used Proposition \ref{sphsym} multiple times.

Suppose $\mu=0$. We apply Proposition \ref{sobs2} , Proposition \ref{poins2} and Proposition \ref{sphsym} together with the assumption that $\psi$ is supported on the modes $l\geq 1$, to obtain for a fixed $t$ and $x$
\begin{equation}
\label{ineq:sobs}
||\psi||^2_{L^{\infty}(\s^2)} \leq C ||\psi||^2_{H^2(\s^2)}\leq C \sum_{i=1}^3 ||\snabla \Omega_i \psi||^2_{L^2(\s^2)},
\end{equation}
where $C$ depends on $r_0$. Hence, we can estimate in $\mathcal{C}$
\begin{equation*}
\begin{split}
\psi^2(t,x,\theta,\phi)&\leq C t^{-1}  \sum_{|k|\leq 1}\int_{\mathcal{H}_s\times \s^2}\left\{ J^T[\Omega^k \psi]\cdot n_{\mathcal{H}_s}+J^T[Y \Omega^k \psi]\cdot n_{\mathcal{H}_s}\right\}\\
&\leq C t^{-1}  \sum_{|k|\leq 1}\int_{\mathcal{H}_1\times \s^2}\left\{ J^T[\Omega^k \psi]\cdot n_{\mathcal{H}_1}+J^T[Y \Omega^k \psi]\cdot n_{\mathcal{H}_1}\right\}\\
& \leq C \tau^{-1}  \sum_{|k|\leq 1}\int_{\Sigma_0}\left\{ J^T[\Omega^k \psi]\cdot n_{\Sigma_0}+J^T[Y \Omega^k \psi]\cdot n_{\Sigma_0}\right\}\
\end{split}
\end{equation*}
where in the second to last step we used Stokes' Theorem on the region in $\mathcal{M}$ bounded in the future by $\mathcal{H}_1$ and in the past by $\Sigma_0$.

Now suppose $\mu\neq 0$. Because our estimate now includes a zeroth-order term, we do not need to restrict to $l\geq 1$. We use (\ref{ineq:sobs}), together with
\begin{equation*}
|\snabla^2 \psi|^2 \leq r_0^{-8} \sum_{i,j=1}^3 (\Omega_i\Omega_j \psi)^2= r_0^{-4} \sum_{i=1}^3 |\snabla \Omega_i\psi|^2.
\end{equation*}
 to obtain in $\mathcal{C}$
\begin{equation*}
\begin{split}
\psi^2(t,x,\theta,\phi)&\leq C \tau^{-1} \sum_{|k|\leq 1}\int_{\Sigma_0}\left\{ J^T[\Omega^k \psi]\cdot n_{\Sigma_0}+J^T[Y \Omega^k \psi]\cdot n_{\Sigma_0}\right\},
\end{split}
\end{equation*}
where $C$ now also depends on $\mu$. Since the constants in the estimates do not depend on the choice of $\mathcal{C}$, we can conclude the statement of the proposition.
\end{proof}

We have now proved Theorem \ref{th:decph}.

\subsection{Integrated local energy decay}
The argument in the previous section is not very robust, as it fundamentally relies on the presence of a boost vector field $Y$. We complement the uniform boundedness result above by an integrated local energy decay statement, in which $Y$ does not appear. Because of the trapping of null geodesics along each fixed $x$ hypersurface, we will lose derivatives in the estimate.

We consider a vector field $V=f(x)\partial_x$. The corresponding spacetime current $K^V$ is given by
\begin{equation*}
K^V[\psi]=\frac{f'}{2}\left[(\partial_t\psi)^2+(\partial_x\psi)^2-|\snabla\psi|^2-\mu^2\psi^2\right].
\end{equation*}
We need to modify the current $J^V[\psi]$ to control the $|\snabla\psi|^2$ and $\mu\psi^2$ terms in the resulting compatible current. 

Let
\begin{equation*}
J^{V,1}_{\alpha}[\psi]:=J^V_{\alpha}[\psi]+\frac{f'}{2}\psi\partial_{\alpha}\psi-\frac{1}{4}\psi^2\partial_{\alpha}f'.
\end{equation*}
Then
\begin{equation*}
\begin{split}
K^{V,1}[\psi]:=\nabla^{\alpha}J_{\alpha}^{V,1}[\psi]&=K^V[\psi]+\frac{f'}{2}g^{\alpha \beta}\partial_{\alpha}\psi\partial_{\beta}\psi-\frac{1}{4}f'''\psi^2\\
&=f'(\partial_x\psi)^2-\frac{f'''}{4}\psi^2.
\end{split}
\end{equation*}
With the modified current we lose control over $(\partial_t\psi)^2$, but we are able to control $(\partial_x\psi)^2$ and $|\snabla \psi|^2$ if we assume $\psi=\psi_l$, where $l\geq 1$. Indeed, by the Poincar\'e inequality,
\begin{equation*}
\int_{\s^2} \psi_l^2\,d\mu_{\s^2}=\frac{r_0^2}{l(l+1)}\int_{\s^2} |\snabla \psi_l|^2\,d\mu_{\s^2}.
\end{equation*}

\begin{proposition}
There exists a uniform constant $C>0$ such that for all $x\in[-x_0,x_0]$, and $l\geq 1$
\begin{equation*}
\int_{\s^2} K^{V,1}[\psi_l]+K^W[\psi_l]\geq C\int_{\s^2} (\partial_t\psi_l)^2+(\partial_x\psi_l)^2+\frac{1}{l(l+1)}|\snabla\psi_l|^2+\mu^2\psi_l^2,
\end{equation*}
for suitable $V=f(x)\partial_x$ and $W=g(x)\partial_x$. Moreover, if $\mu\neq 0$, we can include $l=0$, by removing the $|\snabla \psi_l|^2$ term in the above estimate. If $\mu=0$, we can estimate directly,
\begin{equation*}
\int_{\s^2} K^W[\psi_0]\,d\mu_{\s^2}\geq C\int_{\s^2}(\partial_t\psi_0)^2+(\partial_x\psi_0)^2\,d\mu_{\s^2},
\end{equation*}
for all $x\in[-x_0,x_0]$.
\end{proposition}
\begin{proof}
We restrict to the region $\{-x_0\leq x \leq x_0\}$, where $x_0$ can be chosen arbitrarily large, and want to prove integrated decay in $\{-x_0\leq x \leq x_0\}$.

Define the function $f$ on $[-x_0,x_0]$ by
\begin{align*}
f(x)&=(x+x_0+E)\log(x+x_0+E)-(x+x_0+E)+F>0,
\end{align*}
where $E,F>0$ are sufficiently large constant, to ensure $f(x)>0$ on $[-x_0,x_0]$. We obtain,
\begin{align*}
f'(x)&=\log(x+x_0+E)>0,\\
f''(x)&=\frac{1}{x+x_0+E}>0,\\
f'''(x)&=-\frac{1}{(x+x_0+E)^2}<0.
\end{align*}
We can now estimate for $l\geq 1$
\begin{equation*}
\int_{\s^2} (\partial_x\psi_l)^2+\frac{1}{l(l+1)}|\snabla\psi_l|^2+\mu^2\psi_l^2 \leq C\int_{\s^2}K^{V,1}[\psi_l].
\end{equation*}
for all $x\in [-x_0,x_0]$. 

Now choose $g(x)=\epsilon (x+x_0+E)$, for $\epsilon>0$ a suitably small constant. Then we also gain control of $(\partial_t\psi_l)^2$, and we absorb the badly signed terms in $K^W$ by $K^{V,1}$.
\end{proof}
We resort to Stokes' Theorem to estimate the spacetime integral of $K^{V,1}$, using that $X=\partial_x$ is a Killing vector field.
\begin{proposition}
\label{prop:iledph}
There exists a a uniform constant $C>0$ such that,
\begin{equation*}
\begin{split}
\int_{0}^{\infty} \left( \int_{\Sigma_{\tau}\cap\{|x|\leq x_0\}} J^T[\psi]\cdot n_{\tilde{\Sigma}_{\tau}}\right)\,d\tau&\leq C \sum_{|l|\leq 1} E_{PH}[\Omega^l\psi].
\end{split}
\end{equation*}
\end{proposition}

\begin{proof}
By Stokes' Theorem, we have for $l\geq 1$
\begin{equation*}
\begin{split}
\int_{\{|x| \leq x_0\}}K^{V,1}[\psi]&=\int_{\Sigma_{\tau_1}\cap\{|x| \leq x_0\}}J^{V,1}[\psi]\cdot n_{\Sigma_{\tau_1}}-\int_{\Sigma_{\tau_2}\cap\{|x| \leq x_0\}}J^{V,1}[\psi]\cdot n_{\Sigma_{\tau_2}}\\
&+\int_{\{x=x_0\}\cap\{0\leq \bar{\tau} \leq \tau\}} J^{V,1}[\psi]\cdot \frac{\partial}{\partial x}\\
&-\int_{\{x=-x_0\}\cap \{0\leq \bar{\tau} \leq \tau\}} J^{V,1}[\psi]\cdot\frac{\partial}{\partial x}.
\end{split}
\end{equation*}
It easily follows from the Poincar\'e inequality on the sphere that
\begin{equation*}
|J^{V,1}[\psi_l]\cdot n_{\Sigma_{\tau}}|\leq C J^T[\psi_l]\cdot n_{\Sigma_{\tau_1}}
\end{equation*}
in $\{-x_0\leq x \leq x_0\}$ for $l\geq 1$. The above estimate holds for $l\geq 0$ if $\mu\neq 0$. We can also estimate
\begin{equation*}
|J^{W}[\psi]\cdot n_{\Sigma_{\tau}}|\leq C J^T[\psi_l]\cdot n_{\Sigma_{\tau_1}}.
\end{equation*}

Moreover,
\begin{equation*}
\begin{split}
\int_{\s^2}J^{V,1}[\psi]\cdot \frac{\partial}{\partial x}\Big|_{x=x_0}\,d\mu_{\s^2}&\leq \int_{\s^2} \Bigg\{\frac{f}{2}\left[(\partial_x\psi)^2+(\partial_t\psi)^2-|\snabla\psi|^2-\mu^2\psi^2\right]\\
&+\frac{ff'}{2}\psi\partial_x\psi-\frac{r_0^2}{4 l(l+1)}ff''|\snabla\psi|^2\Bigg\}\Big|_{x=x_0}\,d\mu_{\s^2}\\
&\leq \int_{\s^2} \Bigg\{ \frac{f}{2}\left[\left(1+\frac{1}{\epsilon}\frac{f'^2}{2}\right)(\partial_x\psi)^2+(\partial_t\psi)^2\right]\\
&-\frac{f}{2}\Bigg[1+\frac{r_0^2}{4l(l+1)}f''\\
&-\epsilon\frac{r_0^2}{l(l+1)} \frac{f'^2}{2}\Bigg]|\snabla\psi|^2-\frac{f}{2}\mu^2\psi^2\Bigg\}\Big|_{x=x_0}\,d\mu_{\s^2}\\
&\leq C \int_{\s^2} J^{\frac{\partial}{\partial x}}[\psi]\cdot \frac{\partial}{\partial x}\Big|_{x=x_0}\,d\mu_{\s^2}.
\end{split}
\end{equation*}
where in the second inequality we have used Young's inequality with $\epsilon$ on $\psi\partial_x\psi$. We need to take $\epsilon>0$ suitably small. Similarly,
\begin{equation*}
\begin{split}
-\int_{\s^2}J^{V,1}[\psi]\cdot \frac{\partial}{\partial x}\Big|_{x=-x_0}\,d\mu_{\s^2}&\leq \int_{\s^2}\Bigg\{ -\frac{f}{2}\left[1-\epsilon\frac{f'^2}{2}\right](\partial_x\psi)^2-\frac{f}{2}(\partial_t \psi)^2\\
&+\frac{f}{2}\Bigg[1+\frac{r_0^2}{4l(l+1)}f''+\frac{1}{\epsilon}\frac{r_0^2}{l(l+1)} \frac{f'^2}{2}\Bigg]|\snabla\psi|^2\Bigg\}\,d\mu_{\s^2}\\
&\leq -C \int_{\s^2} J^{\frac{\partial}{\partial x}}[\psi]\cdot \frac{\partial}{\partial x}\Big|_{x=-x_0}\,d\mu_{\s^2}.
\end{split}
\end{equation*}

Since $X=\frac{\partial}{\partial x}$ is Killing, we have that $K^{X}[\psi]=0$. We can therefore apply Stokes' Theorem in the regions $\{|x|\geq x_0\}$ to estimate the energy flux through $\{x=\pm x_0\}$ of the $J^X$ current .

\begin{equation*}
\begin{split}
\pm\int_{\{x=\pm x_0\, 0\leq \bar{\tau}\leq \tau\}} J^X\cdot \partial_x &\leq \int_{\Sigma\cap \{|x|\geq x_0\}} J^X\cdot n_{\Sigma_{\tau_1}}+\int_{\{v=\tau+x_0,\,-x_0\leq u \leq \tau-x_0\}} J^X\cdot \partial_u\\
&+\int_{\{u=\tau+x_0,\,-x_0\leq v \leq \tau-x_0\}} J^X\cdot \partial_v.
\end{split}
\end{equation*}

We have that
\begin{equation*}
\begin{split}
\int_{\{v=\tau+x_0,\,-x_0\leq u \leq \tau-x_0\}} J^X\cdot \partial_u&=\int_{\{v=\tau+x_0,\,-x_0\leq u \leq \tau-x_0\}}\left \{ -(\partial_u\psi)^2+\frac{1}{4}\left(|\snabla\psi|^2+\mu^2\psi^2\right) \right\} \\
&\leq \int_{\{v=\tau+x_0,\,-x_0\leq u \leq \tau-x_0\}} J^T\cdot \partial_u\leq C \int_{\Sigma} J^T\cdot n_{\Sigma}
\end{split}
\end{equation*}
and similarly
\begin{equation*}
\int_{\{u=\tau+x_0,\,-x_0\leq v \leq \tau-x_0\}} J^X\cdot \partial_v \leq  \int_{\{u=\tau+x_0,\,-x_0\leq v \leq \tau-x_0\}} J^T\cdot \partial_v\leq C \int_{\Sigma} J^T\cdot n_{\Sigma}.
\end{equation*}

We conclude that
\begin{equation*}
\pm\int_{\{x=\pm x_0\, 0\leq \bar{\tau}\leq \tau\}} J^{V,1}\cdot \partial_x \leq C \int_{\Sigma} J^T[\psi]\cdot n_{\Sigma}.
\end{equation*}
By a similar argument we obtain the above estimate with $J^W$ replacing $J^{V,1}$.

Consequently, we obtain an integrated local energy decay statement for $\psi=\psi_l$, with $l\geq 1$.
\begin{equation}
\label{iledpha}
\int_{0}^{\infty} \left( \int_{\Sigma_{\tau}\cap\{|x|\leq x_0\}} J^T[\psi_l]\cdot n_{\Sigma_{\tau}}\right)\,d\tau\leq C l(l+1) \int_{\Sigma} J^T[\psi_l]\cdot n_{\Sigma},
\end{equation}
where $C>0$ is independent of $l$. If $l=0$, we can drop the $l(l+1)$ factor.

We can use (\ref{iledpha}) and (\ref{eq:poincares2fixedl}) to obtain
\begin{equation}
\begin{split}
\label{iledph}
\int_{\tau_1}^{\tau_2} \left( \int_{\Sigma_{\tau}\cap\{|x|\leq x_0\}} J^T[\psi]\cdot n_{\tilde{\Sigma}_{\tau}}\right)\,d\tau\leq C \sum_{|k|\leq 1} \int_{\Sigma} J^T[\Omega^k \psi]\cdot n_{\Sigma}. \qedhere
\end{split}
\end{equation}
\end{proof}
 
We have now proved Theorem \ref{th:iledph}.

\section{Non-decay for $\mu=0$}
\label{sec:nondec}
We can now easily show that $\psi$ does not decay in time in the case $\mu=0$. Consider first the region $\mathcal{R}$ of Nariai in double-null coordinates. By (\ref{dalm1}) we have for all $(u,v)\in \mathcal{R}$,
\begin{equation*}
\psi_0(u,v)=\psi_0(u_{\Sigma_0}(v),v)+\int_{u_{\Sigma_0}(v)}^{u}\partial_u\psi_0(\bar{u},v_{\Sigma_0}(\bar{u}))\,d\bar{u}.
\end{equation*}
For generic initial data $\Psi$, $\Psi'$, the above expression converges to a constant $\bar{\psi}$ on $\Sigma_{\tau}$ as $\tau\to \infty$. Theorem \ref{th:nondecn} immediately follows.

In PH, we treat the two null boundaries separately. On the null boundary $\mathcal{N}_B^+$ we can write
\begin{equation*}
\psi_0(\infty,v)=\psi_0(u_{\Sigma}(v),v)+\int_{u_{\Sigma}(v)}^{\infty} \partial_u\psi_0(\bar{u},v_{\Sigma}(\bar{u}))\,d\bar{u}.
\end{equation*} 
We can derive an expression analogous to (\ref{dalm1}), reversing the roles of $u$ and $v$, to obtain on the null boundary $\mathcal{N}_A^+$
\begin{equation*}
\psi_0(u,\infty)=\psi_0(u,v_{\Sigma}(u))+\int_{v_{\Sigma}(u)}^{\infty} \partial_v\psi_0(u_{\Sigma}(\bar{v}),\bar{v})\,d\bar{v}.
\end{equation*}
In both cases, the expressions do not vanish for generic initial data. Theorem \ref{th:nondecph} now follows.

\appendix
\section{Sobolev estimates}
\label{app:sobolev}
\begin{lemma}
\label{lmgen}
Let $u\in H^s(\R^n)$. Then for $s\in \N$, $s>\frac{n}{2}$
\begin{equation*}
||u||_{L^{\infty}(\R^n)}\leq C ||u||_{H^s(\R^n)},
\end{equation*}
where $C$ depends only on $n$ and $s$.
\end{lemma} 

\begin{lemma}
\label{lmbump}
Let $u\in H^s(B)$, where $B=I_1\times I_2\times I_3\subset \R^3$, $I_i$ intervals in $\R$. Then
\begin{equation*}
||u||_{L^{\infty}(B)}\leq C ||u||_{H^2(B)},
\end{equation*}
where $C$ depends only on $B$.
\end{lemma}
\begin{proof}
We want to construct $\bar{u}\in H^2(\R^3)$, such that $\bar{u}(x)=u(x)$ for $x\in B$ and $\bar{u}(x)=0$ for $x \notin U$, where $U\subset \R^3$ is an open containing $B$. Consider the bump function $\eta: \R \to \R$:
\begin{align*}
\eta(x)&:=A e^{-\frac{1}{1-x^2}} &\textnormal{for}\, |x|\leq 1,\\
\eta(x)&:=0 & \textnormal{for}\,|x|>1,
\end{align*}
where $A$ is a constant such that $\int_{\R^n} \eta(x)\,dx=1$. For each $\epsilon>0$ we can define $\eta_{\epsilon}(x):=\epsilon^{-n}\eta(\epsilon^{-1} x)$.  Then $\int_{\R^n} \eta_{\epsilon}(x)\,dx=1$ and $\textnormal{spt}(\eta_{\epsilon})\subset [-\epsilon,\epsilon]$. One can show that $\eta_{\epsilon} \in C^{\infty}(\R)$ for all $\epsilon>0$. Consider the step functions $\chi_i: \R \to \R$, $\chi(x_i):=1$ if $x_i\in I_i$ and $\chi=0$ if $x_i\notin I_i$. $\chi_i$ are locally integrable functions, but are not smooth. However, the mollifications $\chi_i^{\epsilon}:=\eta_{\epsilon}\ast \chi_i$ are smooth and satisfy
\begin{align*}
\chi_i^{\epsilon}(x)&=1  &\textnormal{for}\, x\in [a_i,b_i]\\
\chi_i^{\epsilon}(x)&=0 &\textnormal{for}\, x\notin [a_i-\epsilon,b_i+\epsilon],
\end{align*}
where $I_i=[a_i,b_i]$, $a_i,b_i\in \R \cup \{-\infty, \infty\}$.

Extend $u$ continuously to $\R^3$, such that $|u(x)|\leq \sup_{y\in B} |u(y)|$. Now define $\bar{u}\in H^2(\R^3)$ by $\bar{u}(x):=u(x)v(x)$, where $v(x):=\chi^\epsilon_1(x_1)\chi^\epsilon_2(x_2)\chi^\epsilon_3(x_3)$. let $\tilde{B}:= [a_1-\epsilon,b_1+\epsilon]\times [a_2-\epsilon,b_2+\epsilon]\times  [a_3-\epsilon,b_3+\epsilon]$. We have that $\bar{u}(x)=u(x)$ for $x\in B$ and $\bar{u}(x)=0$ for $x\notin \tilde{B}$. Furthermore, $|\bar{u}(x)|\leq \sup_{y\in B} |u(y)|$.

By Lemma \ref{lmgen} it follows that
\begin{equation*}
||u||_{L^{\infty}(B)}^2=||\bar{u}||^2_{L^{\infty}(\R^3)}\leq K ||\bar{u}||^2_{H^2(\R^3)},
\end{equation*}
with $K>0$ a constant independent of $u$. We have that
\begin{equation*}
\begin{split}
||\bar{u}||^2_{H^2(\R^3)}&= ||u||^2_{H^2(B)}+\int_{\tilde{B}\setminus B} \left[ (uv)^2+ \sum_{i=1}^3 (\partial_i(uv))^2+ \sum_{i,j=1}^3 (\partial_i \partial_j (uv))^2\right]\,dx\\
&=||u||^2_{H^2(B)}+\int_{\tilde{B}\setminus B} \Big[ \sum_{i=1}^3 (v\partial_i u+(\partial_i v) u)^2\\
&+ \sum_{i,j=1}^3 \left( v\partial_i\partial_ju+2\partial_i v \partial_j u+ (\partial_i \partial_j v) u  \right)^2\Big]\,dx\\
&\leq ||u||^2_{H^2(B)}+ \tilde{C}(\epsilon)||u||^2_{H^2(\tilde{B}\setminus B)}\\
& \leq C ||u||^2_{H^2(B)},
\end{split}
\end{equation*}
where we used that $v$ and its derivatives only depend on $\epsilon$, and $\tilde{C}(\epsilon)\to 0$ as $\epsilon \to 0$. We can therefore choose $C(\epsilon)>0$ sufficiently small, such that $||u||^2_{H^2(\tilde{B}\setminus B)}<||u||^2_{H^2(B)}$. Consequently,
\begin{equation*}
||u||_{L^{\infty}(B)}^2\leq K C ||u||^2_{H^2(B)}. \qedhere
\end{equation*}
\end{proof}

\begin{corollary}
\label{sobscale}
Let $u\in H^s(B_R(x_0))$, where $s>\frac{n}{2}$ and $B_R(x_0)=\{y\in \R^n\::\: |y-x_0|<R\}$. Then
\begin{equation*}
||u||_{L^{\infty}(B_R(x_0))}\leq C \sum_{k=0}^s R^{k-\frac{n}{2}}||\partial^k u||_{L^2(B_R(x_0))},
\end{equation*}
where $C$ depends only on $B_1(x_0)$.
\end{corollary}
\begin{proof}
By using a suitable cut-off function analogous to the function $v$ in Lemma \ref{lmbump}, we infer from Lemma \ref{lmgen} that
\begin{equation*}
||u||_{L^{\infty}(B_1(x_0))}\leq C \sum_{k=0}^s ||\partial^k u||_{L^2(B_1(x_0))},
\end{equation*}
for $u\in H^s(B_1(x_0))$. Define $\tilde{u}$ by $\tilde{u}(x_0+y):=u(x_0+Ry)$, then
\begin{equation*}
\begin{split}
\int_{|y-x_0|<1} (\partial^k \tilde{u})^2(x_0+y)\,dy&=R^{2k} \int_{|y-x_0|<1} (\partial^k u)^2(x_0+Ry)\,dy\\
&=R^{2k-n}\int_{|z-x_0|<R}(\partial^k u)^2(x_0+z)\,dz,
\end{split}
\end{equation*}
where $z=Ry$. Hence, $u\in H^s(B_1(x_0))$ if and only if $\tilde{u}\in H^s(B_R(x_0))$. Together with $||u||_{L^{\infty}(B_R(x_0))}=||\tilde{u}||_{L^{\infty}(B_1(x_0))}$, this proves the corollary.
\end{proof}

\begin{proposition}
\label{sobs2}
Let $u\in H^2(I\times \s^2)$, where $I\subseteq \R$. Then there exists a $C>0$ independent of $u$, such that
\begin{equation*}
||u||_{L^{\infty}(I\times \s^2)}\leq C ||u||_{H^2(I\times \s^2)}.
\end{equation*}
\end{proposition}
\begin{proof}
The chart $(\theta, \phi)$ covers the entire sphere, without a meridian connecting the poles. The metric components of $\gamma$ are bounded from below and away from zero, if we restrict $\theta\in [\eta, \pi-\eta]$ for $\eta>0$ small. If we can bound the integral over the region $\mathcal{S}=\{\eta\leq \theta \leq \pi-\eta\}$ independently of $u$, then we can redefine $(\theta,\phi)\mapsto (\tilde{\theta},\tilde{\phi})$, so as to cover the remaining region $\s^2\setminus \mathcal{S}$ by $\eta \leq \tilde{\theta} \leq \pi-\eta$ and use the same bound as before. We find that
\begin{equation*}
\begin{split}
||u||^2_{H^2(I\times \mathcal{S})}&=\int_I \int_0^{2\pi}\int_{\eta}^{\pi-\eta} \Big( u^2+|\mathring{\snabla} u|^2+(\partial_x u)^2+|\mathring{\snabla^2} u|^2\\
&+|\mathring{\snabla} \partial_x u|^2+(\partial_x^2 u)^2 \Big)\sin \theta\,d\theta d\phi dx,
\end{split}
\end{equation*}
where
\begin{align*}
|\mathring{\snabla}u|^2&=\slashed{g}_{\s^2}^{AB}\nabla_Au\nabla_B u=(\partial_{\theta}u)^2+\sin^{-2}\theta(\partial_{\phi} u)^2,\\
|\mathring{\snabla^2}u|^2&=\slashed{g}_{\s^2}^{AB}\slashed{g}_{\s^2}^{CD} \nabla_A \nabla_Cu \nabla_B \nabla_D u=\slashed{g}_{\s^2}^{AB}\slashed{g}_{\s^2}^{CD}(\partial_A\partial_Cu-\Gamma^E_{AC}\partial_E u)\\
&\cdot(\partial_B\partial_Du-\Gamma^E_{BD}\partial_E u)\\
&=(\partial_\theta^2u)^2+2\sin^{-1}\theta \left(\partial_{\theta} \partial_{\phi}u+\cot \theta \partial_{\phi} u\right)^2+\sin^{-2} \theta\left(\partial_{\phi}^2u+\sin\theta\cos \theta \partial_{\theta} u \right)^2,
\end{align*}
where we used that the only non-zero components of the Christoffel symbols on $\s^2$ are $\Gamma^{\theta}_{\theta \phi}=-\cot \theta$ and $\Gamma^{\theta}_{\phi \phi}=-\sin \theta \cos \theta$. We can use Cauchy's inequality with $\epsilon$ to absorb the mixed terms in the squares into the remaining terms in $|\mathring{\snabla}u|^2+|\mathring{\snabla^2}u|^2$ and estimate,
\begin{equation*}
C ||u||^2_{H^2(I\times \s^2)} \geq ||u||^2_{H^2(I\times \mathcal{S})}.
\end{equation*}
We conclude that
\begin{equation*}
||u||^2_{L^{\infty}(I\times \s^2)}\leq \tilde{C} ||u||^2_{H^2(I \times [0,2\pi]\times [0,\pi])} \leq \tilde{C}||u||^2_{H^2(I\times \mathcal{S})} \leq C ||u||^2_{H^2(I\times \s^2)}.
\end{equation*}
where we applied Corollary \ref{lmbump} in the first inequality.
\end{proof}

\section{Estimates on $\s^2$}
\begin{proposition}[Poincar\'e inequality on $\s^2$]
\label{poins2}
Let $\psi \in \mathring{H}^1(\s^2)$, such that $\psi_l=0$ for $l\geq L$, $L\geq1$. Then
\begin{equation*}
||\psi||_{L^2(\s^2)}^2\leq \frac{1}{L(L+1)}||\psi||^2_{\mathring{H}^1(\s^2)}.
\end{equation*}
\end{proposition}
\begin{proof}
We can perform the following integration by parts
\begin{equation*}
\begin{split}
\int_{\s^2}|\mathring{\snabla}\psi|^2&= - \int_{\s^2} \psi \mathring{\slashed{\Delta}}\psi=\int_{\s^2} \sum_{l,l'=0}^{\infty} \sum_{m=-l}^l \sum_{m'=-l'}^{l'} l(l+1)\psi_{m,l}\overline{\psi_{m',l'}}Y^{m,l}\overline{Y^{m',l'}}\\
&\geq L(L+1) \int_{\s^2} |\psi|^2. \qedhere
\end{split}
\end{equation*}
\end{proof}

\begin{proposition}
\label{sphsym}
Let $\Omega_i$, $i=1,2,3$ be the elements of the Lie algebra generating to $SO(3)$ isometries on $\R^3$. Then
\begin{itemize}
\item[(i)] $r^2|\mathring{\snabla}\psi|^2=\sum_{i=1}^3 (\Omega_i \psi)^2$,
\item[(ii)] $r^4|\mathring{\snabla^2}\psi|^2\leq \sum_{i,j=1}^3 |\Omega_i \Omega_j\psi|^2$,
\end{itemize}
\end{proposition}
where $r^2=|x|^2$, $x_i$ standard coordinates on $\R^3$.
\begin{proof}
Assuming for convenience of notation Einstein summation, we can write $\Omega_i=\Omega^k_i \frac{\partial}{\partial x^k}$
\begin{equation*}
\Omega_i^k=\epsilon_{ijk}x_j,
\end{equation*}
which implies that
\begin{equation*}
\Omega^k_i \Omega^l_i=\epsilon_{ijk}\epsilon_{iml}x_jx_m=(\delta_{jm}\delta_{kl}-\delta_{jl}\delta_{km})x_j x_m=r^2\delta_{kl}-x_kx_l.
\end{equation*}
Now we can conclude that
\begin{equation*}
\begin{split}
\sum_{i=1}^3 (\Omega_i \psi)^2=r^2|\nabla \psi|^2-(x\cdot\nabla\psi)^2=r^2\left(|\nabla \psi|^2-\left(\frac{x}{|x|}\cdot \nabla\psi\right)^2\right)=r^2|\mathring{\snabla}\psi|^2,
\end{split}
\end{equation*}
which proves (i). Using that $\Omega_i$ are Killing vector fields and the expressions above, we can obtain the inequality
\begin{equation*}
\begin{split}
r^2|\mathring{\snabla^2}\psi|^2 &\leq \sum_{i=1}^3 |\mathcal{L}_{\Omega_i}d \psi|^2=\sum_{i=1}^3 |d(\mathcal{L}_{\Omega_i} \psi)|^2\\
&=\sum_{i=1}^3 |\mathring{\snabla}(\Omega_i\psi)|^2=r^{-2}\sum_{i,j=1}^3 |\Omega_i \Omega_j\psi|^2,
\end{split}
\end{equation*}
which proves (ii).
\end{proof}

\section{The Einstein equations in spherical symmetry}
\label{app:reform}
We introduce $\hat{T}$, defined as $$R_{\mu\nu}-\frac{1}{2}Rg_{\mu\nu}= 8\pi {T}_{\mu \nu}-\Lambda g_{\mu \nu}=:2\hat{T}_{\mu\nu}$$.

Assume the splitting of the metric
\begin{equation*}
g=g_{ab}(x)dx^adx^b+r^2(x){\slashed{g}_{\s^2}}_{AB}(y)dy^Ady^B,
\end{equation*}
where $\{x^a\}$, $a=0,1$ are coordinates on the quotient manifold $\scrQ=\scrM/SO(3)$ and, $\{y^A\}$, $A=2,3$ are coordinates on the orbits, which are 2-spheres equipped with the standard metric ${\slashed{g}_{\s^2}}$. Relevant quantities which get affected by this splitting are the Christoffel symbols
\begin{equation*}
\Gamma^{\mu}_{\nu \rho}=\frac{1}{2}g^{\mu \sigma}(g_{\sigma \nu,\rho}+g_{\sigma \rho,\nu}-g_{\nu \rho,\sigma}).
\end{equation*}
The quantities $\Gamma^{a}_{bc}$ and $\Gamma^{A}_{BC}$ are equal to the two-dimensional Christoffel symbols of the respectively $\scrQ$ and a sphere of radius $r(x)$, $x\in \mathcal{M}$. There are also mixed terms
\begin{align}
\label{eq:cs1}
\Gamma^a_{bC}&=0,\\
\Gamma^a_{BC}&=-r \partial_b r {\slashed{g}_{\s^2}}_{BC}g^{ab},\\
\Gamma^A_{bc}&=0,\\
\label{eq:cs4}
\Gamma^A_{Bc}&=r^{-1}\partial_c r \delta^A_B.
\end{align}
The Riemann tensor can be expressed in terms of the Christoffel symbols
\begin{equation*}
R^{\mu}_{\:\:\nu \rho \sigma}=\partial_{\rho} \Gamma^{\mu}_{\nu \sigma}-\partial_{\sigma}\Gamma^{\mu}_{\nu \rho}+\Gamma^{\tau}_{\nu \sigma}\Gamma^{\mu}_{\tau \rho}-\Gamma^{\tau}_{\nu \rho}\Gamma^{\mu}_{\tau \sigma}.
\end{equation*}
We can obtain the Ricci tensor by contraction: $R_{\nu \sigma}=R^{\mu}_{\:\:\nu \mu \sigma}$. It's convenient to split
\begin{equation*}
R_{ab}=\bar{R}_{ab}+R'_{ab},
\end{equation*}
where $\bar{R}_{ab}$ contains only Christoffel symbols with lower case indices and $R'_{ab}$ constitutes the remaining terms. The symmetries of the Riemann tensor imply that in two dimensions there is only one independent component, so we can write
\begin{equation}
\label{eq:rt}
\bar{R}_{abcd}=K(g_{ab} g_{cd}-g_{ab}g_{cd}),
\end{equation}
where $K$ is the Gaussian curvature of $(\mathcal{Q},\bar{g})$.

Contraction of (\ref{eq:rt}) implies that $\bar{R}_{ab}=K g_{ab}$. Similarly, we can split $R_{AB}={\slashed{g}_{\s^2}}_{AB}+R'_{AB}$, where we used that the Gaussian curvature on $\mathbb{S}^2$ is 1 and the Christoffel symbols are invariant under a rescaling of the metric by a constant factor. We can therefore write
\begin{equation}
\label{eq:rc1}
\begin{split}
R_{ab}&=\partial_{\mu} \Gamma^{\mu}_{ab}-\partial_b \Gamma^{\mu}_{a \mu}+\Gamma^{\tau}_{ab}\Gamma^{\mu}_{\tau \mu}-\Gamma^{\tau}_{a \mu}\Gamma^{\mu}_{\tau b}\\
&=\bar{R_{ab}}-\partial_b \Gamma^A_{aA}+\Gamma^d_{ab} \Gamma^A_{dA}-\Gamma^A_{aB}\Gamma^B_{Ab}\\
&=Kg_{ab}-\nabla_b(2r^{-1} \partial_a r)-2r^{-2}\partial_a r \partial_b r\\
&=Kg_{ab}-2r^{-1}\nabla_a \nabla_b r,
\end{split}
\end{equation}
where in the second equality we used the expressions (\ref{eq:cs1})-(\ref{eq:cs4}). Similarly,
\begin{equation}
\label{eq:rc2}
\begin{split}
R_{aA}&=\partial_{\mu} \Gamma^{\mu}_{aA}-\partial_A \Gamma^{\mu}_{a \mu}+\Gamma^{\tau}_{aA}\Gamma^{\mu}_{\tau \mu}-\Gamma^{\tau}_{a \mu}\Gamma^{\mu}_{\tau A}\\
&=\Gamma^B_{aA}\Gamma^{C}_{BC}-\Gamma^B_{aC}\Gamma^C_{BA}\\
&=0.
\end{split}
\end{equation}
Finally,
\begin{equation}
\label{eq:rc3}
\begin{split}
R_{AB}&={\slashed{g}_{\s^2}}_{AB}+\partial_{a}\Gamma^{a}_{AB}+\Gamma_{AB}^b\Gamma_{bd}^d+\Gamma_{AB}^b\Gamma_{bC}^C-\Gamma_{Ab}^C\Gamma^b_{CB}-\Gamma_{AC}^d\Gamma_{dB}^C\\
&={\slashed{g}_{\s^2}}_{AB}(1-\nabla^a(r\partial_a r)).
\end{split}
\end{equation}
The form of the Ricci tensor carries over to the energy-momentum tensor via Einstein's equation. We can write
\begin{equation}
\label{eq:t1}
\hat{T}_{\mu \nu}=\frac{1}{2} \left( R_{\mu \nu}-\frac{1}{2} R g_{\mu \nu}\right).
\end{equation}
By inserting the expressions (\ref{eq:rc1})-(\ref{eq:rc3}) into (\ref{eq:t1}), we obtain
\begin{align}
\hat{T}_{Ab}&=0,\\
\hat{T}_{AB}&=\frac{1}{2}\left(1-\nabla^a(r\partial_a r)-\frac{1}{2}R\right) \gamma_{AB}=:r^2(x)S(x)\gamma_{AB}.
\end{align}
Now we can rewrite Einstein's equations in the following form
\begin{align}
\label{eq:es1}
Kg_{ab}-2r^{-1}\nabla_a \nabla_br&=R_{ab}=2(\hat{T}_{ab}-\frac{1}{2}g_{ab}(\textnormal{tr}\,\hat{T}+2S))\\
\label{eq:es2}
(1-\nabla^a(r\partial_ar)){\slashed{g}_{\s^2}}_{AB}&=R_{AB}=-r^2 \textnormal{tr}\, \hat{T}{\slashed{g}_{\s^2}}_{AB}.
\end{align}
Taking the trace of (\ref{eq:es1}) and inserting (\ref{eq:es2}) gives
\begin{equation}
\label{eq:es3}
K=\frac{1}{r^2}(1-\partial^a r \partial_a r)+\textnormal{tr}\, \hat{T}-2S.
\end{equation}
Using (\ref{eq:es3}), we can rearrange (\ref{eq:es2}):
\begin{equation}
\label{eq:es4}
\nabla_a \nabla_b r=\frac{1}{2r} (1-\partial^cr \partial_c r)g_{ab}-r(\hat{T}_{ab}-g_{ab}\textnormal{tr}\, \hat{T}).
\end{equation}
Equations (\ref{eq:es3}) and (\ref{eq:es4}) are therefore equivalent to Einstein's equations in the case of spherical symmetry.

A crucial result of spherical symmetry is that $\scrQ$ is a two dimensional manifold. This implies that ingoing and outgoing null lines span the tangent space. Hence, we can write
\begin{equation*}
\bar{g}=-\Omega^2(u,v)dudv.
\end{equation*}
This allows us to get simple expressions for the Christoffel symbols.
\begin{align*}
\Gamma^u_{vv}&=\Gamma^u_{uv}=\Gamma^v_{uv}=\Gamma^v_{uu}=0,\\
\Gamma^u_{uu}&=\Omega^{-2}\partial_u(\Omega^2),\\
\Gamma^v_{vv}&=\Omega^{-2}\partial_v(\Omega^2).
\end{align*}
In null coordinates $(u,v)$ we can write the equation for the Gaussian curvature as

$$-\frac{1}{2}K \Omega^2=\bar{R}_{uv}=-\partial_v \bar{\Gamma}^u_{uu}=-\partial_v\left(\Omega^{-2} \partial_u \Omega^2\right),$$
where the bar indicates that were working with the induced metric $g_{ab}$ on $\scrQ$. We can rewrite the above expression to obtain:

$$K=2\Omega^{-2}\partial_v(\Omega^{-2}\partial_u \Omega^2)=2 \Omega^{-2}\partial_u\partial_v \log \Omega^2,$$
Hence, (\ref{eq:es3}) can be rewritten to obtain
\begin{equation}
\label{eenull1}
\partial_u\partial_v \log \Omega^2=\frac{\Omega^2}{2r^2}(1-\partial^a r \partial_a r)+\frac{\Omega^2}{2} \tr \hat{T}-\Omega^2 \hat{S}.
\end{equation}
The Einstein-Maxwell stress-energy tensor in null coordinates is given by

$$T=\frac{e^2\Omega^2}{8\pi r^4}dudv+ \frac{e^2}{8 \pi r^2}\slashed{g}_{\s^2},$$
 where $e$ is the total charge (which is topological, because there is no source term in the Maxwell equations). Consequently, 
 
 $$\hat{T}=\left[\frac{e^2}{2 r^4}+\frac{{\Lambda}}{2}\right]\Omega^2 dudv+\left[\frac{e^2}{2r^4}-\frac{\Lambda}{2}\right]r^2 \slashed{g}_{\s^2}.$$ Now we see that
 
$$\frac{1}{2}\tr \hat{T}-\hat{S}=-\left(\frac{e^2}{2r^4}+\frac{\Lambda}{2}\right)-\left(\frac{ e^2}{2r^4}-\frac{\Lambda}{2}\right)=-\frac{e^2}{r^4}.$$

Consequently, equation (\ref{eenull1}) becomes
\begin{equation}
\label{eenull1b}
\partial_u\partial_v \log \Omega^2=\frac{2}{r^2}\partial_u r \partial_v r+\frac{\Omega^2}{2r^2}-\frac{e^2}{r^4}\Omega^2.
\end{equation}
We can also rewrite (\ref{eq:es4}) in null coordinates. We have that

\begin{align*}
\nabla_u \nabla_v r&= \partial_u \partial_v r,\\
0&=\nabla_u \nabla_u r=\partial_u\partial_u r-\Gamma^u_{uu} \partial_u r,\\
0&=\nabla_v \nabla_v r=\partial_v\partial_v r-\Gamma^v_{vv} \partial_u r.
\end{align*}
Consequently,
\begin{equation*}
\Omega^{-2} \partial_u \partial_u r-(\Omega^{-4}\partial_u \Omega^2) \partial_u r=0,
\end{equation*}
and similarly for derivatives with respect to $v$. The above expression gives us \emph{Raychaudhuri's equations}:
\begin{align}
\label{ray1}
\partial_u(\Omega^{-2} \partial_u r)&=0,\\
\label{ray2}
\partial_v(\Omega^{-2} \partial_v r)&=0.
\end{align}

Furthermore,
\begin{equation*}
\begin{split}
\partial_u \partial_v r&=\frac{1}{2r} (1-\partial^cr \partial_c r)g_{uv}-\frac{r \Omega^2}{2}\left( \frac{e^2}{2r^4}+\frac{\Lambda}{2}\right)+\frac{r \Omega^2}{2} \left( \frac{e^2}{r^4}+\Lambda \right)\\
&= -\frac{\Omega^2}{4r}-\frac{1}{r} \partial_u r \partial_v r+\frac{e^2 \Omega^2}{4 r^3}+ \frac{\Lambda}{4} r \Omega^2.
\end{split}
\end{equation*}
Rewriting the above expression, we arrive at the final equation
\begin{equation}
\label{eefin}
\partial_u \partial_v r= -\frac{1}{r} \partial_u r \partial_v r+ \left( e^2-r^2+\Lambda r^4 \right)\frac{\Omega^2}{4r^3}.
\end{equation}
The equations (\ref{eenull1b}), (\ref{ray1}), (\ref{ray2}) and (\ref{eefin}) are Einstein's equations in null coordinates for the spherically symmetric Einstein-Maxwell system with positive cosmological constant.

A useful quantity to consider is the Hawking mass, defined by
\begin{equation*}
m(r):=r\left(\frac{1}{2}+2\Omega^{-2}\partial_ur\partial_vr\right).
\end{equation*}
By filling in the equations (\ref{ray1}), (\ref{ray2})  and (\ref{eefin}) it follows that
\begin{align*}
\partial_u m= \frac{\partial_u r}{2r^2}(e^2+\Lambda r^4),\\
\partial_v m= \frac{\partial_v r}{2r^2}(e^2+\Lambda r^4).
\end{align*}
Consequently, the \emph{renormalised mass} $\varpi:=m+\frac{e^2}{2r}-\frac{\Lambda}{6}r^3$ satisfies
\begin{equation}
\label{eq:varpi}
\partial_u\varpi=\partial_v\varpi=0.
\end{equation}
We can therefore take $\varpi=M$, where $M\in \R$ is constant.
\bibliography{mybib.bib}
\bibliographystyle{amsplain}

\end{document}